\documentclass[11pt]{article}
\usepackage[english]{babel}
\usepackage[letterpaper,top=2cm,bottom=2cm,left=2cm,right=2cm,marginparwidth=1.75cm]{geometry}
\usepackage{amsmath}
\usepackage[none]{hyphenat}
\usepackage{amssymb}
\usepackage[numbers,sort&compress]{natbib}
\usepackage{enumitem}
\usepackage{amsthm}

\usepackage{float}
\usepackage{braket}
\usepackage{pdflscape}

\usepackage{longtable}
\usepackage{array}

\newcolumntype{L}[1]{>{\raggedright\arraybackslash}p{#1}}

\usepackage{tikz}
\usepackage{quantikz}
\usepackage{adjustbox}

\usetikzlibrary{
    positioning,
    fit,
    backgrounds,
    shapes.geometric,
    calc,
    tikzmark
}
\providecommand{\ket}[1]{|#1\rangle}
\providecommand{\bra}[1]{\langle#1|}
\providecommand{\braket}[2]{\langle#1|#2\rangle}

\usepackage[colorlinks=true, allcolors=blue]{hyperref}
\usepackage{booktabs} % for toprule, midrule, bottomrule in tables
\usepackage{authblk} % for multiple authors and affiliations

\newtheorem{proposition}{Proposition}
\makeatletter
\renewcommand{\maketitle}{%
  \begingroup
  \setlength{\parindent}{0pt}%
  \setlength{\parskip}{0pt}%
  \renewcommand\thefootnote{\@fnsymbol\c@footnote}%
  \begin{center}
    {\LARGE\bfseries \@title \par}
    \vspace{4pt}
    {\normalsize \@author \par}
    \vspace{-6pt}
  \end{center}
  \setcounter{footnote}{0}%
  \renewcommand\thefootnote{\arabic{footnote}}%
  \endgroup
}
\makeatother

\usepackage{graphicx}
\usepackage{subcaption}

\usepackage{caption}
\usepackage{setspace}
\usepackage{titlesec}

\titleformat{\section}
  {\normalfont\bfseries\normalsize}{\thesection}{0.6em}{}
\titleformat{\subsection}
  {\normalfont\bfseries\small}{\thesubsection}{0.6em}{}

\titlespacing*{\section}{0pt}{0.6\baselineskip}{0.4\baselineskip}
\titlespacing*{\subsection}{0pt}{0.5\baselineskip}{0.3\baselineskip}

\author[1]{Rulla Al-Haideri}
\author[1]{Bilal Farooq\thanks{Corresponding author. Email: bilal.farooq@torontomu.ca}}
\author[2]{Karim Ismail}

\affil[1]{\footnotesize Laboratory of Innovations in Transportation (LiTrans), Toronto Metropolitan University}
\affil[2]{\footnotesize Department of Civil and Environmental Engineering, Carleton University}

\title{Q-SCM: A Quantum-Sequential Choice Model for Driver Mental State Evolution}
\date{}

\begin{document}
\maketitle
\begingroup \renewcommand\thefootnote{\fnsymbol{footnote}} \footnotetext[1]{Corresponding author. Email: bilal.farooq@torontomu.ca} \endgroup
\vspace{12pt}

\begin{abstract}
\noindent
Driver behaviour in traffic interactions is shaped by evolving perception, exposure to perceptual cues, and the order in which risk related information is processed in the driver's mind.
%Conventional latent class choice models can represent unobserved behavioural heterogeneity, however, class membership is commonly specified as a static or observation specific function of observed covariates. 
%Although lagged variables or dynamic extensions can be added, these models do not usually provide an explicit and natural mechanism where cue history, cue processing order, and phase-like memory jointly shape the evolution of a driver's latent mental state.
%Conventional models do not usually provide an explicit and natural mechanism where cue history, cue processing order, and phase-like memory jointly shape the evolution of a driver's latent mental state.
We propose a Quantum-Sequential Choice Model (Q-SCM) for modelling driver mental state evolution in interactive traffic environments. 
Q-SCM retains the classical latent class choice structure, but replaces the conventional class membership formulation with a quantum cognitive state model. 
%A unique feature of this model is that the quantum component is confined to the class membership layer, while the action choice layer remains a classical RUM. 
% In this way, Q-SCM departs from existing quantum choice models, which either derive choice probabilities directly from the Born rule or insert a quantum interference term inside the choice equation
% the model preserves the interpretability of random utility discrete choice models while still capturing memory, phase, and cue order effects.
The driver's latent state is represented as a two-state quantum system on the Bloch sphere including neutral and defensive states. 
Perceptual cues, including separation distance, closing time-to-collision (CTTC), and lane deviation induce sequential unitary rotations governed by Pauli quantum state evolution operators. 
This formulation allows the model to capture memory, phase effects, cue order dependence, and transitions between behavioural regimes that depend on prior cue history.
% add the relaxation mech
%To ensure well behaved state evolution, we introduce three control mechanisms: a monotonicity constraint that prevents pendulum-like overshoot, a geodesic safeguard mechanism that ensures convergence toward the defensive state under sustained threat exposure, and a relaxation step that allows recovery toward the neutral baseline when the threat weakens.
The model is estimated using 85,754 observations from 9,610 drivers extracted from naturalistic trajectories. 
The empirical results show that defensive state formation is not governed only by the instantaneous values of traffic cues, but also by the accumulated cue history and the order in which cues are processed. 
In particular, state changing cues, represented by separation distance and CTTC, directly increase the defensive probability. 
The phase changing cue, represented by lane deviation, modifies how subsequent state changing cue are interpreted.
%These findings suggest that Q-SCM provides a behaviourally meaningful framework for representing memory, cue order effects, and evolving defensive appraisal in driver behaviour modelling.

\end{abstract}

\noindent
{keywords: Driver behaviour; latent class model; quantum probability theory; Quantum cognition; dynamic mental state}

% =========
\section{Introduction}

Latent class models provide a highly flexible framework to capture unobserved behavioural heterogeneity in discrete choices. 
In these models, observed choices are modelled as a mixture of class-specific choice probabilities, weighted by the probability that the decision maker belongs to each latent class~\cite{greene2003latent}. 
This structure is attractive for driver behaviour modelling because the same traffic interaction might be interpreted as routine by one driver and threatening by another.
% drivers may interpret the same traffic interaction differently. 
For example, 
% the same surrounding traffic conditions may be perceived as routine by one driver but threatening to another driver.
% In many applications, latent class membership could be specified as a function of observed individual characteristics, contextual variables, or traffic interaction indicators
in many applications, class membership is specified as a function of individual characteristics, contextual variables or contemporaneous traffic conditions~\cite{alhaideri2026dualState}. 
% This provides a flexible and interpretable way to explain heterogeneity across drivers or observations. 
However, driving can be seen as a sequential cognitive process rather than a series of independent responses to isolated observations. 
Previous cue exposure and the order of which traffic cues are processed can influence how a driver interprets the current situation. 
Consequently, two drivers facing the same current traffic conditions may have different probabilities of being in a neutral or defensive state because they reached those conditions through different sequences of prior cues.
% who has just experienced a rapidly closing vehicle may not process the next separation distance in the same way as a driver who reached that same separation distance after a stable interaction. 
Therefore, a model of mental state evolution should connect the observed maneuvre not only to the current traffic conditions, also to the evolving internal state produced by earlier perceptual information.
% of latent state membership is fundamental to modelling defensive driving behaviour.

Temporal dependence has traditionally been represented in classical latent class models through lagged variables, interaction terms, hidden Markov structures, and transition equations.
%Therefore, the motivation for this paper is not that classical probability models are incapable of representing memory, order dependence or state transitions. 
%Instead, 
These effects generally require an explicitly designed dynamic structure involving additional variables, parameters or model-specific assumptions. 
%The narrower comparison examined in 
%This work concerns the mechanisms used to generate the time-varying latent state probabilities. 
In the classical latent class model, the probability of belonging to the defensive state is specified as a logit function of contemporaneous traffic conditions and it is therefore memoryless in the present specification.
Our work concerns the mechanisms used to generate the time-varying latent state probabilities.
In the proposed formulation, the cognitive state is carried forward and updated sequentially as new perceptual cues are processed. We use quantum state and transition mechanism to encode this cognitive state.
% Both formulations retain the same class-specific random utility model (RUM) for the observed manoeuvre. 

The probability amplitudes determine the current neutral and defensive state probabilities.
The relative phase and non-commuting cue updates allow prior context and cue processing order to influence subsequent state evolution~\cite{pothos2013BBS,bruza2015quantum}. 
This controlled comparison isolates the class membership mechanism.
%In the classical benchmark, the membership is contemporaneous and memoryless in its present specification.
In the proposed mechanism, the membership allows the current state probabilities to depend on the accumulated sequence of prior cue updates.
% The probability amplitudes determine the current neutral and defensive state probabilities.
% The relative phase and non-commuting cue updates allow prior context and cue processing order to influence subsequent state evolution~\cite{pothos2013BBS,bruza2015quantum}.
The quantum component is used only to generate the time-varying latent state probabilities for class participation. %These probabilities are then passed to the same class-specific classical random utility model used in the benchmark that determines the probability of the observed manoeuvre. Therefore, this controlled comparison isolates the latent class membership mechanism while keeping the observed choice model unchanged. 

In this paper, we model the driver's observed manoeuvre choice at successive trajectory point over time. 
The observed maneuver is assumed to be influenced by an unobserved mental state, represented here as either neutral or defensive~\cite{alhaideri2026dualState}. 
A model intended to represent the evolution of these states must do more than relating state membership to the traffic conditions observed at a single trajectory point. It must allow the neutral and defensive states evolve over time, carry forward the effects of previous cue exposure, account for context and sequence in which cues are processed, and connect the resulting latent state to the observed manoeuvre.
Quantum probability theory (QPT) provides a compact mathematical structure for the sequential state evolution. 
%The term ``quantum'' in this paper does not imply that the brain is assumed to be a physical quantum system. The term ``quantum'' in this context does not imply that the brain is physically a quantum system. 
Quantum cognition instead uses the mathematical formalism of QPT to represent context-dependent judgment and decision processes
~\cite{busemeyer2015quantum_cognition,pothos2022quantum}
A quantum cognitive state is represented by a state vector that changes when information is processed.
The squared magnitudes of its probability amplitudes produce the observable probabilities of the latent states, while the relative phase retains contextual information that can influence subsequent updates~\cite{pothos2013BBS,bruza2015quantum}.
% uses the mathematical structure of QPT to describe cognitive phenomena that are difficult to represent parsimoniously using a single fixed classical probability space. 
% These phenomena include, order effects~\cite{trueblood2011quantum,wang2014context}, interference and conjunction-disjunction violations~\cite{pothos2013BBS}, entanglement-like context dependence in concept combinations~\cite{aerts2014entanglement} and state change in evaluation~\cite{bruza2015quantum}.
% In this framework, an individual's cognitive state is represented by a state vector, and an evaluation changes that state. 
% Then, the probability of a later judgment or action can depend on the sequence of prior evaluations.
% This distinction is important for driver behaviour. 
% In classical probability theory (CPT), events are represented within a common sample space.
% The probabilities are updated according to classical rules. 
% In QPT, cognitive evaluations can be incompatible~\cite{busemeyer2015quantum_cognition}.
% Incompatibility means that evaluating one aspect of a situation can change the state from which another aspect is evaluated. 
% This distinction is useful for modelling driver appraisal. Two drivers may have the same current defensive state probability but different internal phases because of their previous cue histories, and may therefore respond differently to the same subsequent cue. Moreover, cue-update operators associated with different cognitive effects need not commute.
Processing cue $A$ followed by cue $B$ can produce a different state from processing cue $B$ followed by cue $A$. 
In this way, state carry forward, contextual phase and cue order dependence emerge from one sequential state representation rather than from separate lag, interaction and transition specification.

This paper proposes a Quantum-Sequential Choice Model (Q-SCM) for modelling the evolution of neutral and defensive driver states established in \citet{alhaideri2026dualState}, in interactive traffic environments. 
The Q-SCM retains the standard latent class choice structure but replaces the conventional class membership function with a quantum cognitive state layer.
% The class membership is reformulated using a quantum cognitive state model while the conditional choice layer remains classical. 
% Once the probabilities of the latent states are obtained, the final action probability is computed as a weighted sum of class specific choice probabilities. 
% The proposed contribution focuses on how class membership probabilities are generated. 
% The Q-SCM represents class membership as an evolving cognitive state shaped by perceived sequential cue updates.
% An important feature of the proposed model is that the quantum mechanism is restricted to the class membership layer. The conditional choice layer remains a classical random utility model (RUM). 
Perceptual traffic cues update the driver's latent state sequentially.
The resulting neutral and defensive probabilities are obtained from the state amplitudes using the Born rule. 
These probabilities serve only as the time-varying class membership weights. 
Conditional on each latent state, the observed manoeuvre remains governed by a classical random utility maximization (RUM).
The final manoeuvre probability is obtained as the weighted sum of the state-specific choice probabilities. 
Thus, Q-SCM combines two probability frameworks with distinct roles. First, QPT which governs the evolution of latent class membership.
Second, the classical probability and random utility theory govern the observed manoeuvre choice. The model does not replace RUM based discrete choice modelling with a quantum choice rule.
The contributions of this paper are threefold:
\begin{itemize}
    \item Formulation of dynamic latent class membership as an evolving quantum cognitive state model that carries forward the effects of previous perceptual cue exposure within the latent class modelling framework.
    \item Sequential neutral--defensive state representation where state-changing and phase-changing updates allow the memory, context and cue effects to influence the class membership probabilities.
    \item  Estimation on naturalistic trajectory data and comparison with a classical latent class benchmark which uses the same conditional choice specification.
\end{itemize}
%First, the paper formulates latent class membership as an evolving quantum cognitive state model that carries forward the effects of previous perceptual cue exposure. 
%Second, it develops a sequential neutral--defensive state representation where state-changing and phase-changing updates allow the memory, context and cue effects to influence the class membership probabilities, while retaining a fully classical RUM for the observed manoeure. 
% Third, unlike existing quantum choice models, which obtain choice probabilities from the Born rule or embed a quantum interference term inside the choice equation over the alternatives, the Q-SCM confines the quantum mechanism to the class membership layer and retains a fully classical RUM choice layer. 
% This retains the interpretability and welfare-analytic properties of RUM discrete choice models while still capturing memory, phase, and cue order effects through the membership layer. 
%Third, it estimates the proposed model using. 
%Since the choice kernel is held fixed, differences between the models can be attributed to the mechanism used to generate the latent state probabilities instead of a different utility specification.
% it demonstrates the proposed formulation using naturalistic trajectory data and compares it with a classical latent class benchmark that uses the same conditional choice layer.

The remainder of the paper is organized as follows. 
Section~\ref{sec:lit_rev} reviews the relevant background literature.
Section~\ref{sec:framework} presents the proposed Q-SCM framework and explains how the quantum membership layer is combined with the classical manoeuvre choice layer. 
Section~\ref{sec:application} applies the proposed framework to a publicly-available naturalistic trajectory dataset.
Section~\ref{sec:results} presents the estimation results and discusses the behavioural interpretation of the model.
Section~\ref{sec:conc} concludes the paper.

% =====================================================
% add a comparison table for the reviewed studies at end of sec
\section{Background}
\label{sec:lit_rev}
% Existing models that use the QPT attempted to model how people make choices
Hanckock et al.~\cite{hancock2020quantum_trb} developed choice models based on QPT. 
Their models represent the preferences of an individual using a belief state (i.e., state vector) as a unit vector in a complex Hilbert space. 
In their model, each alternative is set at an axis. The choice probabilities are obtained by squaring the projection of that vector onto each alternative (i.e., Born rule).
They develop two versions of the model. The amplitude model captures the final belief state amplitudes as a function of attribute differences between alternatives. 
The Hamiltonian model evolves an initial belief state over time using the Schrödinger equation. 
% The attribute comparisons in this model are embedded in the Hamiltonian. 
% Both models can transform objective attribute differences such as differences in travel time or cost to subjective evaluations using several behavioural value functions.
The models also employ quantum rotations to capture context and order effects, such as changes in perspective when choosing the best alternative versus the worst alternative. 
However, these models are not RUMs as they replace the additive error and utility maximization structure with choice probabilities derived from probability amplitudes by the Born rule.
%As an extension to these quantum choice models, 
Hanckock et al.~\cite{hancock2020quantum_jocm} further extend the quantum amplitude framework for moral choices. Those choices represent ones that are defined by a change of perspective by the decision maker using two mechanisms. 
The first one applies a quantum rotation based on the Pauli matrices to the state vector when a morally charged trade-off is present which affects the choice probabilities of the alternatives. 
The second assigns a separate complex phase to each attribute specific value function. 
This allows moral attributes such as deaths and injuries, to be processed differently from concrete attributes like time and cost.

Epping et al.~\cite{epping2023open} compared three models that capture how individuals make choices between two alternatives. 
They focused on jointly predicting which alternative is chosen and how long does the decision takes. 
% Their models capture the gradual accumulation of preference as the individual weighs the alternatives until enough evidence builds up to select an alternative.
Vitetta~\cite{vitetta2016quantum} proposed a quantum utility model (QUM) for route choice when individuals retain multiple possible routes before making a final decision during the trip. 
The model extends classical RUM probability by adding a quantum interference term that captures the coexistence of intermediate route preferences. 
When the pre-trip decision is unique, the interference term collapses and the QUM reduces to conventional RUM.
However, when multiple alternatives coexist, the model departs from the classical RUM formulation.
Di Gangi and Vitetta~\cite{digangi2021quantum} proposed Quantum Utility integrated with a Random Utility Model (QUARUM) for route choice. 
The model retains the classical RUM structure and adds a QUM-derived interference term at the route generation level to represent interactions among candidate routes during perception.
After this perception stage, route choice is modelled using a RUM formulation where route overlap is accounted for by a covariance structure among alternatives.
Lipovetsky~\cite{lipovetsky2018quantum} used probability amplitudes to model when individuals make choices when selecting among competing products. 
In their formulation, each alternative corresponds to one component of a state vector representing the decision maker's preferences. Those components are complex probability amplitudes where the squared moduli produce the choice probabilities.
Yu and Jayakrishnan~\cite{yu2018quantum} proposed a quantum model where the traveler's mental state is represented as a state vector in a Hilbert space. 
The survey question and the real world decision were treated as two different measurements acting on that state. 
Each possible choice such as driving alone or taking a ride sharing service is represented as its own direction in this space. Then the probability of a choice is given by how closely the mental state vector aligns with that direction.

To the best of the authors' knowledge, there is no existing study that has formulated the quantum state vector as an evolving latent cognitive state which updates sequentially by perceptual traffic cues to generate class membership probabilities in a choice model.
In the reviewed quantum choice models, the state vector generally represents preferences over choice alternatives, the basis axes correspond to those alternatives, and choice probabilities are obtained from the Born rule. Thus, the quantum component operates at the level of the alternatives rather than at the level of an internal driver state that carries memory of prior cues before an action is chosen.
Table~\ref{tab:qscm_comparison} summarizes this distinction. Most reviewed quantum choice models replace the additive random utility and utility maximization structure with probabilities derived from quantum amplitudes. 
The main exceptions are the quantum utility models of Vitetta~\cite{vitetta2016quantum} and Di Gangi and Vitetta~\cite{digangi2021quantum}. They retain a RUM component but augment the alternative choice probability with a quantum interference term. In those models, the quantum and RUM components are blended within the same choice equation defined over alternatives.

\begin{landscape}
\begin{table}[!htbp]
\centering
\caption{Comparison of existing quantum-like choice models and proposed Q-SCM.}
\label{tab:qscm_comparison}
\footnotesize
\setlength{\tabcolsep}{4pt}
\renewcommand{\arraystretch}{1.15}
\begin{adjustbox}{max width=\linewidth}
\begin{tabular}{p{0.1\linewidth} p{0.10\linewidth} p{0.1\linewidth} p{0.11\linewidth} p{0.13\linewidth} p{0.11\linewidth} p{0.15\linewidth} p{0.11\linewidth}}
\toprule
\textbf{Paper} &
\textbf{Application} &
\textbf{State vector represents} &
\textbf{Basis axes represent} &
\textbf{Model memory} &
\textbf{Order effects} &
\textbf{Choice probability mechanism} &
\textbf{RUM structure?} \\
\midrule
Hancock et al.~\cite{hancock2020quantum_trb}
& Travel choice (SP)
& Decision maker's preference
& Choice alternatives, one axis per alternative
& Partial: Hamiltonian model evolves belief in time; amplitude model is static
& Perspective rotations
& Born rule
& \textbf{No} \\
\addlinespace
Hancock et al.~\cite{hancock2020quantum_jocm}
& Moral choices
& Decision maker's preference
& Choice alternatives
& No
& Pauli rotation and per-attribute phase
& Born rule 
& \textbf{No} \\
\addlinespace
Epping et al.~\cite{epping2023open}
& Binary choice and response time
& Accumulated preference
& Two alternatives, response / evidence states
& No
& No
& Open system quantum Markov evidence accumulation
& \textbf{No} \\
\addlinespace
Di Gangi \& Vitetta~\cite{digangi2021quantum}
& Path / route choice
& Traveler's route preference
& Candidate routes
& No 
& Via interference (not sequential cue order)
& RUM (logit) choice + quantum interference terms
& \textbf{Hybrid}: interference added in choice equation over alternatives (RUM is special case) \\
\addlinespace
Vitetta~\cite{vitetta2016quantum}
& Route choice
& Traveler's route preference
& Candidate routes
& No (intermediate decision levels only)
& Via interference
& RUM probability + interference terms
& \textbf{Hybrid}: interference added in choice equation of alternatives \\
\addlinespace
Lipovetsky~\cite{lipovetsky2018quantum}
& Product choice
& Decision maker's preferences
& Alternatives, one amplitude component per product
& No
& Entanglement / interference
& Squared moduli of complex amplitudes, augmenting MNL
& \textbf{No} \\
\addlinespace
Yu and Jayakrishnan~\cite{yu2018quantum}
& SP vs.\ RP mode choice
& Traveller's mental state
& Choice alternatives as directions, SP and RP are two measurements
& No
& Measurement order
& Born rule
& \textbf{No} \\
\midrule
\textbf{Q-SCM}
& Driver manoeuvre choice
& \textbf{Driver's internal latent mental state}
& \textbf{Two cognitive states}: neutral and defensive
& \textbf{Yes}: sequential Bloch rotations with state carry forward
& Non-commuting Pauli updates
& Quantum layer generates latent class membership only, classical MNL choice layer
& \textbf{Yes} in choice layer, quantum only in membership \\
\bottomrule
\end{tabular}
\end{adjustbox}
\end{table}
\end{landscape}

The proposed Q-SCM differs in both the role and location of the quantum component. 
The state vector represents the driver's internal mental state, defined here by neutral and defensive states, rather than the choice alternatives. 
The quantum layer only generates the time varying latent class membership probabilities.
The observed manoeuvre choice is still modelled through a classical RUM layer. 
Therefore, the quantum probabilities do not enter the choice equation directly.
They only determine the weights assigned to the neutral and defensive RUM sub-models. 
In this essence, Q-SCM keeps the quantum membership layer and the RUM action choice layer as two separate and stacked components, rather than fusing them within a single alternative level choice equation.

% --------------------------------------------------------------------------

\section{Methods}
\label{sec:framework}

Let $i=1,\ldots,I$ index drivers, $t=0,\ldots,T_i$ index trajectory samples within the driving episode of driver $i$, and $j\in\mathcal{A}$ index the observed manoeuvre alternatives. 
Also, let $\mathcal{S}=\{N,D\}$ represent the two latent mental states, where $N$ and $D$ represent the neutral and defensive states, respectively.
The Q-SCM retains the standard latent class choice structure. The probability that driver $i$ selects manoeuvre $j$ at trajectory sample $t$ is~\cite{greene2003latent}:
\begin{equation}
P_{ijt} = \sum_{s\in\mathcal{S}} P_{ijt\mid s}P_{ist}
\label{eq:latent_class_choice}
\end{equation}
where $P_{ist}$ is probability that driver $i$ is in latent mental state $s$ at trajectory sample $t$. $P_{ijt\mid s}$ is the probability that driver $i$ select manoeuvre $j$ conditional on being in latent mental state $s$. 
For ease of reference, a consolidated summary of all the indices, sets, variables, operators, probabilities, and parameters used throughout the paper are provided in Appendix Table~\ref{tab:notation_summary}.

The distinguishing feature of Q-SCM lies in the generation of the time varying latent state membership probabilities $P_{ist}$. 
Rather than specifying these probabilities as direct functions of contemporaneous covariates, Q-SCM represents the driver's latent mental state as an evolving quantum state vector. 
Perceptual traffic cues update this state sequentially.
The resulting neutral and defensive probabilities are used as the latent class membership weights in Eq.~\eqref{eq:latent_class_choice}. 
The conditional manoeuvre choice probabilities $P_{ijt\mid s}$ remain governed by a classical random utility model.

The remainder of this section is organized in two parts. We first provide a concise overview of the five connected steps in the Q-SCM architecture for computing class membership and summarize their relationships in Figure~\ref{fig:qscm_architecture}. 
The subsequent subsections then develop each component in detail, including representation and initialization of the latent mental state, the sequential cue update process, the model control and relaxation mechanisms, and the likelihood and estimation procedure. The Q-SCM transforms perceptual cue information into an observed manoeuvre
probability through five connected steps.

\medskip

\noindent\textbf{Step 1: Initial state specification.}
At the beginning of the observed driving episode, driver $i$ is assigned an initial latent mental state:
\begin{equation}
\ket{\psi_{i0,0}} = \mathcal{U}_{\mathrm{init}}(\boldsymbol{\zeta}_i)\ket{N},
\label{eq:architecture_initialization}
\end{equation}
where $\mathcal{U}_{\mathrm{init}}(\boldsymbol{\zeta}_i)$ is the selected initialization operator and $\boldsymbol{\zeta}_i$ contains its governing parameters. 
The framework allows different initial state specifications, such as a fully neutral state, an equal neutral--defensive superposition using a Hadamard gate, or an estimated initial state. Initialization is applied
only once, at $t=0$.

\medskip

\noindent\textbf{Step 2: Cue encoding.}
At trajectory sample $t$, let $c_{it,m}\in\mathcal{K}$ represent the type of cue processed at within-sample update $m$.
The $\mathcal{K}$ is the set of admissible cue types. 
Let $x_{it,m}$ represent the observed input associated with that cue. The raw cue input is first transformed to a cue-specific representation:
\begin{equation}
s_{it,m} = f_{c_{it,m}} \left( x_{it,m}; \boldsymbol{\xi}_{c_{it,m}} \right),
\label{eq:general_cue_representation}
\end{equation}
% \noindent
where $f_c(\cdot)$ is the cue encoding function and $\boldsymbol{\xi}_c$ contains its parameters. 

The encoded cue is then mapped to a rotation angle:
\begin{equation}
\theta_{it,m} = g_{c_{it,m}} \left( s_{it,m}; \boldsymbol{\rho}_{c_{it,m}} \right),
\label{eq:general_rotation_mapping}
\end{equation}
where $g_c(\cdot)$ is a cue-specific rotation-response function and $\boldsymbol{\rho}_c$ contains its parameters. The cue is also assigned to a rotation axis by:
\begin{equation}
\mathfrak{a}(c_{it,m}) \in \{x,z\}.
\label{eq:architecture_axis_assignment}
\end{equation}
The cue encoding functions, rotation-response functions, and axis
assignments are specified separately for each empirical application.

\medskip

\noindent\textbf{Step 3: Sequential state updates.}
Each encoded cue produces a candidate update of the current latent mental
state:
\begin{equation}
\ket{\psi^{*}_{it,m+1}} = \mathcal{U}_{\mathfrak{a}(c_{it,m})} \left( \theta_{it,m} \right) \ket{\psi_{it,m}}, 
\label{eq:architecture_candidate_update}
\end{equation}
where:
\begin{equation}
\mathcal{U}_a(\theta_{it,m}) = \exp\left(-i\frac{\theta_{it,m}}{2}\sigma_a \right), \qquad a\in\{x,z\}.
\label{eq:architecture_single_axis_operator}
\end{equation}
The candidate state is evaluated by the model control operation $\mathcal{C}_{it,m}$, which incorporates the monotonicity constraint and, when required, the geodesic safeguard:
\begin{equation}
\ket{\psi_{it,m+1}} = \mathcal{C}_{it,m}
\left( \ket{\psi^{*}_{it,m+1}}, \ket{\psi_{it,m}} \right).
\label{eq:architecture_controlled_update}
\end{equation}
The cue update process is repeated for $m=0,\ldots,M_{it}-1$ until all cues at the trajectory sample $t$ have been processed.

\medskip

\noindent\textbf{Step 4: Relaxation and state carry forward.}
After all $M_{it}$ within-sample cue updates have been implemented, the pre-relaxation state is defined as:
\begin{equation}
\ket{\widetilde{\psi}_{it}} = \ket{\psi_{it,M_{it}}}.
\label{eq:architecture_pre_relaxation_state}
\end{equation}
The relaxation operation then gives the final latent mental state at trajectory sample $t$:
\begin{equation}
\ket{\psi_{it}} = \mathcal{R}_{\lambda} \left( \ket{\widetilde{\psi}_{it}}, \ket{\psi_{i0,0}} \right),
\label{eq:architecture_relaxation}
\end{equation}
where $\mathcal{R}_{\lambda}$ is normalized interpolation toward the initial baseline state $\ket{\psi_{i0,0}}$.
The $\lambda$ governs the strength of relaxation. The resulting state $\ket{\psi_{it}}$ is the final state at trajectory sample $t$ after all cue updates and relaxation have been applied.

This final state is carried forward as the start-of-sample state for the next trajectory sample:
\begin{equation}
\ket{\psi_{i,t+1,0}} = \ket{\psi_{it}}, \qquad t=0,\ldots,T_i-1.
\label{eq:architecture_carry_forward}
\end{equation}
Therefore, the initialization is applied only once at $t=0$.
The cue update, model control, relaxation, and state carry forward operations are repeated throughout the observed driving episode.

\medskip

\noindent\textbf{Step 5: Probability extraction and manoeuvre choice.}
The final latent mental state produces neutral and defensive latent state membership probabilities by the Born rule:
\begin{equation}
P_{iNt} = \left| \braket{N}{\psi_{it}} \right|^2, \qquad P_{iDt} = \left| \braket{D}{\psi_{it}} \right|^2.
\label{eq:architecture_membership_probabilities}
\end{equation}
These probabilities are combined with the classical state conditional
manoeuvre choice probabilities:
\begin{equation}
P_{ijt} = P_{iNt}P_{ijt\mid N} + P_{iDt}P_{ijt\mid D}.
\label{eq:architecture_final_probability}
\end{equation}

% In Figure~\ref{fig:qscm_circuit}, the steps and colour coding are linked to the architecture diagram in Figure~\ref{fig:qscm_architecture}. 
% The notation $R_a(\theta)$ denotes a rotation of the cognitive state by angle $\theta$ around axis $a\in\{x,y,z\}$ on the Bloch sphere. 
% % the step 1 includes initialization + carry forward
% The Step~1 gate initializes the driver's cognitive state by setting the baseline mixture between the neutral state $\ket{N}$ and defensive state $\ket{D}$. This is achieved by applying the unitary matrix operation on \ket{\psi} to rotate around y-axis:

\newpage
Figure~\ref{fig:qscm_architecture} shows the five steps of the Q-SCM architecture. 
Steps~1 and~2 provide the initial state and cue-specific update parameters to the sequential state evolution process in Step~3.
After all within-sample cues have been processed, Step~4 applies relaxation and carries the final state forward to the next trajectory sample. 
Step~5 converts the final state into latent state membership probabilities and combines them with the state-conditional manoeuvre choice probabilities.

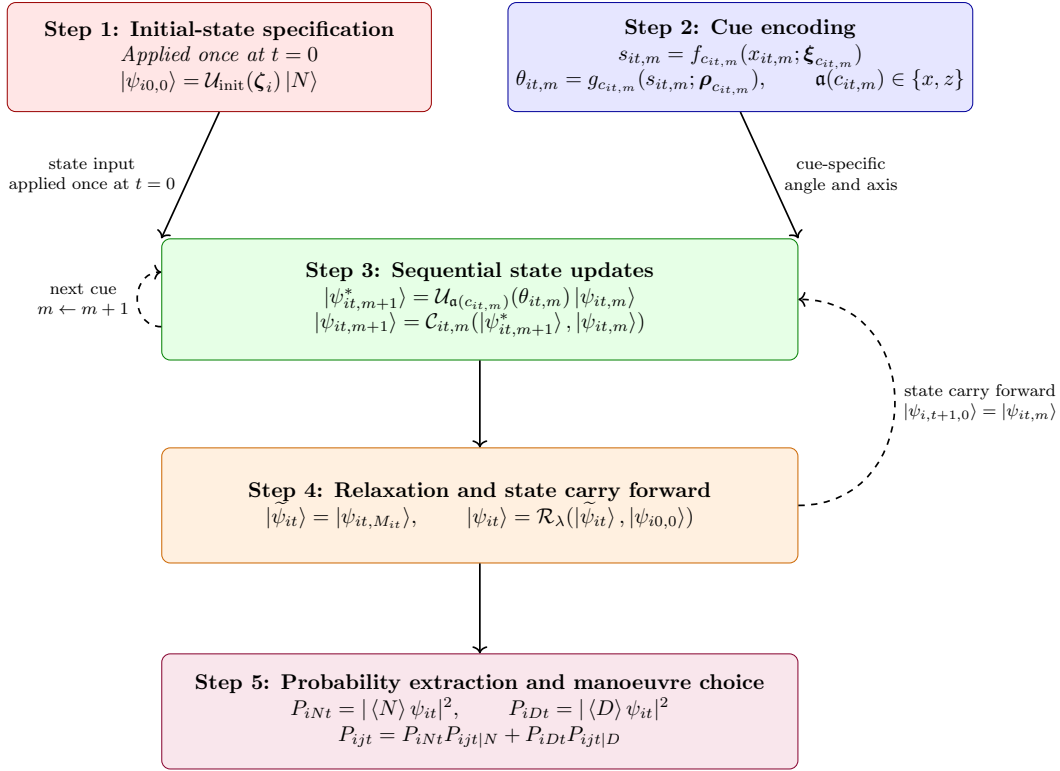
\begin{figure}[H]
\centering
\resizebox{0.8\textwidth}{!}{%
\begin{tikzpicture}[
    arrow/.style={->,thick},
    stepone/.style={
        rectangle,
        rounded corners,
        draw=red!60!black,
        fill=red!10,
        align=center,
        minimum width=7.0cm,
        minimum height=1.8cm,
        font=\small
    },
    steptwo/.style={
        rectangle,
        rounded corners,
        draw=blue!60!black,
        fill=blue!10,
        align=center,
        minimum width=7.0cm,
        minimum height=1.8cm,
        font=\small
    },
    stepthree/.style={
        rectangle,
        rounded corners,
        draw=green!50!black,
        fill=green!10,
        align=center,
        minimum width=10.5cm,
        minimum height=2.0cm,
        font=\small
    },
    stepfour/.style={
        rectangle,
        rounded corners,
        draw=orange!80!black,
        fill=orange!12,
        align=center,
        minimum width=10.5cm,
        minimum height=1.9cm,
        font=\small
    },
    stepfive/.style={
        rectangle,
        rounded corners,
        draw=purple!70!black,
        fill=purple!10,
        align=center,
        minimum width=10.5cm,
        minimum height=1.9cm,
        font=\small
    }
]

\node[stepone] (init) at (-4.3,0) {
\textbf{Step 1: Initial-state specification}\\
\textit{Applied once at $t=0$}\\
$\ket{\psi_{i0,0}} = \mathcal{U}_{\mathrm{init}}(\boldsymbol{\zeta}_i)\ket{N}$
};

\node[steptwo] (encoding) at (4.3,0) {
\textbf{Step 2: Cue encoding}\\
$s_{it,m} = f_{c_{it,m}}
(x_{it,m};\boldsymbol{\xi}_{c_{it,m}})$\\
$\theta_{it,m} = g_{c_{it,m}}
(s_{it,m};\boldsymbol{\rho}_{c_{it,m}})$,
\qquad
$\mathfrak{a}(c_{it,m})\in\{x,z\}$
};

\node[stepthree] (update) at (0,-4.0) {
\textbf{Step 3: Sequential state updates}\\
$\ket{\psi^{*}_{it,m+1}}
= \mathcal{U}_{\mathfrak{a}(c_{it,m})} (\theta_{it,m})
\ket{\psi_{it,m}}$\\
$\ket{\psi_{it,m+1}} = \mathcal{C}_{it,m} (\ket{\psi^{*}_{it,m+1}},\ket{\psi_{it,m}})$
};

\node[stepfour] (relaxation) at (0,-7.4) {
\textbf{Step 4: Relaxation and state carry forward}\\
$\ket{\widetilde{\psi}_{it}} = \ket{\psi_{it,M_{it}}}$,
\qquad
$\ket{\psi_{it}} = \mathcal{R}_{\lambda} (\ket{\widetilde{\psi}_{it}},\ket{\psi_{i0,0}})$
};

\node[stepfive] (choice) at (0,-10.8) {
\textbf{Step 5: Probability extraction and manoeuvre choice}\\
$P_{iNt}=|\braket{N}{\psi_{it}}|^2$,
\qquad
$P_{iDt}=|\braket{D}{\psi_{it}}|^2$\\
$P_{ijt} = P_{iNt}P_{ijt\mid N} + P_{iDt}P_{ijt\mid D}$
};

\draw[arrow]
(init.south)
--
node[
    midway,
    left=3pt,
    font=\scriptsize,
    align=center
]
{state input\\applied once at $t=0$}
(update.north west);

\draw[arrow]
(encoding.south)
--
node[
    midway,
    right=3pt,
    font=\scriptsize,
    align=center
]
{cue-specific\\angle and axis}
(update.north east);

\draw[arrow] (update) -- (relaxation);
\draw[arrow] (relaxation) -- (choice);

\draw[
    arrow,
    dashed
]
([yshift=-0.45cm]update.west)
to[
    out=180,
    in=180,
    looseness=1.5
]
node[
    left,
    font=\scriptsize,
    align=center
]
{ next cue\\ $m\leftarrow m+1$ }
([yshift=0.45cm]update.west);

\draw[
    arrow,
    dashed
]
(relaxation.east)
to[
    out=0,
    in=0,
    looseness=1.6
]
node[
    right,
    font=\scriptsize,
    align=center
]
{
state carry forward\\
$\ket{\psi_{i,t+1,0}}=\ket{\psi_{it,m}}$
}
(update.east);

\end{tikzpicture}%
}

\caption{General architecture of the Q-SCM.}
\label{fig:qscm_architecture}
\end{figure}

The remainder of this section develops the components summarized in
Figure~\ref{fig:qscm_architecture}. 
First we define the latent mental state representation and its initialization. 
Then, we formulate the sequential state evolution process, including the cue-induced rotations, monotonicity constraint, geodesic safeguard, relaxation, and state carry forward mechanisms. 
Finally, we present the likelihood function and estimation procedure.

%In applied quantum mechanics, a circuit is the standard notation used to show how a state is transformed through an ordered sequence of operations. 
%Each operation is shown as a unitary gate acting on the state. 

\subsection{Mental State Representation}
\label{sec:mental_state}

In its most general form, the latent mental state can be represented as a
normalized vector in a $K$-dimensional complex Hilbert space:
\begin{equation}
\mathcal{H}_K = \operatorname{span} \left\{ \ket{S_1}, \ldots, \ket{S_K} \right\}
\label{eq:general_hilbert_space}
\end{equation}
where each basis vector $\ket{S_k}$ corresponds to one of $K$ possible latent mental states. The state of driver $i$ at trajectory sample $t$ and within-sample update $m$ can be expressed as:
\begin{equation}
\ket{\psi_{it,m}} = \sum_{k=1}^{K} \alpha_{k,it,m}\ket{S_k}, \label{eq:general_cognitive_state}
\end{equation}
where $\alpha_{k,it,m}\in\mathbb{C}$ is the complex probability amplitude associated
with the state $S_k$. 

Normalization requires:
\begin{equation}
\sum_{k=1}^{K} \left| \alpha_{k,it,m} \right|^2 = 1.
\label{eq:general_state_normalization}
\end{equation}
Consequently, the squared modulus $\left|\alpha_{k,it,m}\right|^2$ gives the latent state membership probability associated with basis state $\ket{S_k}$.
The present paper develops the foundational two-state formulation were $K=2$ \cite{alhaideri2026dualState}. The two basis states represent neutral $\ket{N}$ and defensive $\ket{D}$ driving:
\begin{equation}
\ket{N} = \begin{bmatrix} 1\\ 0
\end{bmatrix},
\qquad
\ket{D} = \begin{bmatrix} 0\\ 1
\end{bmatrix}.
\label{eq:basis_states}
\end{equation}
The latent mental state of driver $i$ at trajectory sample $t$ after $m$ within-sample cue updates can be written as:
\begin{equation}
\ket{\psi_{it,m}} = \alpha_{it,m}\ket{N} + \beta_{it,m}\ket{D},
\label{eq:two_state_representation}
\end{equation}
subject to:
\begin{equation}
\left|\alpha_{it,m}\right|^2 + \left|\beta_{it,m}\right|^2 = 1.
\label{eq:two_state_normalization}
\end{equation}
Here, $\alpha_{it,m}$ and $\beta_{it,m}$ are complex probability amplitudes. Their squared moduli produce the neutral and defensive state membership probabilities.
Also, their relative phase influence how the latent mental state responds to subsequent cue updates.
For the two-state formulation, the same latent mental state can be parameterized using a polar angle $\vartheta_{it,m}\in[0,\pi]$ and an azimuthal angle $\varphi_{it,m}\in[0,2\pi)$:
\begin{equation}
\ket{\psi_{it,m}} = \cos\left(\frac{\vartheta_{it,m}}{2}\right)\ket{N} + e^{i\varphi_{it,m}} \sin\left(\frac{\vartheta_{it,m}}{2}\right)\ket{D}.
\label{eq:bloch_parameterization}
\end{equation}
The corresponding neutral and defensive state membership probabilities
are:
\begin{equation}
\left|\braket{N}{\psi_{it,m}}\right|^2 = \cos^2\left(\frac{\vartheta_{it,m}}{2}\right), \qquad \left|\braket{D}{\psi_{it,m}}\right|^2 = \sin^2\left(\frac{\vartheta_{it,m}}{2}\right).
\label{eq:bloch_intermediate_probabilities}
\end{equation}
The polar angle $\vartheta_{it,m}$ is the neutral--defensive probability balance.
The azimuthal angle $\varphi_{it,m}$ is the relative phase. Two latent mental states may have the same current defensive state probability but different relative phases and may consequently respond differently to the same subsequent cue. The complete Bloch vector representation is provided in Appendix~\ref{sec:bloch}.
% add fig
The north and south poles of the Bloch sphere correspond to the neutral state $\ket{N}$ and defensive state $\ket{D}$, respectively. 
The polar angle $\vartheta_{it,m}$ determines the neutral--defensive probability
balance.
The azimuthal angle $\varphi_{it,m}$ represents the relative phase. Consequently, two latent mental states may have the same current defensive state probability but respond differently to a subsequent cue because their relative phases differ, as illustrated in Figure~\ref{fig:bloch_sphere}.

\begin{figure}[!htbp]
\centering
\includegraphics[
    width=0.4\linewidth,
    trim=1.4cm 1cm 1.4cm 1.9cm,
    clip
]{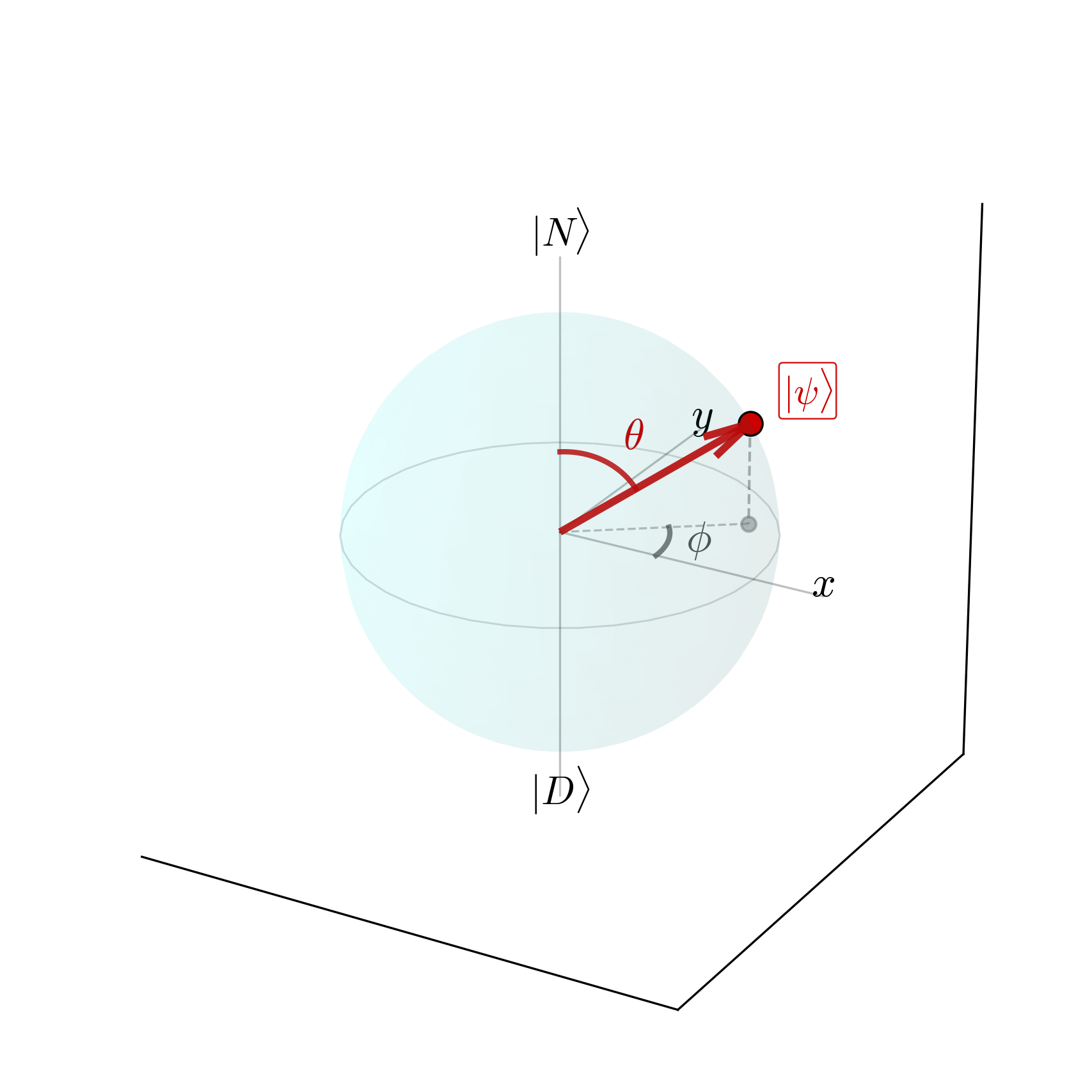}
\caption{Bloch sphere representation of the two-state Q-SCM latent mental
state.}
\label{fig:bloch_sphere}
\end{figure}

\subsection{Sequential Mental State Evolution}
\label{sec:state_evolution}

Each observed driving episode for driver $i$ consists of trajectory samples $t=0,\ldots,T_i$. 
At sample $t$, driver process $M_{it}$ perceptual cues sequentially. 
The index $m=0,\ldots,M_{it}$ identifies the intermediate latent mental states in the same trajectory sample. 
Specifically, $\ket{\psi_{it,0}}$ denotes the state before any cue at sample $t$ has been processed.
The $\ket{\psi_{it,M_{it}}}$ denotes the state after all $M_{it}$ cues have been processed.
The cue update index runs on $m=0,\ldots,M_{it}-1$, since the cue processed at update $m$ transforms the current state $\ket{\psi_{it,m}}$ into the next state $\ket{\psi_{it,m+1}}$. 
At $t=0$, the start-of-sample state $\ket{\psi_{i0,0}}$ is obtained from the initialization mechanism. In each subsequent trajectory sample, $\ket{\psi_{it,0}}$ is obtained by carrying forward the final state from the preceding sample.
Let $c_{it,m}\in\mathcal{K}$ denote the type of cue processed at update $(i,t,m)$, where $\mathcal{K}$ is the set of admissible cue types. 

\newpage
As defined in Eqs.~\eqref{eq:general_cue_representation} and \eqref{eq:general_rotation_mapping}, the observed cue input $x_{it,m}$ is first transformed into the encoded cue representation $s_{it,m}$ and then mapped to the rotation angle $\theta_{it,m}$. Each cue type is assigned to a rotation axis by:
\begin{equation}
\mathfrak{a}:\mathcal{K}\rightarrow\{x,z\},
\qquad
c_{it,m}\mapsto\mathfrak{a}(c_{it,m}).
\label{eq:axis_mapping}
\end{equation}
The cue encoding functions, rotation-response functions, axis assignments, and within-sample cue-processing order are application dependent modelling choices.
For the two-state Q-SCM, the cue-induced rotations are generated using the Pauli matrices:
\begin{equation}
\sigma_x = \begin{bmatrix} 0 & 1\\ 1 & 0
\end{bmatrix},
\qquad
\sigma_y = \begin{bmatrix} 0 & -i\\ i & 0
\end{bmatrix},
\qquad
\sigma_z = \begin{bmatrix} 1 & 0\\ 0 & -1
\end{bmatrix}.
\label{eq:pauli_matrices}
\end{equation}
Although cue types are assigned only to the $x$- and $z$-axes in the
present formulation, $\sigma_y$ is included because it emerges from the
non-commutativity of the $x$- and $z$-axis generators.
The unitary operator associated with cue $c_{it,m}$ is: 
\begin{equation}
\mathcal{U}_{it,m} \equiv \mathcal{U}_{\mathfrak{a}(c_{it,m})} \left( \theta_{it,m} \right)
=
\exp\left[
-i\frac{\theta_{it,m}}{2}
\sigma_{\mathfrak{a}(c_{it,m})}
\right].
\label{eq:cue_unitary_operator}
\end{equation}
This operator satisfies the unitarity condition:

\begin{equation}
\mathcal{U}_{it,m}^{\dagger}\mathcal{U}_{it,m} = I_2,
\label{eq:cue_unitarity_condition}
\end{equation}
where $\mathcal{U}_{it,m}^{\dagger}$ is the conjugate transpose of
$\mathcal{U}_{it,m}$ and $I_2$ is the two dimensional identity matrix. 
Therefore, each cue-induced rotation preserves the norm of the latent mental state and the total probability associated with the neutral and defensive basis states.
For a state-changing cue assigned to the $x$-axis, the rotation operator
is:
\begin{equation}
\mathcal{U}_x(\theta_{it,m}) = \begin{bmatrix}
\cos\left(\frac{\theta_{it,m}}{2}\right)
&
-i\sin\left(\frac{\theta_{it,m}}{2}\right)
\\[1ex]
-i\sin\left(\frac{\theta_{it,m}}{2}\right)
&
\cos\left(\frac{\theta_{it,m}}{2}\right)
\end{bmatrix}.
\label{eq:x_rotation_matrix}
\end{equation}

% Applying this operator to the current state:
% \begin{equation}
% \ket{\psi_{it,m}} = \alpha_{it,m}\ket{N} + \beta_{it,m}\ket{D}
% \end{equation}

\newpage
Using the representation in
Eq.~\eqref{eq:two_state_representation}, application of the
$x$-axis rotation produces the candidate amplitudes:
\begin{equation}
\begin{aligned}
\alpha^{*}_{it,m+1}
&=
\cos\left(\frac{\theta_{it,m}}{2}\right)\alpha_{it,m}
-
i\sin\left(\frac{\theta_{it,m}}{2}\right)\beta_{it,m},
\\
\beta^{*}_{it,m+1}
&=
-i\sin\left(\frac{\theta_{it,m}}{2}\right)\alpha_{it,m}
+
\cos\left(\frac{\theta_{it,m}}{2}\right)\beta_{it,m}.
\end{aligned}
\label{eq:x_rotation_amplitudes}
\end{equation}
Since $x$-axis rotation mixes the neutral and defensive amplitudes, it can directly change their squared moduli and therefore change the neutral--defensive probability balance.
For a phase-changing cue assigned to the $z$-axis, the rotation operator
is:
\begin{equation}
\mathcal{U}_z(\theta_{it,m}) = \begin{bmatrix}
e^{-i\frac{\theta_{it,m}}{2}} & 0
\\[1ex]
0 & e^{i\frac{\theta_{it,m}}{2}}
\end{bmatrix}.
\label{eq:z_rotation_matrix}
\end{equation}
Applying the $z$-axis rotation gives:
\begin{equation}
\alpha^{*}_{it,m+1} = e^{-i\frac{\theta_{it,m}}{2}}\alpha_{it,m}, \qquad \beta^{*}_{it,m+1} = e^{i\frac{\theta_{it,m}}{2}}\beta_{it,m}.
\label{eq:z_rotation_amplitudes}
\end{equation}
Consequently:
\begin{equation}
\left|\alpha^{*}_{it,m+1}\right|^2 = \left|\alpha_{it,m}\right|^2,
\qquad \left|\beta^{*}_{it,m+1}\right|^2 = \left|\beta_{it,m}\right|^2.
\label{eq:z_rotation_probability_invariance}
\end{equation}
Therefore, a $z$-axis rotation does not directly change the pre-relaxation neutral- or defensive state probabilities. 
It changes the  relative phase between the two amplitudes. Relaxation map acts on the complex amplitudes, this phase can influence the post relaxation
state at the current trajectory sample. The phase is also retained through
state carry forward and can influence subsequent state-changing updates.
For either axis assignment, the cue-induced candidate update is:
\begin{equation}
\ket{\psi^{*}_{it,m+1}} = \mathcal{U}_{\mathfrak{a}(c_{it,m})}
\left( \theta_{it,m} \right) \ket{\psi_{it,m}},
\qquad
m=0,\ldots,M_{it}-1.
\label{eq:proposed_cue_update}
\end{equation}

The superscript $*$ indicates that $\ket{\psi^{*}_{it,m+1}}$ is a candidate state produced by unitary cue-induced rotation. It is not necessarily the accepted state because the candidate update must first be evaluated by the model control operation.
The accepted state is therefore defined as:

\begin{equation}
\ket{\psi_{it,m+1}} = \mathcal{C}_{it,m} \left( \ket{\psi^{*}_{it,m+1}}, \ket{\psi_{it,m}} \right),
\label{eq:controlled_cue_update}
\end{equation}
where $\mathcal{C}_{it,m}$ incorporates the monotonicity constraint and, if required, geodesic safeguard. These mechanisms determine if the candidate update is accepted directly or redirected before the next cue is processed. These are developed in the upcoming subsections.
In addition, the order in which cues are processed can affect the resulting latent mental state. 
For instance, suppose the driver first perceives a state-changing threat cue assigned to the $x$-axis and then processes a phase-changing contextual cue assigned to the $z$-axis. 
Processing the same two cues in the reverse order may produce a different final state because the $x$- and $z$-axis generators do not commute:
\begin{equation}
[\sigma_x,\sigma_z] = \sigma_x\sigma_z-\sigma_z\sigma_x
= -2i\sigma_y \neq 0.
\label{eq:pauli_noncommutativity}
\end{equation}
Consequently, in general:
\begin{equation}
\mathcal{U}_x(\theta_x)\mathcal{U}_z(\theta_z) \neq \mathcal{U}_z(\theta_z)\mathcal{U}_x(\theta_x).
\label{eq:rotation_order_dependence}
\end{equation}
Therefore, processing a state-changing cue followed by a phase-changing cue need not produce the same latent mental state as processing the same two cues in the reverse order. This noncommutativity provides the Q-SCM mechanism through which within-sample cue order and prior cue exposure influence the subsequent evolution of the driver's latent mental state.

\subsubsection{Monotonicity Constraint Mechanism}
\label{sec:monotonicity_constraint}

The unitary cue induced rotations are periodic. 
Therefore, repeatedly applying a state-changing rotation could move the latent mental state beyond the defensive pole and begin reducing the defensive state probability, even though if the processed cue continues to represent a threat. Such behaviour of an increase followed by a decrease would be inconsistent with the intended behavioural interpretation of a state-changing threat cue.
To prevent this rotational overshoot, the Q-SCM evaluates every candidate
state-changing update before it is accepted. Let:
\begin{equation}
p^{D}_{it,m} = \left| \braket{D}{\psi_{it,m}} \right|^2
\label{eq:current_intermediate_defensive_probability}
\end{equation}
represnet the defensive state probability before cue $m$ is processed, and
let:
\begin{equation}
p^{D,*}_{it,m+1} = \left| \braket{D}{\psi^{*}_{it,m+1}} \right|^2
\label{eq:candidate_intermediate_defensive_probability}
\end{equation}
represent the defensive state probability associated with the candidate state generated by Eq.~\eqref{eq:proposed_cue_update}. The superscript $*$ indicates that this probability is obtained before application of the model control operation.
For a state-changing threat cue assigned to the $x$-axis, the
monotonicity operation is defined as:
\begin{equation}
\mathcal{M}_{it,m} \left( \ket{\psi^{*}_{it,m+1}},
\ket{\psi_{it,m}} \right) = \begin{cases} \ket{\psi^{*}_{it,m+1}},
&p^{D,*}_{it,m+1} \geq p^{D}_{it,m}, \\[1ex] \ket{\psi_{it,m}},
& p^{D,*}_{it,m+1} < p^{D}_{it,m}.
\end{cases}
\label{eq:monotonicity_operation}
\end{equation}
Phase-changing cues assigned to the $z$-axis do not require a separate monotonicity test because, as shown in
Eq.~\eqref{eq:z_rotation_probability_invariance}.
These leave the pre-relaxation Neutral- and Defensive state probabilities unchanged. The candidate update is accepted directly. Its phase effect may nevertheless influence the post relaxation state because the relaxation
operation acts on the complex amplitudes.

\begin{proposition}
\label{prop:cue_update_monotonicity}
Every controlled cue update performed before relaxation preserves or
increases the Defensive state probability.
\end{proposition}

The guarantee is non-strict, so a controlled update may leave the Defensive state probability unchanged. It applies only to the cue driven evolution before relaxation.
The subsequent relaxation operation may reduce the Defensive state probability by moving the state toward its
baseline. A proof of Proposition~\ref{prop:cue_update_monotonicity} is provided in
Appendix~\ref{app:formal_guarantees}.

\subsubsection{Geodesic Safeguard Mechanism}
\label{sec:geodesic_safeguard}

The monotonicity operation prevents an accepted cue-induced rotation to reduce defensive state probability when a persistent threat is prsent. 
However, simply retaining the current state when a candidate rotation is blocked does not guarantee continued progress toward the defensive pole under sustained threat exposure.
This limitation is geometric. A rotation around a fixed Pauli axis moves the Bloch vector in a circular path around that axis. 
In some locations on the Bloch sphere, the assigned $x$-axis rotation could begin
to move around the defensive pole rather than moving closer to it. 
The monotonicity operation correctly identifies and blocks that candidate update.
Howver, the repeated blocking can cause the state to stall at a defensive state probability below one. Figure~\ref{fig:geodesic_safeguard} shows this problem and the proposed mechanism.
\begin{figure}[!htbp]
\centering

\begin{subfigure}[b]{0.4\textwidth}
    \centering
    \includegraphics[
        width=\textwidth,
        trim=0 1cm 0 0.4cm,
        clip
    ]{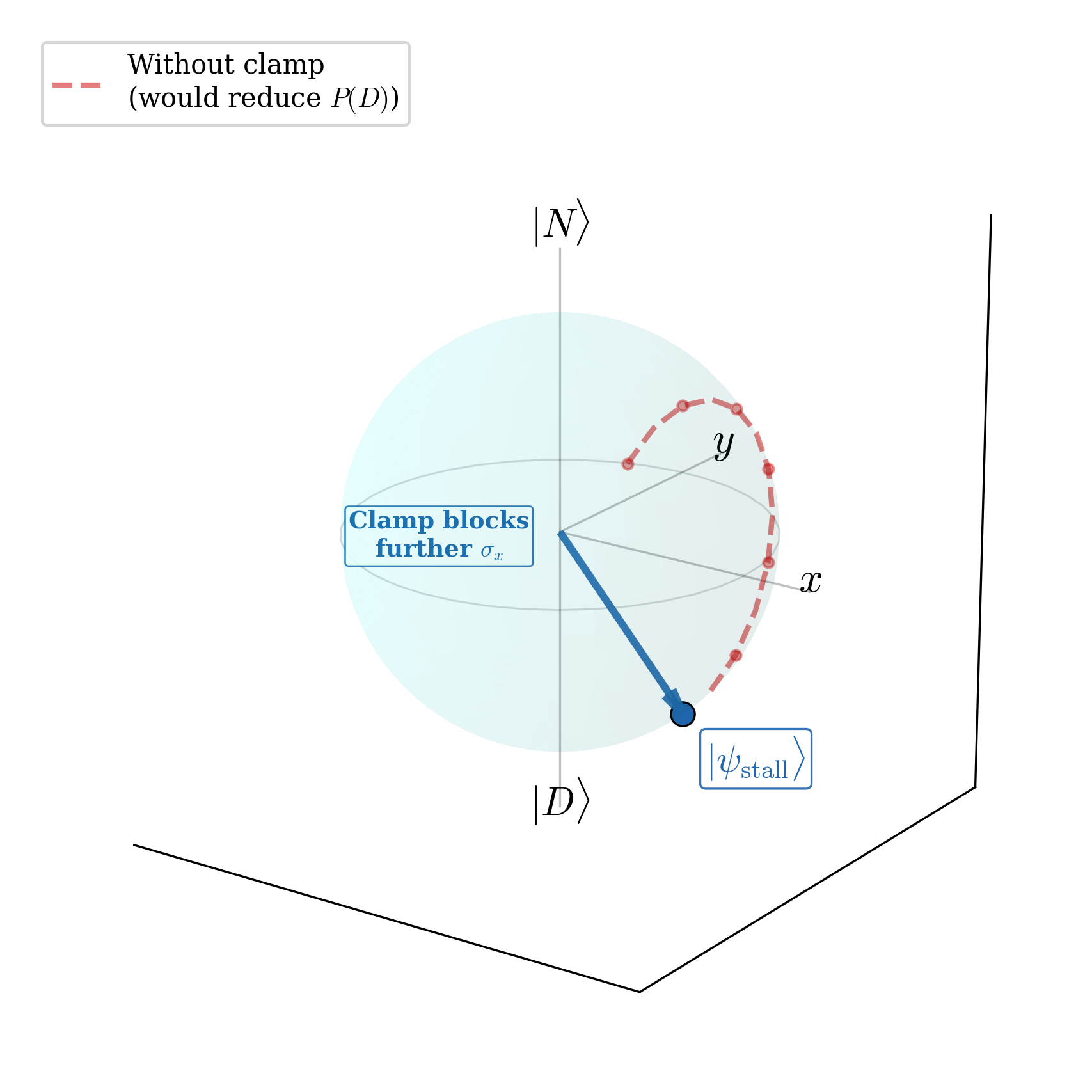}
    \caption{A fixed axis Pauli rotation is blocked.}
    \label{fig:geo_stall}
\end{subfigure}
\hfill
\begin{subfigure}[b]{0.4\textwidth}
    \centering
    \includegraphics[
        width=\textwidth,
        trim=0 1cm 0 0.4cm,
        clip
    ]{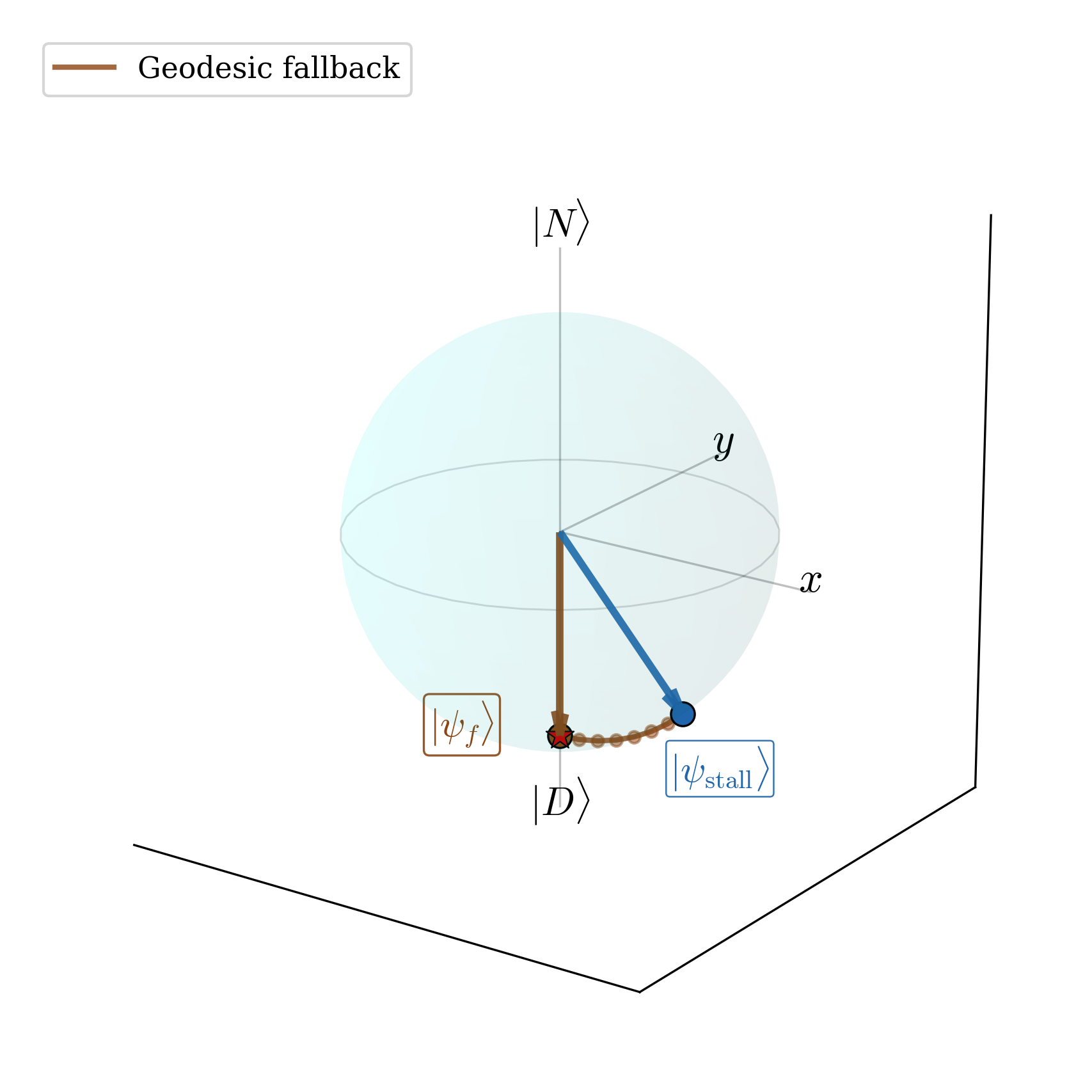}
    \caption{The blocked update is redirected to the defensive pole.}
    \label{fig:geo_redirect}
\end{subfigure}

\caption{Geodesic safeguard mechanism.}
\label{fig:geodesic_safeguard}
\end{figure}
Let:
\begin{equation}
\mathbf{r}_{it,m} = \bra{\psi_{it,m}}
\boldsymbol{\sigma} \ket{\psi_{it,m}}
= \left(
r_{x,it,m},
r_{y,it,m},
r_{z,it,m}
\right)
\label{eq:current_bloch_vector}
\end{equation}
represent the Bloch vector of the current latent mental state, where $\boldsymbol{\sigma}=(\sigma_x,\sigma_y,\sigma_z)$. The Bloch vector of the defensive state is:
\begin{equation}
\mathbf{r}_{D} = (0,0,-1).
\label{eq:defensive_bloch_vector}
\end{equation}
The remaining angular distance from the current state to the defensive pole is then: 
\begin{equation}
\phi_{it,m} = \arccos
\left( \mathbf{r}_{it,m}^{\mathsf T}\mathbf{r}_{D} \right)
= \arccos \left( -r_{z,it,m} \right).
\label{eq:angular_distance_to_defensive}
\end{equation}
When the current state is not located at either pole, the rotation axis
which move the Bloch vector to $\mathbf{r}_{D}$ in the shortest path is:
\begin{equation}
\widehat{\mathbf{n}}_{\mathrm{geo},it,m} = \frac{
\mathbf{r}_{it,m}\times\mathbf{r}_{D} }{ \left\| \mathbf{r}_{it,m}\times\mathbf{r}_{D} \right\|}
= \frac{ \left( -r_{y,it,m}, r_{x,it,m}, 0 \right) }{ \sqrt{ r_{x,it,m}^{2} + r_{y,it,m}^{2}} }.
\label{eq:geodesic_axis}
\end{equation}
The corresponding geodesic rotation operator is:
\begin{equation}
\mathcal{U}_{\mathrm{geo},it,m}(\theta) = \exp\left[ -i\frac{\theta}{2} \widehat{\mathbf{n}}_{\mathrm{geo},it,m} \cdot \boldsymbol{\sigma} \right].
\label{eq:geodesic_operator}
\end{equation}
For compactness, let
\begin{equation}
\delta_{it,m} = \min \left\{ \theta_{it,m}, \phi_{it,m} \right\}
\label{eq:geodesic_applied_angle}
\end{equation}
denote the safeguarded geodesic angle.
The full model-control operation is:
\begin{equation}
\begin{aligned}
\ket{\psi_{it,m+1}} &= \mathcal{C}_{it,m}
\left( \ket{\psi^{*}_{it,m+1}}, \ket{\psi_{it,m}} \right)
\\[1ex]
&= \begin{cases} \ket{\psi^{*}_{it,m+1}},
&
\substack{ \mathfrak{a}(c_{it,m})=z \\ \text{or }\theta_{it,m}=0, }
\\[2.5ex]
\ket{\psi^{*}_{it,m+1}}, & \substack{ \mathfrak{a}(c_{it,m})=x,
\\ p^{D,*}_{it,m+1} \geq p^{D}_{it,m}, } \\[2.5ex] 
\ket{\psi_{it,m}}, & \substack{ \mathfrak{a}(c_{it,m})=x, \\
\phi_{it,m}=0, } \\[2.5ex] 
\mathcal{U}_{\mathrm{geo},it,m}
\left( \delta_{it,m} \right) \ket{\psi_{it,m}}, & \substack{
\mathfrak{a}(c_{it,m})=x,\;
\theta_{it,m}>0, \\
\phi_{it,m}>0,\; p^{D,*}_{it,m+1}<p^{D}_{it,m}. }
\end{cases}
\end{aligned}
\label{eq:geodesic_control_operation}
\end{equation}

\newpage
\noindent
The originally assigned Pauli rotation is retained when it is a phase-changing cue, a zero angle update, or a state-changing candidate that preserves or increases the defensive state probability. 
When an $x$-axis candidate might reduce the defensive state probability, it is
redirected along the geodesic toward the defensive pole. The applied angle $\delta_{it,m}$ prevents the redirected update from passing beyond the target pole.
\begin{proposition}
\label{prop:geodesic_safeguard}
When an $x$-axis candidate update would reduce the defensive state
probability, the geodesic safeguard replaces it with an update directed toward the defensive pole. The redirected update cannot pass beyond the defensive pole and preserves the non-decreasing cue-update property.
\end{proposition}
The proposition concerns updates that activate the safeguard. A candidate that keeps the defensive state probability is accepted by
the implemented monotonicity rule, the model does not claim that every nonzero cue necessarily produces a strict increase.
This result applies to the cue-driven update before relaxation. The
post relaxation state reflects the balance between the defensive movement generated by the cue sequence and the subsequent pull toward the baseline generated by relaxation.
A proof of Proposition~\ref{prop:geodesic_safeguard} is provided in
Appendix~\ref{app:formal_guarantees}. 
The Bloch vector representation and the general rotation about an arbitrary axis used in the safeguard are reviewed in Appendices~\ref{sec:bloch} and~\ref{app:rotations}, respectively.
The safeguard is activated only when the originally assigned $x$-axis Pauli rotation would reduce the defensive state probability. In all other updates, the originally assigned $x$- or $z$-axis rotation is retained.
The non-commutativity and cue order effects introduced in
Eqs.~\eqref{eq:pauli_noncommutativity} and \eqref{eq:rotation_order_dependence} remain part of the state evolution process. The geodesic safeguard acts as a directional correction that prevents rotational overshoot and redirects probability reducing updates to the defensive pole.

% relaxaion
\subsubsection{Relaxation Toward the Baseline State}
\label{sec:relaxation_step}

The cue-induced rotations represent the effect of currently perceived information on the driver's latent mental state. 
For cue types interpreted as threats, these updates are designed to preserve or increase the defensive state probability.  Without an opposing recovery mechanism, the effects of earlier threats could remain indefinitely in the
carried-forward state.
The relaxation mechanism is applied after the complete cue sequence at every trajectory sample. It pulls the resulting latent mental state toward its initial baseline. During strong threat exposure, the state-changing cue updates may outweigh this pull. When the cues weaken or disappear, relaxation becomes the dominant influence and allows the effects of earlier threats to fade.
For example, after a nearby vehicle moves farther away and the conflict begins to resolve, repeated relaxation moves the driver's latent mental state toward its baseline rather than allowing the earlier interaction to maintain a persistent defensive bias.

After all $M_{it}$ cues at trajectory sample $t$ have been processed, the
pre-relaxation state is given by:
\begin{equation}
\ket{\widetilde{\psi}_{it}} = \ket{\psi_{it,M_{it}}}.
\label{eq:pre_relaxation_state}
\end{equation}
The relaxation operation moves this state toward the driver's initial baseline state $\ket{\psi_{i0,0}}$. It is defined as the normalized interpolation:
\begin{equation}
\mathcal{R}_{\lambda} \left(
\ket{\widetilde{\psi}_{it}}, \ket{\psi_{i0,0}}
\right) = \frac{ (1-\lambda)\ket{\widetilde{\psi}_{it}}
+ \lambda\ket{\psi_{i0,0}} }{
\left\| (1-\lambda)\ket{\widetilde{\psi}_{it}}
+
\lambda\ket{\psi_{i0,0}} \right\| },
\qquad
0\leq\lambda<1.
\label{eq:relaxation_operator}
\end{equation}
The final latent mental state at trajectory sample $t$ is:
\begin{equation}
\ket{\psi_{it}} = \mathcal{R}_{\lambda}
\left( \ket{\widetilde{\psi}_{it}}, \ket{\psi_{i0,0}} \right).
\label{eq:relax_step}
\end{equation}

\newpage
\noindent
The parameter $\lambda$ is an interpolation weight applied before normalization. It should not be interpreted as the exact percentage reduction in the defensive state probability, the polar angle, or the Bloch sphere distance from the baseline. Larger values of $\lambda$ produce a stronger pull toward the initial baseline, but the resulting change in state probability also depends on the current complex amplitudes.
The normalized interpolation is evaluated using the recursively generated complex amplitudes under the fixed phase convention produced by the initialization and ordered state-update operations. It is a model specific relaxation relaxation mechanism.

Figure~\ref{fig:qscm_circuit} shows a circuit style representation of
the recursive Q-SCM state evolution process. 
Step~2 is not shown as a state transformation block because cue encoding is a classical preprocessing operation that supplies the rotation angle $\theta_{it,m}$ and axis assignment $\mathfrak{a}(c_{it,m})$ to the Step~3 update operations.
The initialization operator is applied once at $t=0$. The cue-induced rotation and model-control operations are then repeated for all $M_{it}$ cues at trajectory sample $t$. After all cues have been processed, Step~4 applies relaxation and carries the resulting state forward. 
Step~5 extracts the state membership probabilities using the Born rule.
The diagram is a circuit style computational representation not a unitary quantum circuit. 
The cue-induced rotation operators are unitary, but the model control operation $\mathcal{C}_{it,m}$ may redirect a candidate update and the relaxation operation $\mathcal{R}_{\lambda}$ is non-unitary. 
The Born rule operation $\mathcal{B}$ extracts the state membership probabilities without replacing or collapsing the final state that is thern carried
forward.

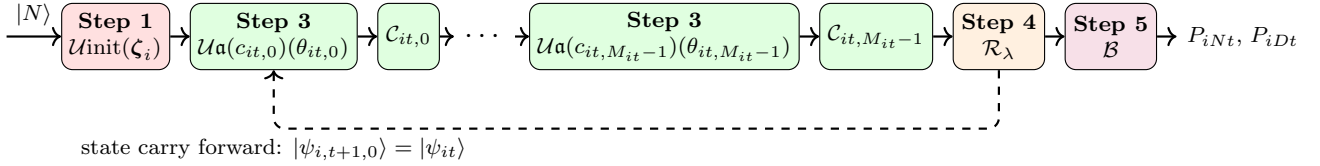
\begin{figure}[!htbp]
\centering

\resizebox{0.98\linewidth}{!}{%
\begin{tikzpicture}[
    node distance=0.25cm,
    circuitbox/.style={
        rectangle,
        rounded corners,
        draw=black,
        minimum height=0.85cm,
        align=center,
        inner xsep=3pt,
        inner ysep=2pt,
        font=\scriptsize
    },
    arrow/.style={->,thick},
    feedback/.style={->,thick,dashed,rounded corners}
]

\node[circuitbox,fill=red!12] (init)
{\textbf{Step 1}\\
$\mathcal{U}{\mathrm{init}}(\boldsymbol{\zeta}_i)$};

\node[circuitbox,fill=green!12,right=of init] (update1)
{\textbf{Step 3}\\
$\mathcal{U}{\mathfrak{a}(c_{it,0})}(\theta_{it,0})$};

\node[circuitbox,fill=green!12,right=of update1] (control1)
{$\mathcal{C}_{it,0}$};

\node[right=0.18cm of control1,font=\normalsize] (dots)
{$\cdots$};

\node[circuitbox,fill=green!12,right=0.18cm of dots] (updatelast)
{\textbf{Step 3}\\
$\mathcal{U}{\mathfrak{a}(c_{it,M_{it}-1})}
(\theta_{it,M_{it}-1})$};

\node[circuitbox,fill=green!12,right=of updatelast] (controllast)
{$\mathcal{C}_{it,M_{it}-1}$};

\node[circuitbox,fill=orange!12,right=of controllast] (relax)
{\textbf{Step 4}\\
$\mathcal{R}_{\lambda}$};

\node[circuitbox,fill=purple!12,right=of relax] (born)
{\textbf{Step 5}\\
$\mathcal{B}$};

\node[right=0.25cm of born,align=center,font=\scriptsize] (output)
{$P_{iNt},\,P_{iDt}$};

\draw[arrow]
([xshift=-0.70cm]init.west)
--
node[above,font=\scriptsize]
{$\ket{N}$}
(init.west);

\draw[arrow] (init) -- (update1);
\draw[arrow] (update1) -- (control1);
\draw[arrow] (control1) -- (dots);
\draw[arrow] (dots) -- (updatelast);
\draw[arrow] (updatelast) -- (controllast);
\draw[arrow] (controllast) -- (relax);
\draw[arrow] (relax) -- (born);
\draw[arrow] (born) -- (output);

\draw[feedback]
(relax.south)
-- ++(0,-0.75cm)
-| node[
    pos=0.52,
    below,
    align=center,
    font=\scriptsize
]
{state carry forward:
$\ket{\psi_{i,t+1,0}}=\ket{\psi_{it}}$}
(update1.south);

\end{tikzpicture}%
}

\caption{Circuit representation of the Q-SCM
state evolution process. }
\label{fig:qscm_circuit}
\end{figure}
In Figure~\ref{fig:qscm_circuit}, $\mathcal{B}$ is the Born rule
probability extraction operation: 
\begin{equation}
\mathcal{B} \left( \ket{\psi_{it}} \right)
= \left( \left| \braket{N}{\psi_{it}} \right|^2, \left| \braket{D}{\psi_{it}} \right|^2 \right) = \left( P_{iNt}, P_{iDt} \right).
\label{eq:born_rule_operation}
\end{equation}
The extracted probabilities enter the classical choice layer but do not
replace the final state $\ket{\psi_{it}}$. 
The same final state is carried
forward through Eq.~\eqref{eq:architecture_carry_forward}.

\subsection{Likelihood and Estimation}
\label{sec:likelihood_estimation}

\subsubsection{Objective Function}
\label{sec:likelihood}

Let $y_{ijt}=1$ if driver $i$ selects manoeuvre $j$ at trajectory sample $t$, and let $y_{ijt}=0$ otherwise. 
Conditional on latent mental state $s\in\mathcal{S}$, the state-specific  manoeuvre probability is represented using a multinomial logit model:
\begin{equation}
P_{ijt\mid s} =\frac{ \exp\left(V_{ijt\mid s}\right) }{ \displaystyle \sum_{r\in\mathcal{A}} \exp\left(V_{irt\mid s}\right) }, \qquad
s\in\{N,D\},
\label{eq:general_state_conditional_mnl}
\end{equation}
where $V_{ijt\mid s}$ is the systematic utility of manoeuvre $j$ conditional on latent mental state $s$. The empirical specification of these utilities is provided in Section~\ref{sec:application}.
The unconditional probability of observing manoeuvre $j$ is obtained by mixing the state-specific probabilities using the Q-SCM state membership probabilities:
\begin{equation}
P_{ijt} = P_{iNt}P_{ijt\mid N} + P_{iDt}P_{ijt\mid D}.
\label{eq:estimated_manoeuvre_probability}
\end{equation}

Let $\boldsymbol{q}$ include the parameters governing the quantum state membership layer, including initialization, cue encoding, rotation, and relaxation parameters. 

\newpage
\noindent
Let $\boldsymbol{c}$ contain the parameters of the state-conditional classical choice layer, and define:
\begin{equation}
\boldsymbol{\Theta} = \left( \boldsymbol{q}, \boldsymbol{c} \right).
\label{eq:full_parameter_vector}
\end{equation}
The sample log-likelihood is:
\begin{equation}
\ell
\left(
\boldsymbol{\Theta}
\right) = \sum_{i=1}^{I}
\sum_{t=0}^{T_i}
\sum_{j\in\mathcal{A}}
y_{ijt}
\log
P_{ijt}
\left(
\boldsymbol{\Theta}
\right).
\label{eq:qscm_log_likelihood}
\end{equation}

Since one manoeuvre is observed at each trajectory sample, the same log-likelihood can equivalently be written as:
\begin{equation}
\ell \left( \boldsymbol{\Theta} \right)
=
\sum_{i=1}^{I}
\sum_{t=0}^{T_i}
\log
P_{i,a_{it},t}
\left( \boldsymbol{\Theta} \right),
\label{eq:qscm_observed_log_likelihood}
\end{equation}
where $a_{it}\in\mathcal{A}$ is the manoeuvre observed for driver
$i$ at trajectory sample $t$.

The state membership probabilities entering Eq.~\eqref{eq:estimated_manoeuvre_probability} must be evaluated recursively and in temporal order for each driver. 
Specifically, the state is initialized once, updated by the ordered cue sequence, controlled by the monotonicity and geodesic mechanisms, relaxed, and then
carried forward. Consequently, $P_{iNt}$ and $P_{iDt}$ depend on the driver's cue history up to trajectory sample $t$, rather than only on the covariates observed at that sample.

\subsubsection{Estimation Procedure}
\label{sec:estimation}

The sample log-likelihood in Eq.~\eqref{eq:qscm_observed_log_likelihood} provides the data fit component of the estimation objective. Let $\Omega$ denote the feasible parameter space defined by the admissible ranges of the initialization, cue response, rotation, relaxation, and choice parameters.
The likelihood may be optimized directly or combined with an application-specific regularization term. In the empirical application
of Section~\ref{sec:application}, the Q-SCM parameters are obtained from a penalized objective that includes an $\ell_2$ penalty on selected defensive state utility coefficients. The exact penalized objective and regularization weight are specified in the application section. The constraints in $\Omega$ enforce the admissible ranges of the initialization, rotation, cue response and relaxation parameters.  For each trial parameter vector, the likelihood is evaluated driver by
driver. The computational sequence is:
\begin{enumerate}[label=\arabic*.,leftmargin=2em]
\item Construct the initial latent mental state $\ket{\psi_{i0,0}}$.
\item For each trajectory sample $t$, encode the observed cues and obtain their rotation angles and axis assignments.
\item Apply the cue-induced rotation and model control operation sequentially for $m=0,\ldots,M_{it}-1$.
\item Apply relaxation operation to get $\ket{\psi_{it}}$ and extract $P_{iNt}$ and $P_{iDt}$ using the Born rule.
\item Compute the state-conditional and unconditional manoeuvre probabilities and accumulate their log-likelihood contributions.
\item Carry $\ket{\psi_{it}}$ forward to initialize trajectory sample $t+1$.
\end{enumerate}

Since the state recursion makes the likelihood highly nonlinear and potentially non-concave, estimation should use parameter bounds and multiple suitable starting values. 
In the empirical application, the continuous parameters are estimated using the bounded L-BFGS-B algorithm. A coarse-to-fine grid search is used to obtain the starting value of the parameter for which the likelihood displayed the greatest sensitivity, after which all parameters are estimated jointly. 
The case specific parameter vector, starting values, bounds are reported in Section~\ref{sec:application}.

\newpage
Approximate parameter uncertainty is evaluated using the Hessian of the unpenalized log-likelihood, evaluated at the penalized estimate
$\widehat{\boldsymbol{\Theta}}_{\mathrm{pen}}$, along with a
driver-level cluster-robust sandwich covariance estimator. Approximate Wald statistics are calculated from the resulting standard errors. 
These quantities provide local uncertainty summaries around the penalized
solution and should not be interpreted as exact conventional
maximum-likelihood inference.
Let: 
\begin{equation}
\ell_i \left( \boldsymbol{\Theta}
\right) = \sum_{t=0}^{T_i}
\log P_{i,a_{it},t}
\left( \boldsymbol{\Theta} \right)
\label{eq:driver_log_likelihood}
\end{equation}
is the contribution of driver $i$, and define the driver-level score:
\begin{equation}
\mathbf{g}_i = \nabla_{\boldsymbol{\Theta}}
\ell_i\left(\widehat{\boldsymbol{\Theta}}_{\mathrm{pen}}\right).
\label{eq:driver_cluster_score}
\end{equation}
Let:
\begin{equation}
\mathbf{H} = - \nabla_{\boldsymbol{\Theta}}^{2}
\ell\left(\widehat{\boldsymbol{\Theta}}_{\mathrm{pen}}\right)
\label{eq:observed_information_matrix}
\end{equation}
is the observed information matrix. The driver-clustered covariance estimator is:
\begin{equation}
\widehat{\operatorname{Var}} \left( \widehat{\boldsymbol{\Theta}}
\right) = \mathbf{H}^{-1}
\left( \sum_{i=1}^{I} \mathbf{g}_i\mathbf{g}_i^{\mathsf T} \right)
\mathbf{H}^{-1}.
\label{eq:driver_cluster_covariance}
\end{equation}
These standard errors provide local uncertainty summaries around the reported penalized solution. They should not be interpreted as exact conventional maximum likelihood standard errors for an unpenalized estimator.
Clustering at the driver level allows the trajectory observations of the same driver to be statistically dependent and treating different drivers as independent clusters. Standard errors are obtained from the square roots of the diagonal elements of Eq.~\eqref{eq:driver_cluster_covariance}.

\section{Application to Naturalistic Trajectories }
\label{sec:application}
\subsection{Dataset}
\label{sec:dataset}

The rounD (Roundabout Drone Dataset) \cite{Krajewski2020rounD}, is employed in this paper which is a publicly available naturalistic trajectory dataset collected using drones at roundabouts in Germany. 
The focus is on the Neuweiler roundabout, which has the highest number of recordings in the dataset. 
We retain passenger car trajectories that interact with motorized vehicles (buses, trucks, trailers, vans, motorcycles) and exclude non motorized road users to avoid mixing fundamentally different dynamics and risk perceptions. 
Only trajectories that enter, circulate, and exit the roundabout are included. 
The channelized right turns are excluded from the analysis. The trajectories are downsampled from 0.04\,s to 1\,s intervals. 
This yields a final dataset of 85{,}754 observations from 9{,}610 passenger car drivers.
One observation corresponds to one passenger car driver's position and associated kinematic state at one 1\,s time step along the driver's trajectory.
% At each 1\,s time step, three actions are defined for the discrete choice set: decelerate ($a=1$), maintain speed ($a=2$), and accelerate ($a=3$). 
% These are determined from change in speed at each next time step. Then this value is mapped onto the choice set grid which is comprised of four speed regimes that ... XXX 
% The cue variables used to drive the quantum state evolution and the choice model are as follows. 
% The separation distance and closing time-to-collision (CTTC) serve as longitudinal threat cues assigned to the $\sigma_x$ axis. 
% The distance from the maintain-speed alternative cell to the ego vehicle's intended path serves as a contextual cue assigned to the $\sigma_z$ axis, and acting as a phase modulating signal. 
% Additional covariates enter the MNL choice layer are the STCP XXX, 
% overlap with other vehicles' projected paths, distance to the opposing path, and XX.
At each 1\,s time step, three actions are defined for the discrete choice set: decelerate ($a=1$), maintain speed ($a=2$), and accelerate ($a=3$). 
The action label is assigned from the change in speed between $t$ and $t+1$ and mapped to a three alternative choice set~\cite{alhaideri2026dualState}. The observed action distribution is dominated by deceleration: $58{,}993$ observations ($68.8\%$) are decelerations, $12{,}638$ ($14.7\%$) are maintain speed actions, and $14{,}123$ ($16.5\%$) are accelerations. 
We believe that the skewness toward deceleration is in line with the prevalence of yielding and gap acceptance manoeuvres at a roundabout.

The cue variables that drive the quantum state evolution are summarized as follows. 
The separation distance ($\mathrm{mDist}_t$) and CTTC ($\mathrm{CTTC}_t$) are the state changing cues assigned to the $\sigma_x$ axis.
The distance from the ego vehicle's current position to its intended path is the phase changing cue assigned to the $\sigma_z$ axis. 
Both state changing cues are equal to zero in the absence of an interacting vehicle.
$\mathrm{mDist}_t$ is valid in $72.2\%$ of observations.
$\mathrm{CTTC}_t$ in $46.4\%$, reflecting that drivers spend substantial portions of their trajectory without a relevant conflict partner. 

\newpage
\noindent
The lane deviation cue is defined for every observation. The MNL choice layer takes a per alternative distance from each candidate action cell to the ego's intended path ($d_{j,t}$), a per alternative non overlap indicator with interacting vehicles' projected polygons ($o_{j,t}$), and the frontal and rear directional intensities $f_t$ and $r_t$ defined in Table~\ref{tab:variables}.

\subsection{Model Specification}
\label{sec:model_spec}
To simplify notation in this application section, the driver index $i$
is suppressed. All cue variables, latent states, and probabilities remain
driver-specific, and the state recursion is restarted separately for each
driver.
We formalize the two models that we will estimate and compare. 
Both models share the same set of input variables, the same set of alternatives, and the same within-state utility specification. 
They differ only in how the latent state probabilities $P(N_t)$ and $P(D_t)$ are generated as functions of the cue history. 
In this subsection, we introduce the variables (Table~\ref{tab:variables}), the common within-state utility, and the two distinct class membership mechanisms together with their parameter sets (Tables~\ref{tab:choice_params}--\ref{tab:quantum_params}).

% Each driver contributes a trajectory of $T$ observations indexed by $t$, consistent with the indexing of Section~\ref{sec:state_evolution}.
% At every $t$ the choice set is $\mathcal{A}=\{1,2,3\}$ with $j=1$ denoting decelerate, $j=2$ maintain speed, and $j=3$ accelerate. The variables that enter both models are summarized in Table~\ref{tab:variables}.

% \subsubsection{Common Within-State Utility}

In both models, the conditional choice probability given latent state
$s\in\{N,D\}$ is specified using a multinomial logit model:
\begin{equation}
P_{jt\mid s} \equiv P(a_t=j\mid s) = \frac{ \exp\left(V_{jt\mid s}\right)
}{ \displaystyle \sum_{h\in\mathcal{A}} \exp\left(V_{ht\mid s}\right) },
\qquad j\in\mathcal{A}, \quad s\in\{N,D\}.
\label{eq:mnl}
\end{equation}
The maintain speed alternative, $j=2$, is the reference alternative for the alternative-specific constants within each latent state. Its alternative-varying attributes remain in the systematic utility.
The state-specific systematic utilities are:
\begin{align}
V_{jt\mid N} &= \beta^{N}_{\mathrm{decel}}\mathbf{1}[j=1]
+ \beta^{N}_{\mathrm{accel}}\mathbf{1}[j=3] + \beta^{N}_{\mathrm{dist\_path}}d_{j,t} \nonumber\\
&\quad +
\beta^{N}_{\mathrm{nonoverlap}}o_{j,t} + \beta^{N}_{\mathrm{int}}
\left( f_t\mathbf{1}[j=1] + r_t\mathbf{1}[j=3]
\right), 
\label{eq:V_N}
\\[1ex]
V_{jt\mid D}
&= \beta^{D}_{\mathrm{decel}}\mathbf{1}[j=1] + \beta^{D}_{\mathrm{accel}}\mathbf{1}[j=3] + beta^{D}_{\mathrm{nonoverlap}}o_{j,t}
\nonumber\\
&\quad
+ \beta^{D}_{\mathrm{int}} \left( f_t\mathbf{1}[j=1] + r_t\mathbf{1}[j=3] \right).
\label{eq:V_D}
\end{align}

The intensity coefficient $\beta^{s}_{\mathrm{int}}$ within each state governs both the deceleration response to frontal threats and the acceleration response to rear threats; a single shared coefficient per state is imposed for parsimony.
Eqs.~\eqref{eq:V_N}--\eqref{eq:V_D} yield nine choice parameters
$\boldsymbol{c}=(\beta^{N}_{\mathrm{decel}},\beta^{N}_{\mathrm{accel}},\beta^{N}_{\mathrm{dist\_path}},\beta^{N}_{\mathrm{nonoverlap}},\beta^{N}_{\mathrm{int}},\beta^{D}_{\mathrm{decel}},\beta^{D}_{\mathrm{accel}},\beta^{D}_{\mathrm{nonoverlap}},\beta^{D}_{\mathrm{int}})\in\mathbb{R}^{9}$. These parameters are estimated under both models and are summarized in Table~\ref{tab:choice_params}.

\begin{table}[!htbp]
\centering
\caption{Variables used in the Q-SCM and the classical latent class specifications. 
}
\label{tab:variables}
\footnotesize
\begin{tabular}{p{0.18\linewidth} p{0.18\linewidth} p{0.55\linewidth}}
\toprule
Symbol & Type & Description \\
\midrule
$a_t$                 & Observed action      & Driver's discrete choice at $t$; $a_t\in\{1,2,3\}$. \\
$\mathrm{CTTC}_t$     & Observation-level    & Minimum value of closing time-to-collision at $t$ (s). \\
$\mathrm{mDist}_t$    & Observation-level    & Separation distance to the nearest interacting vehicle at $t$ (m). \\
$d^{\mathrm{ego}}_{\mathrm{path},t}$ & Observation-level & Distance from the ego vehicle's current position to its intended path (m); used to construct the phase-changing deviation cue. \\
$s_{\mathrm{dist},t}$ & Encoded cue & Dimensionless distance cue defined in Eq.~\eqref{eq:distance_cue}; larger values indicate a smaller separation distance. \\
$s_{\mathrm{cttc},t}$ & Encoded cue & Dimensionless CTTC-based cue defined in Eq.~\eqref{eq:cttc_cue}; larger values indicate a more immediate closing interaction. \\
$s_{\mathrm{dev},t}$ & Encoded cue & Dimensionless deviation cue defined in Eq.~\eqref{eq:deviation_cue}; larger values indicate greater deviation from the intended path. \\

$f_t$ & Observation-level & Frontal proximity of the minimum CTTC vehicle, on a $0$--$1$ scale: 
$1$ when it is directly ahead, fading to $0$ toward the sides and behind. 
Enters the utility of the decelerate alternative. \\

$r_t$ & Observation-level & Rear proximity of the minimum CTTC vehicle, on a $0$--$1$ scale: 
$1$ when it is directly behind, fading to $0$ toward the sides and in front. 
Enters the utility of the accelerate. \\

$o_{j,t}$             & Per-alternative      & Non overlap indicator: $1$ if the polygon of alternative $j$ does not overlap with any interacting vehicle, $0$ otherwise. \\

$d_{j,t}$             & Per-alternative      & Distance from alternative $j$'s choice centroid to the ego vehicle's intended path (m). \\
\bottomrule
\end{tabular}
\end{table}

\begin{table}[!htbp]
\centering
\caption{Nine within-state choice utility parameters, common to both models.}
\label{tab:choice_params}
\footnotesize
\begin{tabular}{p{0.27\linewidth} p{0.65\linewidth}}
\toprule
Parameter & Interpretation \\
\midrule
$\beta^{N}_{\mathrm{decel}}$        & Alternative-specific constant for decelerate in the neutral state. \\
$\beta^{N}_{\mathrm{accel}}$        & Alternative-specific constant for accelerate in the neutral state. \\
$\beta^{N}_{\mathrm{dist\_path}}$   & Coefficient on per-alternative distance to intended path, neutral state. \\
$\beta^{N}_{\mathrm{nonoverlap}}$      & Coefficient on per-alternative non overlap indicator, neutral state. \\
$\beta^{N}_{\mathrm{int}}$          & Intensity coefficient (frontal$\!\to\!$decel, rear$\!\to\!$accel), neutral state. \\
$\beta^{D}_{\mathrm{decel}}$        & Alternative-specific constant for decelerate in the defensive state. \\
$\beta^{D}_{\mathrm{accel}}$        & Alternative-specific constant for accelerate in the defensive state. \\
$\beta^{D}_{\mathrm{nonoverlap}}$      & Coefficient on per-alternative non overlap indicator, defensive state. \\
$\beta^{D}_{\mathrm{int}}$          & Intensity coefficient (frontal$\!\to\!$decel, rear$\!\to\!$accel), defensive state. \\
\bottomrule
\end{tabular}
\end{table}

\subsubsection{Classical Latent Class Model Specification}

The classical latent class baseline specifies the class probabilities as a logit formulation using the contemporaneous CTTC:
\begin{equation}
W^{\mathrm{cl}}_{D,t} \;=\; \frac{\gamma_1}{1+\mathrm{CTTC}_t} + \gamma_2,\qquad W^{\mathrm{cl}}_{N,t} \;=\; 0,
\label{eq:classical_class}
\end{equation}
\begin{equation}
P^{\mathrm{cl}}(D_t) \;=\; \frac{\exp(W^{\mathrm{cl}}_{D,t})}{1+\exp(W^{\mathrm{cl}}_{D,t})},\qquad P^{\mathrm{cl}}(N_t)\;=\;1-P^{\mathrm{cl}}(D_t).
\label{eq:classical_softmax}
\end{equation}
The classical model thus introduces two additional parameters $\gamma_1$ and $\gamma_2$ beyond the nine choice parameters in Eq.~\eqref{eq:V_N}--\eqref{eq:V_D}, with a total of eleven free parameters. 
A similar specification of the classical model is described in detail in \cite{alhaideri2026dualState}.

% our QSCM
\subsubsection{Q-SCM Model Specification}

All cue variables, latent states, and probabilities remain driver-specific, and the state recursion is initialized separately for each driver.
The Q-SCM generates the statemembership probabilities $P^{\mathrm{Q}}(N_t)$ and $P^{\mathrm{Q}}(D_t)$ by the recursive state evolution framework introduced in Section~\ref{sec:state_evolution}. The application uses the cue set: 
\begin{equation}
\mathcal{K}_{\mathrm{app}} = \left\{ \mathrm{dist}, \mathrm{cttc}, \mathrm{dev} \right\}.
\label{eq:application_cue_set}
\end{equation}
Let $d_0=1\,\mathrm{m}$ and $t_0=1\,\mathrm{s}$ denote fixed scaling constants. The encoded cues are:
\begin{subequations}
\label{eq:application_cue_encoding}
\begin{align}
s_{\mathrm{dist},t}
&=
\begin{cases}
\dfrac{d_0}{d_0+\mathrm{mDist}_t},
& \text{when a separation distance is available},
\\[2ex]
0, &
\text{otherwise},
\end{cases}
\label{eq:distance_cue}
\\[2ex]
s_{\mathrm{cttc},t} &= \begin{cases}
\dfrac{t_0}{t_0+\mathrm{CTTC}_t},
& \text{when a valid closing interaction is available},
\\[2ex] 0,
& \text{otherwise},
\end{cases}
\label{eq:cttc_cue}
\\[2ex] s_{\mathrm{dev},t}
&= \frac{
d^{\mathrm{ego}}_{\mathrm{path},t} }{d_0 }.
\label{eq:deviation_cue}
\end{align}
\end{subequations}
Here, $d^{\mathrm{ego}}_{\mathrm{path},t}$ is the distance from the ego vehicle's current position to its intended path. It is distinct from $d_{j,t}$, which is the alternative-specific distance from the centroid of candidate manoeuvre choice $j$ to the intended path.
Each encoded cue is mapped to a bounded saturation factor and rotation given by:
\begin{equation}
\tau_{c,t} = \frac{ s_{c,t} }{ k_c+s_{c,t} }, \qquad \theta_{c,t} = \theta_{\max,c}\tau_{c,t}, \qquad c\in\mathcal{K}_{\mathrm{app}}.
\label{eq:application_rotation_mapping}
\end{equation}
Here, $k_c>0$ is the cue-specific saturation (sensitivty) parameter and $\theta_{\max,c}$ is the maximum rotation angle for cue type $c$. The factor $\tau_{c,t}$ equals zero when the encoded cue is zero and approaches one as the encoded cue increases. 
Equation \eqref{eq:application_rotation_mapping} is the case-specific form of the general response function $g_c(\cdot)$ introduced in the Section~\ref{sec:framework}.

The application uses the current encoded cue level $s_{c,t}$ at each trajectory sample rather than the change in cue level between consecutive samples. 
This is an application-specific modelling assumption. It allows a traffic condition that remains present, such as continued close proximity or a persistently short CTTC, to continue influencing the latent state even when the cue changes only slightly between consecutive observations. 
A change-based specification would instead emphasize abrupt
cue variations and would assign little additional influence to a persistent but nearly constant condition. Recovery after a threat weakens is represented separately through the relaxation mechanism.

The general initialization mechanism is specified in this application as a population-level $y$-axis rotation by:
\begin{equation}
\ket{\psi_{0,0}} = \mathcal{U}_y(\theta_0)\ket{N} = \cos\left( \frac{\theta_0}{2} \right)\ket{N} + \sin\left( \frac{\theta_0}{2} \right)\ket{D}, \qquad \theta_0\in[0,\pi].
\label{eq:application_initial_state}
\end{equation}
The parameter $\theta_0$ determines the common baseline neutral--defensive balance before the first observed cue is processed.

The common initialization is a population-level modelling assumption.
All drivers begin from the same estimated baseline state and share the same state-evolution and choice parameters. This does not imply that drivers retain identical latent states during their observed episodes.
After initialization, their state trajectories can diverge because they experience different cue values and cue histories. Thus, the present specification represents within-driver temporal variation and history-dependent heterogeneity, but it does not separately estimate driver-specific random coefficients or initial states.

At every trajectory sample, three cues are processed, so that $M_t=3$. The fixed within-sample cue sequence is:
\begin{equation}
c_{t,0} = \mathrm{dist}, \qquad c_{t,1} = \mathrm{cttc}, \qquad c_{t,2} = \mathrm{dev}. 
\label{eq:application_cue_order}
\end{equation}
The corresponding axis assignment is:
\begin{equation}
\mathfrak{a}(\mathrm{dist})=x, \qquad \mathfrak{a}(\mathrm{cttc})=x, \qquad \mathfrak{a}(\mathrm{dev})=z.
\label{eq:application_axis_assignment}
\end{equation}
The cue-to-axis assignments in Eq.~\eqref{eq:application_axis_assignment}, together with the fixed distance--CTTC--deviation processing order in
Eq.~\eqref{eq:application_cue_order}, are application-specific modelling assumptions. They are not intended to imply that every driver consciously
or universally processes perceptual information in this exact discrete
sequence. Rather, the ordering provides a parsimonious and behaviourally
interpretable representation of how different types of traffic
information may enter the driver's evolving appraisal.

One possible behavioural interpretation is that separation distance is processed first as an immediate spatial proximity cue indicating the presence and closeness of an interaction. CTTC is then processed as a dynamic urgency cue that evaluates how rapidly that spatial interaction is developing. Finally, deviation from the intended path is processed as
a contextual cue reflecting the driver's current path position and manoeuvring context. Under the adopted axis assignment, the distance and CTTC cues act as state-changing signals through $x$-axis rotations.
The deviation cue acts as a phase-changing signal by a $z$-axis rotation. The latter does not directly change the
pre-relaxation state membership probabilities, but it can modify the context carried by the complex amplitudes and thereby influence the post relaxation state and subsequent cue responses.

Applying the same fixed sequence at every trajectory sample also provides a consistent recursive specification across drivers and observations. 
Although the underlying distance and CTTC Pauli rotations share the same $x$-axis and commute when considered as uncontrolled rotations, their controlled updates need not commute because the monotonicity and geodesic mechanisms are evaluated after each intermediate state. Consequently,
the distance--CTTC--deviation sequence is an explicit component of the empirical Q-SCM specification. Alternative cue assignments and processing orders may be examined in future applications.

The cue-induced state evolution is: 
\begin{align}
\ket{\psi^{*}_{t,1}} &=
\mathcal{U}_x \left( \theta_{\mathrm{dist},t} \right) \ket{\psi_{t,0}},
&
\ket{\psi_{t,1}} &= \mathcal{C}_{t,0} \left( \ket{\psi^{*}_{t,1}}, \ket{\psi_{t,0}} \right),
\label{eq:application_distance_update}
\\[1ex]
\ket{\psi^{*}_{t,2}} &= \mathcal{U}_x \left( \theta_{\mathrm{cttc},t} \right) \ket{\psi_{t,1}},
&
\ket{\psi_{t,2}} &= \mathcal{C}_{t,1} \left( \ket{\psi^{*}_{t,2}}, \ket{\psi_{t,1}} \right),
\label{eq:application_cttc_update}
\\[1ex]
\ket{\psi^{*}_{t,3}} &= \mathcal{U}_z \left( \theta_{\mathrm{dev},t} \right) \ket{\psi_{t,2}}, & \ket{\psi_{t,3}}
&= \mathcal{C}_{t,2} \left( \ket{\psi^{*}_{t,3}}, \ket{\psi_{t,2}} \right).
\label{eq:application_deviation_update}
\end{align}

The deviation cue is processed last as a phase-changing cue. Its $z$-axis rotation leaves the pre-relaxation state membership probabilities unchanged. However, because relaxation is applied immediately after the three cue sequence and acts on the complex amplitudes, the phase generated by the deviation cue can influence the post relaxation probabilities at the current trajectory sample. The phase is also retained through state carry forward and can influence subsequent state-changing updates.
The pre-relaxation state is:
\begin{equation}
\ket{\widetilde{\psi}_t} = \ket{\psi_{t,3}}.
\label{eq:application_pre_relaxation_state}
\end{equation}
The final latent mental state at trajectory sample $t$ is:
\begin{equation}
\ket{\psi_t} = \mathcal{R}_{\lambda} \left( \ket{\widetilde{\psi}_t}, \ket{\psi_{0,0}} \right),
\label{eq:application_relaxed_state}
\end{equation}
and is carried forward according to:
\begin{equation}
\ket{\psi_{t+1,0}} = \ket{\psi_t}. 
\label{eq:application_carry_forward}
\end{equation}

The Q-SCM state membership probabilities are then extracted using the Born rule:
\begin{equation}
P^{\mathrm{Q}}(N_t) = \left| \braket{N}{\psi_t} \right|^2, \qquad P^{\mathrm{Q}}(D_t) = \left| \braket{D}{\psi_t} \right|^2, \qquad P^{\mathrm{Q}}(N_t) + P^{\mathrm{Q}}(D_t) = 1.
\label{eq:application_state_probabilities}
\end{equation}
The Q-SCM includes eight free parameters beyond the nine choice parameters of Table~\ref{tab:choice_params}.
These are the three maximum rotation angles $\theta_{\max,c}$, the three saturation parameters $k_c$ for
$c\in\{\mathrm{dist},\mathrm{cttc},\mathrm{dev}\}$, the population-level initialization angle $\theta_0$, and the relaxation  weight $\lambda$.
The parameters are summarized in Table~\ref{tab:quantum_params}. The complete Q-SCM specification therefore contains seventeen estimated parameters.

\begin{table}[!htbp]
\centering
\caption{Quantum layer parameters of the Q-SCM.}
\label{tab:quantum_params}
\footnotesize
\begin{tabular}{p{0.22\linewidth} p{0.10\linewidth} p{0.58\linewidth}}
\toprule
Parameter & Range & Interpretation \\
\midrule
$\theta_{\max,\mathrm{dev}}$ & $[0,\pi]$ & Maximum phase rotation generated by the deviation cue.\\
$\theta_{\max,\mathrm{dist}}$ & $[0,\pi]$ & Maximum state-changing rotation generated by the distance cue. \\
$\theta_{\max,\mathrm{cttc}}$ & $[0,2.5]$ & Maximum state-changing rotation generated by the CTTC cue. \\
$k_{\mathrm{dev}}$ & $[10^{-3},100]$ & Saturation parameter of the deviation cue in Eq.~\eqref{eq:application_rotation_mapping}. \\
$k_{\mathrm{dist}}$ & $[10^{-3},100]$ & Saturation parameter of the distance cue. \\
$k_{\mathrm{cttc}}$ & $[10^{-3},100]$ & Saturation parameter of the CTTC-based cue. \\ 
$\theta_0$ & $[0,\pi]$ & Population-level initialization angle in Eq.~\eqref{eq:application_initial_state}. \\
$\lambda$ & $[0,0.99]$ & Per-sample relaxation weight in Eq.~\eqref{eq:application_relaxed_state}. \\
\bottomrule
\end{tabular}
\end{table}

% \subsubsection{Estimation Strategy}
% \label{sec:estimation_strategy}

The case-study parameter vector is $\boldsymbol{\Theta}=(\boldsymbol{q},\boldsymbol{c})\in\mathbb{R}^{17}$, where: 
\begin{equation}
\boldsymbol{q} = \left( \theta_{\max,\mathrm{dev}}, \theta_{\max,\mathrm{dist}}, \theta_{\max,\mathrm{cttc}}, k_{\mathrm{dev}},
k_{\mathrm{dist}}, k_{\mathrm{cttc}}, \theta_0, \lambda \right)
\label{eq:application_quantum_parameter_vector}
\end{equation}
contains the eight parameters of the Q-SCM state membership layer, and $\boldsymbol{c}$ contains the nine state conditional choice parameters listed in Table~\ref{tab:choice_params}.
The Q-SCM parameters are obtained using penalized maximum likelihood. Define: 
\begin{equation}
\mathcal{P}_D = \left\{ \beta^{D}_{\mathrm{decel}}, \beta^{D}_{\mathrm{accel}}, \beta^{D}_{\mathrm{nonoverlap}}, \beta^{D}_{\mathrm{int}} \right\}
\label{eq:penalized_parameter_set}
\end{equation}
as the set of defensive state choice coefficients subject to regularization. The penalized objective is:
\begin{equation}
\mathcal{Q} \left( \boldsymbol{\Theta} \right) = - \ell \left(\boldsymbol{\Theta} \right) + \eta_{\mathrm{reg}} \sum_{\beta\in\mathcal{P}_D} \beta^2,
\label{eq:penalized_objective}
\end{equation}
where $\eta_{\mathrm{reg}}=0.1$ is regularization weight used in the numerical implementation. The Q-SCM estimate is:
\begin{equation}
\widehat{\boldsymbol{\Theta}}_{\mathrm{pen}} = \underset{\boldsymbol{\Theta}\in\Omega}{\arg\min} \; \mathcal{Q} \left( \boldsymbol{\Theta} \right).
\label{eq:penalized_estimator}
\end{equation}
The regularization term is introduced as a numerical stabilization choice to limit excessively large Defensive state utility coefficients during joint estimation. The selected penalty set $\mathcal{P}_D$ and fixed weight $\eta_{\mathrm{reg}}=0.1$ are components of the reported empirical specification; they should not be interpreted
as behavioural restrictions requiring Defensive state coefficients to be small. The reported results and approximate uncertainty summaries are conditional on this regularization choice.
The unpenalized log-likelihood:
\begin{equation}
\widehat{\ell}_{\mathrm{eval}} = \ell \left( \widehat{\boldsymbol{\Theta}}_{\mathrm{pen}} \right)
\label{eq:reported_unpenalized_likelihood}
\end{equation}
is reported as the measure of in sample data fit. Hence, the reported log-likelihood excludes the penalty even though the Q-SCM parameters are obtained by minimizing the penalized objective.

Because the likelihood could include multiple local maxima with respect to the maximum lane deviation rotation, $\theta_{\max,\mathrm{dev}}$, we adopt a two stage grid search to obtain a suitable starting value for this parameter.
% can be non-concave. Therefore we use a two stage grid search to initialize this parameter.
At first, we evaluate a coarse grid of 50 points uniformly spaced over $[0.1,\pi-0.01]$.

\newpage
Then  a fine grid of 30 points within a $\pm 0.3$ rad neighbourhood of the best performing course grid value. 
At each grid point, the $(\theta_{\max,\mathrm{dev}})$ is fixed at the candidate value.
The remaining quantum layer parameters are temporarily fixes at 
$\theta_{\max,\mathrm{dist}}=\theta_{\max,\mathrm{CTTC}}=\pi$, $k_{\mathrm{dev}}=k_{\mathrm{dist}}=k_{\mathrm{CTTC}}=1$, $\theta_0=0$, and $\lambda=0$. 
The log-likelihood function is then maximized only with respect to the nine choice coefficients $\boldsymbol{c}$ using Limited-memory Broyden--Fletcher--Goldfarb--Shanno with Bounds (L-BFGS-B). 
The value of $(\theta_{\max,\mathrm{dev}})$ that has the highest likelihood in the fine grid search is then used as the starting value for the final joint estimation.
% The fine grid maximizer is then used as the warm start for the joint estimation stage.
Finally, the penalized objective in Eq.~\eqref{eq:penalized_objective} is minimized jointly over the complete 17 parameter vector using the bounded L-BFGS-B algorithm. 
The maximum rotation parameters $(\theta_{\max,\mathrm{dev}})$ and $(\theta_{\max,\mathrm{dist}})$ are constrained to $[0,\pi]$, except for $\theta_{\max,\mathrm{CTTC}}$, which is constrained to $[0,2.5]$ to avoid boundary attraction. The saturation constants are constrained to $[10^{-3},100]$, the initial angle is constrained to $\theta_0\in[0,\pi]$, and the relaxation parameter is constrained to $\lambda\in[0,0.99]$.

The classical specification has 11 parameters and is estimated using the Expectation--Maximization algorithm. In the expectation step, the posterior probabilities of belonging to the neutral and defensive classes are obtained using Bayes' rule. In the maximization step, the state-specific choice coefficients are updated using weighted logit estimation. These two steps are repeated untill convergence. 
Since the conditional choice utility specification is identical in the classical and Q-SCM models, any difference in model fit is attributable to the class assignment mechanism. 
In the classical model, the defensive state probability $P^{\mathrm{cl}}(D_t)$ in Eq.~\eqref{eq:classical_class} is a memoryless function of contemporaneous CTTC. 

We deliberately adopt the latent class logit, rather than a richer classical model as the benchmark in this work. 
Holding the nine parameter choice utility fixed and changing only the class assignment mechanism isolates the contribution of the quantum sequential evolution.
In other words, any difference in fit is attributable to cue history dependence and order effects, not to a different choice kernel or additional covariates. 
A random parameters (mixed) latent class model would instead introduce between driver heterogeneity, and this is a different and complementary source of richness from the within driver sequential dynamics that Q-SCM represents.
Therefore, we treat it as future work rather than as the benchmark here.
To verify that the Q-SCM is identifiable and that the proposed estimation procedure can recover its parameters, additionally we estimate the model on a synthetic dataset. 
The cues are drawn from the empirical distributions of the original rounD trajectories. The observed actions are simulated from the Q-SCM under the estimated parameters and treated as ground truth.

% =============================================================
\section{Estimation Results }
\label{sec:results}

% The Q-SCM and the classical benchmark are both estimated in Python. 
% The PennyLane library~\cite{bergholm2018pennylane} was used to define and verify quantum rotation operators for the membership layer. 
% However, for computational efficiency, the recursive state evolution for 85{,}754 observations including  monotonicity constraint, geodesic safeguard, and relaxation mechanism were evaluated directly using vectorized complex valued linear algebra in NumPy~\cite{harris2020numpy}. 
% The Q-SCM parameters were estimated by maximizing the observed log-likelihood in Eq.~\eqref{eq} using the bounded L-BFGS-B algorithm implemented in SciPy~\cite{virtanen2020scipy}. 
% The optimization is initialized using the two stage grid search procedure as discribed earlier. 
% The classical latent class benchmark is estimated using an Expectation--Maximization algorithm implemented with the same NumPy and SciPy computational framework. 
% Parameter uncertainty was quantified using a finite-difference approximation of the log-likelihood Hessian together with a driver-level cluster-robust sandwich covariance estimator constructed from observation-level score vectors. Clustering at the driver level accounts for dependence among repeated trajectory observations from the same driver.

We compare the proposed Q-SCM to a classical probability formulation which is the latent class specification, similar to \cite{alhaideri2026dualState} and also described at the end of Section~\ref{sec:likelihood}. 
Both models are estimated on the full 85{,}754 observations from 9{,}610 passenger car drivers in the rounD Neuweiler dataset and share the identical nine parameter choice utility specification (Eqs.~\ref{eq:V_N}--\ref{eq:V_D}). 
% The class membership in the classical model is defined using the CTTC which is a traffic conflict indicator in addition to an alternative specific constant specified for the defensive class. A similar specification is described in detail in \cite{alhaideri2026dualState}. 
% The conditional choice layer in both models is identical. By doing this, by holding the choice layer fixed we isolate the contribution of the class assignment mechanism, which is the methodological focus of this paper.
Table~\ref{tab:gof_summary} presents a comparison between the two estimated formulations, the Q-SCM and the classical model. 
The Q-SCM has a log-likelihood of $\hat{\ell}^{\mathrm{Q\text{-}SCM}} = -63{,}359.9$, which is larger than the the classical model $\hat{\ell}^{\mathrm{Classical}} = -63{,}831.7$ by $+471.8$ units.
% As can be seen from the table, the AIC of the Q-SCM is lower by $\Delta\mathrm{AIC} = 931.5$ ($126{,}753.9$ vs.\ $127{,}685.4$). 
% $\bar{P}(D)$ is the sample mean of the model implied defensive class probability.

\begin{table}[h]
\centering
\caption{Descriptive in-sample fit comparison
($n=85{,}754$).}
\label{tab:gof_summary}
\footnotesize
\begin{tabular}{lrrr}
\toprule
Model
&
\multicolumn{1}{c}{Unpenalized log-likelihood}
&
\multicolumn{1}{c}{Reported parameters}
&
$\bar{P}(D)$
\\
\midrule
Q-SCM
&
$-63{,}359.9$
&
$17$
&
$0.424$
\\
Classical
&
$-63{,}831.7$
&
$11$
&
$0.663$
\\
\midrule
$\Delta$
&
$+471.8$
&
$+6$
&
$-0.239$
\\
\bottomrule
\end{tabular}
\end{table}

Since the Q-SCM parameters are obtained from a penalized objective, standard likelihood-based information criteria do not have their usual unpenalized maximum-likelihood interpretation. 
Therefore, the comparison focuses on the unpenalized log-likelihood evaluated at the reported estimates and treats the fit differences as descriptive rather than as a formal likelihood-ratio comparison. 

\newpage
\noindent
All the reported standard errors and $p$-values are approximate local Wald type summaries. They are calculated using the driver-level cluster robust covariance matrix based on the unpenalized log-likelihood evaluated at the penalized Q-SCM estimate.

% \begin{table}[h]
% \centering
% \caption{Goodness of fit comparison ($n=85{,}754$).}
% \label{tab:gof_summary}
% \begin{tabular}{llllll}
% \toprule
% Model       & Log-Likelihood  & $df$ & AIC          & BIC          & $\bar{P}(D)$ \\
% \midrule
% Q-SCM       & $-63{,}359.9$ & $17$ & $126{,}753.9$ & $126{,}913.0$ & $0.424$ \\
% Classical   & $-63{,}831.7$ & $11$ & $127{,}685.4$ & $127{,}788.4$ & $0.663$ \\
% \midrule
% $\Delta$    & $+471.8$      & $+6$ & $-931.5$     & $-875.4$     & $-0.239$ \\
% \bottomrule
% \end{tabular}
% \end{table}

\subsection{Q-SCM quantum layer parameter estimates}
\label{sec:results_quantum}

Table~\ref{tab:quantum_estimates} shows the estimates of the quantum layer parameters.
All three maximum rotation angles are statistically different from zero ($p<0.001$). Since $\theta_{\max,c}=0$ would remove cue $c$ from membership layer entirely, each cue is retained by the data as a driver of state evolution.
The estimated angles are $\theta_{\max,\mathrm{dist}} = 3.12$\,rad (at the upper bound $\pi$), $\theta_{\max,\mathrm{dev}} = 2.61$\,rad, and $\theta_{\max,\mathrm{CTTC}} = 1.76$\,rad.
However, because the distance and CTTC cues are bounded above by one (Eqs.~\ref{eq:distance_cue}--\ref{eq:cttc_cue}), the saturation factor $\tau_{c,t}=s_{c,t}/(k_c+s_{c,t})$ cannot reach one.
The attainable single update rotation is $\theta_{\max,c}/(k_c+1)$ rather than $\theta_{\max,c}$.
Based on the estimated parameters, this attainable rotation is $2.13$\,rad for distance.
It can increase $P(D)$ from $0$ to at most $\sin^2(2.13/2)\approx 0.77$ in a single update, and $0.63$\,rad for CTTC. The deviation cue which has an unbounded encoding can approach its full $2.61$\,rad.
Hence, the attainable per-update responses are behaviourally large but bounded below a complete neutral-to-defensive flip.
The saturation constants differ by a factor of around four ($\hat{k}_{\mathrm{dist}} = 0.47$, $\hat{k}_{\mathrm{dev}} = 1.07$, $\hat{k}_{\mathrm{CTTC}} = 1.78$). This implies the different scales of the underlying cue signals.
Since $k_c$ controls both the saturation rate and the attainable rotation ceiling, having a single shared $k$ across cues could mis-state both.
% The rotation magnitudes are all substantial, indicating that each of the three cues contributes meaningfully to the state evolution.
% The three maximum rotation angles are statistically significant. 
% Their estimated magnitudes are also large relative to their admissible range $[0,\pi]$.
% A rotation of $\pi$ corresponds to a complete neutral-to-defensive transition.
% A fully saturated distance cue can move the state across the full neutral--defensive range in a single update ($\theta_{\max,\mathrm{dist}} = 3.12 \approx \pi$), while the
% deviation and CTTC cues generate maximum rotations of $2.61$\,rad ($\approx 150^{\circ}$) and $1.76$\,rad ($\approx 101^{\circ}$), respectively.
% The deviation cue rotates the state by up to $\theta_{\max,\mathrm{dev}} = 2.61$\,rad ($\approx 150^{\circ}$) per timestep.
% The distance cue reaches the bound at $\theta_{\max,\mathrm{dist}} = 3.12$\,rad ($\approx \pi$), saturated at the feasibility boundary.
% The CTTC cue reaches $\theta_{\max,\mathrm{CTTC}} = 1.76$\,rad ($\approx 101^{\circ}$).
% The cue specific saturation constants differ by a factor of almost four: $\hat{k}_{\mathrm{dist}} = 0.47$, $\hat{k}_{\mathrm{CTTC}} = 1.78$, $\hat{k}_{\mathrm{dev}} = 1.07$.
% This heterogeneity reflects the different scales of the underlying cue signals (inverse separation distance, inverse CTTC, and lane deviation).
% It indicates that imposing a single shared $k$ across cues would substantially mis-specify the rate at which each cue saturates.
The population level initial polar angle is estimated to be $\theta_0 = 0.065$,rad ($\approx 3.7^{\circ}$), which implies a baseline defensive probability of $\sin^2(\theta_0/2) \approx 0.001$.
Although close to lower bound of its admissible range, the estimate is precisely determined (SE $= 0.0078$, $t \approx 8.3$, $p<0.001$). This means that the data distinguish this small nonzero baseline from an exactly neutral initial state. In other words, the drivers enter the observation window in a near pure but not perfectly pure neutral state.
% This implies that drivers enter the observation window in a near pure neutral state.

The estimated relaxation weight is $\widehat{\lambda}=0.44$. 
Before normalization, the relaxation map assigns weight $\widehat{\lambda}$ to the initial baseline state and weight $1-\widehat{\lambda}$ to the pre-relaxation state. 
Because the resulting vector is normalized, this value should not be interpreted as an exact $44\%$ reduction in the defensive state probability or angular displacement.
This model specific relaxation mechanism prevents the carried forward state from retaining the effects of earlier interactions indefinitely.
Since the estimated initialization angle is small but greater than zero, relaxation occurs toward a near neutral initial baseline rather than toward a mathematically pure neutral state.

\begin{table}[h]
\centering
\caption{Quantum layer parameter estimates}
\label{tab:quantum_estimates}
\footnotesize
\begin{tabular}{llll}
\toprule
Parameter & Estimate & SE & \textit{p}-value \\
\midrule
$\theta_{\max,\mathrm{dev}}$   & $2.6143$ & $0.0130$ & $<0.001$ \\
$\theta_{\max,\mathrm{dist}}$  & $3.1191$ & $0.0213$ & $<0.001$ \\
$\theta_{\max,\mathrm{CTTC}}$ & $1.7565$ & $0.0180$ & $<0.001$ \\
$k_{\mathrm{dev}}$             & $1.0721$ & $0.0079$ & $<0.001$ \\
$k_{\mathrm{dist}}$            & $0.4652$ & $0.0088$ & $<0.001$ \\
$k_{\mathrm{CTTC}}$           & $1.7781$ & $0.0131$ & $<0.001$ \\
$\theta_0$                     & $0.0649$ & $0.0078$ & $<0.001$ \\
$\lambda$                      & $0.4415$ & $0.0077$ & $<0.001$ \\
\bottomrule
\end{tabular}
%\vspace{2pt}

%\vspace{-2pt}
%\begin{minipage}{0.96\linewidth}
%\footnotesize
%\textit{Note:} The reported standard errors and $p$-values are approximate %local Wald type summaries. They are calculated using the
%driver-level cluster-robust covariance matrix based on then penalized log-%likelihood evaluated at the penalized Q-SCM  estimate. 
%\end{minipage}

\end{table}

\subsection{Choice Layer Parameter Estimates}
\label{sec:results_choice}

Table~\ref{tab:choice_estimates} presents the within-state MNL coefficients for both models. 
The signs are consistent for the two specifications, except for $\beta^N_{\mathrm{decel}}$ which is insignificant in the classical model.
A positive deceleration constant and negative acceleration constant in the Defensive state, capturing the expected within-state action preferences.
The negative alternative-specific path distance coefficient in the Neutral state indicates that drivers prefer candidate manoeuvres whose centroids remain closer to the intended path.
And a positive intensity coefficient in both states, indicating that frontal threats induce deceleration and rear threats induce acceleration. 
The Neutral state constants differ in sign across the two models.

\newpage
The classical model has $\beta^N_{\mathrm{decel}} < 0$ and $\beta^N_{\mathrm{accel}} > 0$.
The Q-SCM has both Neutral constants positive. This difference is consistent with the two models assigning the Neutral and Defensive classes to qualitatively different observation sets, as discussed below.
The maintain speed alternative ($j=2$) is the reference alternative for
the alternative-specific constants within each latent state. 
Its alternative varying attributes, such as $d_{j,t}$ and $o_{j,t}$, remain in the utility specification.
% The maintain speed ($j=2$) is the within-state reference alternative
% The neutral class is the reference class in the classical model.

\begin{table}[h]
\centering
\footnotesize
\caption{Q-SCM and Classical models within-state choice coefficient estimates.}
\label{tab:choice_estimates}
\begin{tabular}{lllllll}
\toprule
 & \multicolumn{3}{c}{Q-SCM} & \multicolumn{3}{c}{Classical} \\
\cmidrule(lr){2-4} \cmidrule(lr){5-7}
Parameter   & Estimate & SE       & \textit{p}-value & Estimate & SE       & \textit{p}-value \\
\midrule
$\beta^{N}_{\mathrm{decel}}$       & $ 0.4288$ & $0.0336$ & $<0.001$ & \textcolor{gray}{-0.1422} & \textcolor{gray}{0.1038} & \textcolor{gray}{0.171}  \\
$\beta^{N}_{\mathrm{accel}}$       & $ 2.0198$ & $0.0471$ & $<0.001$ & $ 5.7358$ & $0.1595$ & $<0.001$ \\
$\beta^{N}_{\mathrm{dist\_path}}$  & $-2.1288$ & $0.0524$ & $<0.001$ & $-5.7142$ & $0.1700$ & $<0.001$ \\
$\beta^{N}_{\mathrm{nonoverlap}}$     & $ 0.5190$ & $0.0299$ & $<0.001$ & $ 0.6494$ & $0.0623$ & $<0.001$ \\
$\beta^{N}_{\mathrm{int}}$         & $ 0.3087$ & $0.0315$ & $<0.001$ & $ 0.5562$ & $0.0661$ & $<0.001$ \\
$\beta^{D}_{\mathrm{decel}}$       & $ 1.8540$ & $0.0392$ & $<0.001$ & $ 1.3543$ & $0.0177$ & $<0.001$ \\
$\beta^{D}_{\mathrm{accel}}$       & $-2.0404$ & $0.0879$ & $<0.001$ & $-1.7118$ & $0.0777$ & $<0.001$ \\
$\beta^{D}_{\mathrm{nonoverlap}}$     & $ 0.9378$ & $0.1485$ & $<0.001$ & $ 0.4611$ & $0.0807$ & $<0.001$ \\
$\beta^{D}_{\mathrm{int}}$         & $ 0.3348$ & $0.0652$ & $<0.001$ & $ 0.1654$ & $0.0344$ & $<0.001$ \\
\bottomrule
\end{tabular}
\end{table}
%\vspace{-2pt}
%\begin{minipage}{0.96\linewidth}
%\footnotesize
%\textit{Note:} The reported standard errors and $p$-values are approximate local Wald type summaries. They are calculated using the
%driver-level cluster-robust covariance matrix based on the  npenalized log-likelihood evaluated at the penalized Q-SCM  estimate. 
%\end{minipage}

\begin{table}[h]
\centering
\small
\caption{Classical model class membership coefficients estimates (Eq.~\ref{eq:classical_class}).}
\label{tab:classical_class}
\footnotesize
\begin{tabular}{llll}
\toprule
Parameter (Defensive Class) & Estimate & SE & \textit{p}-value \\
\midrule
$\gamma_1$ (CTTC-based encoded cue coefficient) & $-4.4431$ & $0.1192$ & $<0.001$ \\
$\gamma_2$ (constant)         & $ 1.0134$ & $0.0224$ & $<0.001$ \\
\bottomrule
\end{tabular}
\end{table}
% \vspace{-2pt}
%\begin{minipage}{0.96\linewidth}
%\footnotesize
%\textit{Note:} The reported standard errors and $p$-values are approximate local Wald type summaries. They are calculated using the
%driver-level cluster-robust covariance matrix based on the  npenalized log-likelihood evaluated at the penalized Q-SCM  estimate. 
%\end{minipage}

The classical class membership estimates suggest a switching pattern rather than a simple conflict to defensive state response. 
The positive constant indicates that under weak or moderate instantaneous conflict, the model assigns a higher baseline probability to the Defensive class. 
This class is behaviourally consistent with cautious driving in the within-state choice layer, as it favours deceleration and avoids acceleration relative to maintaining speed. 
However, the negative CTTC coefficient indicates that, as the conflict becomes more immediate, represented by smaller CTTC and larger values of $1/(1+\mathrm{CTTC}_t)$, the probability of belonging to this deceleration oriented class decreases. 

Thus, the classical model suggests that drivers may move away from a cautious deceleration state when a severe conflict materializes.
Instead they enter a more committed manoeuvre state, characterized by stronger acceleration and lane following behaviour. 
This interpretation is plausible in urgent interactions.
In such instances, a driver may choose to accelerate through or complete the manoeuvre rather than continue slowing down. 
However, it also shows that the classical state labels should be interpreted cautiously.
The class called Defensive is defensive in terms of its within-state manoeuvre preferences, but it is not directly activated by severe instantaneous CTTC. 
This contrasts with Q-SCM, where the defensive probability is generated through sequential cue processing, accumulated state evolution, and state carry forward rather than a single contemporaneous CTTC membership term.

The unpenalized log-likelihood evaluated at the reported Q-SCM estimate is 471.8 units higher than that of the classical benchmark. 
Since Q-SCM estimate is obtained using regularization and its membership layer uses a broader set of cues, this difference is interpreted as a descriptive improvement relative to the selected contemporaneous CTTC-based benchmark. It should not be attributed exclusively to quantum sequential history and order effects.
% We beleive that the 471.8 LL improvement is not a consequence of the Q-SCM having strictly more parameters than the classical baseline, because the AIC penalty already adjusts for that. Rather, it identifies a real difference in the structure of the class assignment mapping. 
We highlight three methodological features of the Q-SCM that are absent from the classical baseline specification used here.
The classical class probability $P(D_t)$ is a function only of the contemporaneous cue $\mathrm{CTTC}_t$. Two observations with identical CTTC are assigned identical $P(D)$ irrespective of the driver's history. 

\newpage
The Q-SCM's $P(D_t)$, in contrast, is a function of the entire cue trajectory $\{s_{\mathrm{dist},u},\,s_{\mathrm{CTTC},u},\,s_{\mathrm{dev},u}\}_{u\le t}$ for the driver, integrated through the recursive quantum state evolution. 
Under the classical baseline used in this paper, two drivers with identical CTTC at $t$ receive identical $P(D_t)$, even if their cue histories differ. 
In Q-SCM, the same two drivers can have different $P(D_t)$ because the current state carries the accumulated effect of prior cue updates.
% , and inspection of the per observation log-likelihood differences shows that the advantage is broadly distributed across the dataset rather than concentrated in a few outlier observations.

The aggregate gain is consistent but modest at the level of individual observations. 
Figure~\ref{fig:cumulative_ll} shows the cumulative difference
in per observation log likelihood between the two models across all
$85{,}754$ observations.
As can be seen fromt he figure, the curve rises steadily to its final value of $+471.8$ without large jumps.
This implies that the Q-SCM advantage is broadly distributed rather than driven by a few outliers. 
At the same time, the per observation improvement is small on average ($+471.8/85{,}754 \approx +0.006$), and Q-SCM assigns a higher likelihood than the classical model to roughly $54\%$ of drivers and of observations. 
The contribution of Q-SCM is therefore best understood as a
consistent structural improvement in the class assignment mechanism,
together with a more behaviourally interpretable latent state, rather than a large gain in point wise predictive accuracy.

% add the individual cumulative LL fig here
\begin{figure}[h]
\centering
\includegraphics[width=0.8\linewidth]{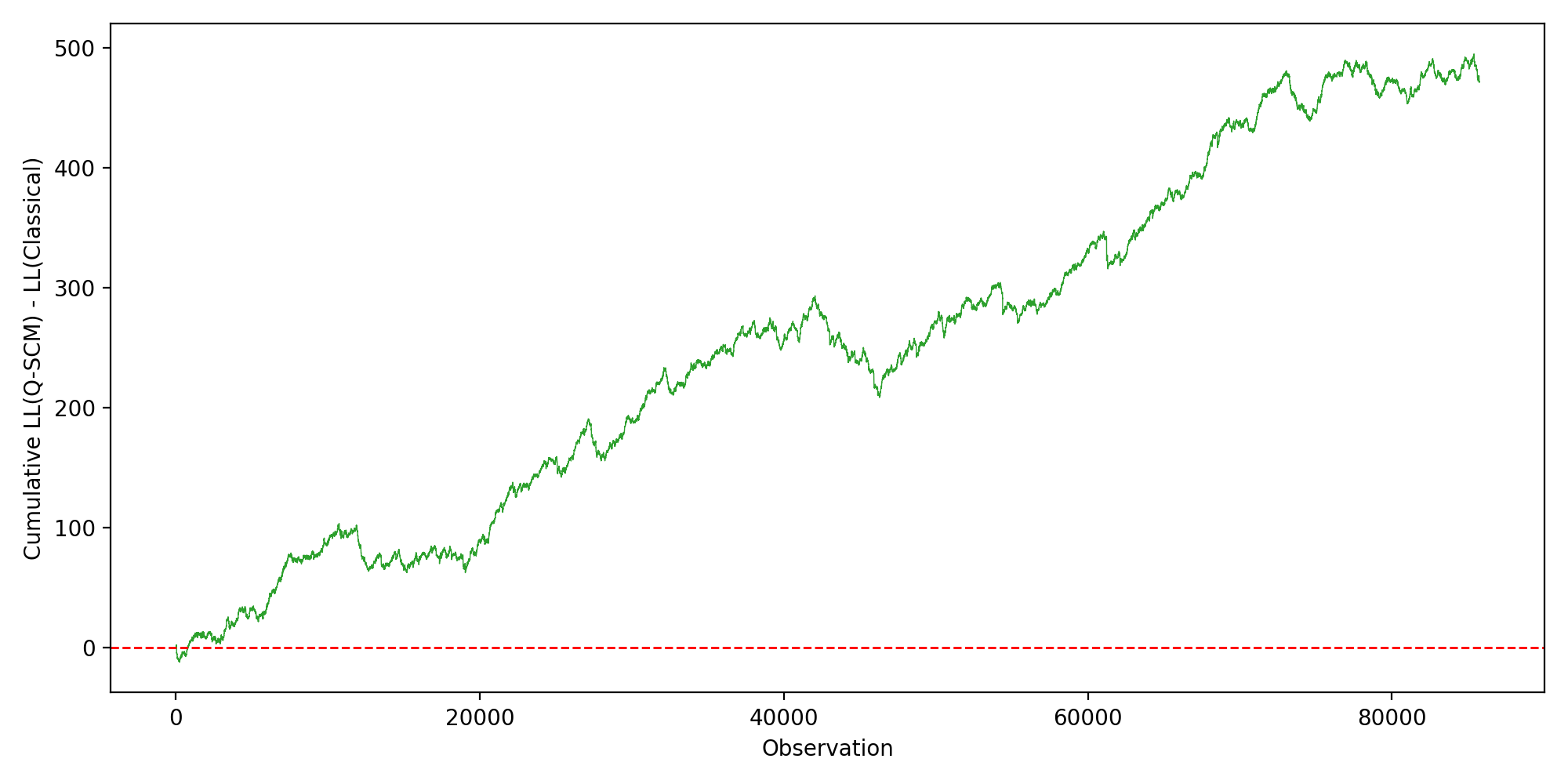}
\caption{Cumulative difference in per observation log likelihood between Q-SCM and the classical model.}
\label{fig:cumulative_ll}
\end{figure}

Unlike the classical class membership, where the effect of CTTC on class utility is controlled directly by the estimated coefficient and can change sharply from one observation to the next ,the Q-SCM updates the defensive probability by a bounded state rotations. 
Each cue can only move the cognitive state by a limited amount at each timestep, and the saturating response function prevents abrupt unbounded changes. 
The no pendulum constraint ensures that sustained threat exposure does not produce unrealistic oscillation in the defensive probability. 
When the threat weakens or disappears, the estimated relaxation weight $\lambda=0.44$ allows the state to move back to the neutral baseline. 
Hence, Q-SCM provides a smoother and more behaviourally constrained representation of how defensive state probability increases under continued threat and decreases after the threat passes.

Figure~\ref{fig:pd_histograms} shows distributions of $P(D_t)$ for the 85{,}754 observations in the two models. 
The classical distribution is concentrated near $P(D)\approx0.73$ with a long left tail.
Since the membership utility of the Defensive class decreases when the short CTTC conflict term increases, observations with severe instantaneous conflict are shifted away from the Defensive toward the Neutral class.
The Q-SCM distribution is bimodal instead.
Drivers either are near $P(D)\approx 0$ or near $P(D)\approx 1$, with intermediate probabilities representing observations during state transition. 
This near neutral versus near defensive separation is a direct consequence of the bounded sequential response. 
Once cumulative cue processing has moved the state vector close to one pole, the state tends to remain near that pole until subsequent cue evidence and relaxation move it away. 

\newpage
The two models identify qualitatively different latent constructs. 
The classical construct corresponds to a default cruise mode in which deceleration is the preferred response. 
The Q-SCM construct corresponds to a committed response mode that drivers enter only after sufficient accumulated threat has been processed over time. 
Both interpretations are internally consistent, but the Q-SCM interpretation is more closely aligned with the behavioural notion of defensive driving as a state shift triggered by accumulated risk perception.

% Figure~\ref{fig:pd_trajectories} illustrate the qualitative difference at the driver-trajectory level
% Q-SCM trajectories show persistent runs in the high-$P(D)$ regime once a driver has been pushed there,  classical trajectories track instantaneous CTTC without the kind of within-driver autocorrelation that the quantum walk imposes.

\begin{figure}[H]
\centering
\includegraphics[width=0.9\linewidth]{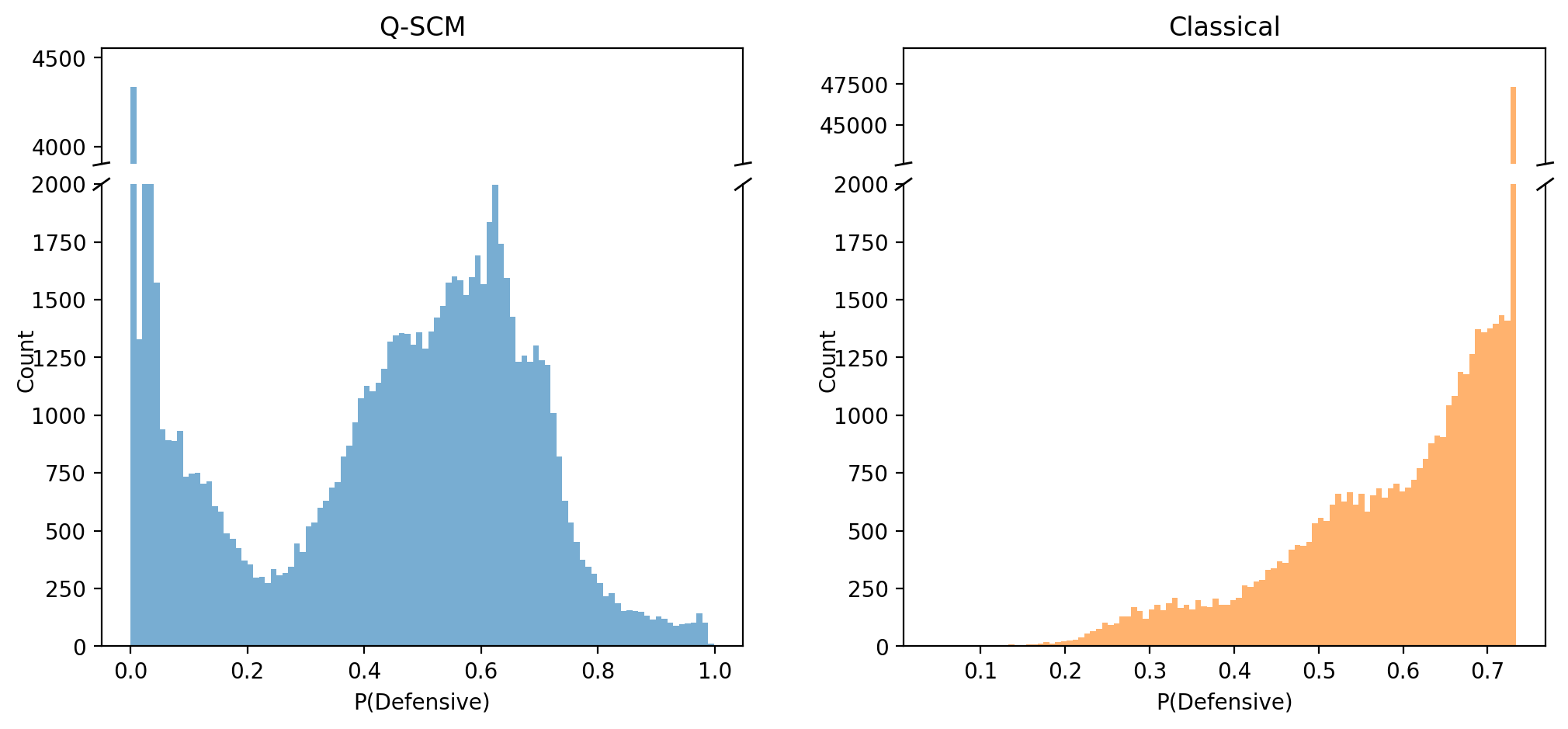}
\caption{Distribution of $P(\mathrm{Defensive}_t)$ in the Q-SCM and classical baseline.}
\label{fig:pd_histograms}
\end{figure}

% add synthetic
% add synthetic
\subsection{Synthetic Data Estimation}
\label{sec:synthetic}

The purpose of the synthetic data estimation is to evaluate the proposed Q-SCM in a controlled finite sample setting in which the data generating mechanism and parameter values are known by construction.
Unlike the empirical application, where the driver's true latent state process is 
unobserved, the synthetic experiment generates actions directly from the Q-SCM using a fixed data-generating parameter vector. This provides a check of whether the estimation procedure can approximately recover the imposed Q-SCM parameters and whether the resulting action patterns can be distinguished from those represented by the selected classical benchmark. 
% The exercise is a single finite sample recovery experiment rather than a formal proof of global identification or a repeated sampling Monte Carlo assessment of estimator bias.
The synthetic dataset is generated in two stages. First, $9{,}610$ driver trajectories are sampled with replacement from the original set of $9{,}610$ drivers in the rounD Neuweiler dataset. For each sampled driver, the complete observed trajectory is retained including chronological sequence of cue and alternative-specific variables$ \left( \mathrm{mDist}_t, \mathrm{CTTC}_t, d^{\mathrm{ego}}_{\mathrm{path},t},
o_{j,t}, d_{j,t}, f_t, r_t \right)$.
% Each sampled trajectory is assigned a unique synthetic driver identifier, including when the same empirical trajectory is selected more than once. This ensures that recursive state evolution and driver-level clustering are applied separately to each bootstrap draw. The resulting synthetic sample has $9{,}610$ synthetic drivers and $85{,}569$ observations.
Sampling complete trajectories preserves the within-driver temporal ordering and the joint empirical dependence among the cue and feature variables within each selected trajectory. Because the trajectories are sampled with replacement, the resulting synthetic dataset is a new
bootstrap analysis rather than a direct replication of the empirical
sample. 
In the second stage, the action variable is simulated. The data generating parameter vector $\boldsymbol{\Theta}^{\star}$ is specified by the reported empirical Q-SCM estimates. Once fixed for the synthetic experiment, these values are then treated as the known data generating parameters. 
At each synthetic observation, one of the three manoeuvres is sampled from the Q-SCM choice probabilities at $\boldsymbol{\Theta}^{\star}$. The simulated actions are treated as the dependent variable in the synthetic dataset.
Finally, Q-SCM is re-estimated using the synthetic data, and the recovered estimates are compared with the known data-generating values $\boldsymbol{\Theta}^{\star}$. The synthetic Q-SCM is estimated using bounded unpenalized maximum likelihood, with the regularization weight set
to zero.
Figure~\ref{fig:synth_distributions} shows the empirical and synthetic distributions of the cue variables, alternative-specific variables, and action shares. 
The figure shows that the driver-level bootstrap preserves the empirical distributions and within-driver sequences represented in the sampled trajectories.

\begin{figure}[h]
\centering
\includegraphics[width=0.98\linewidth]
{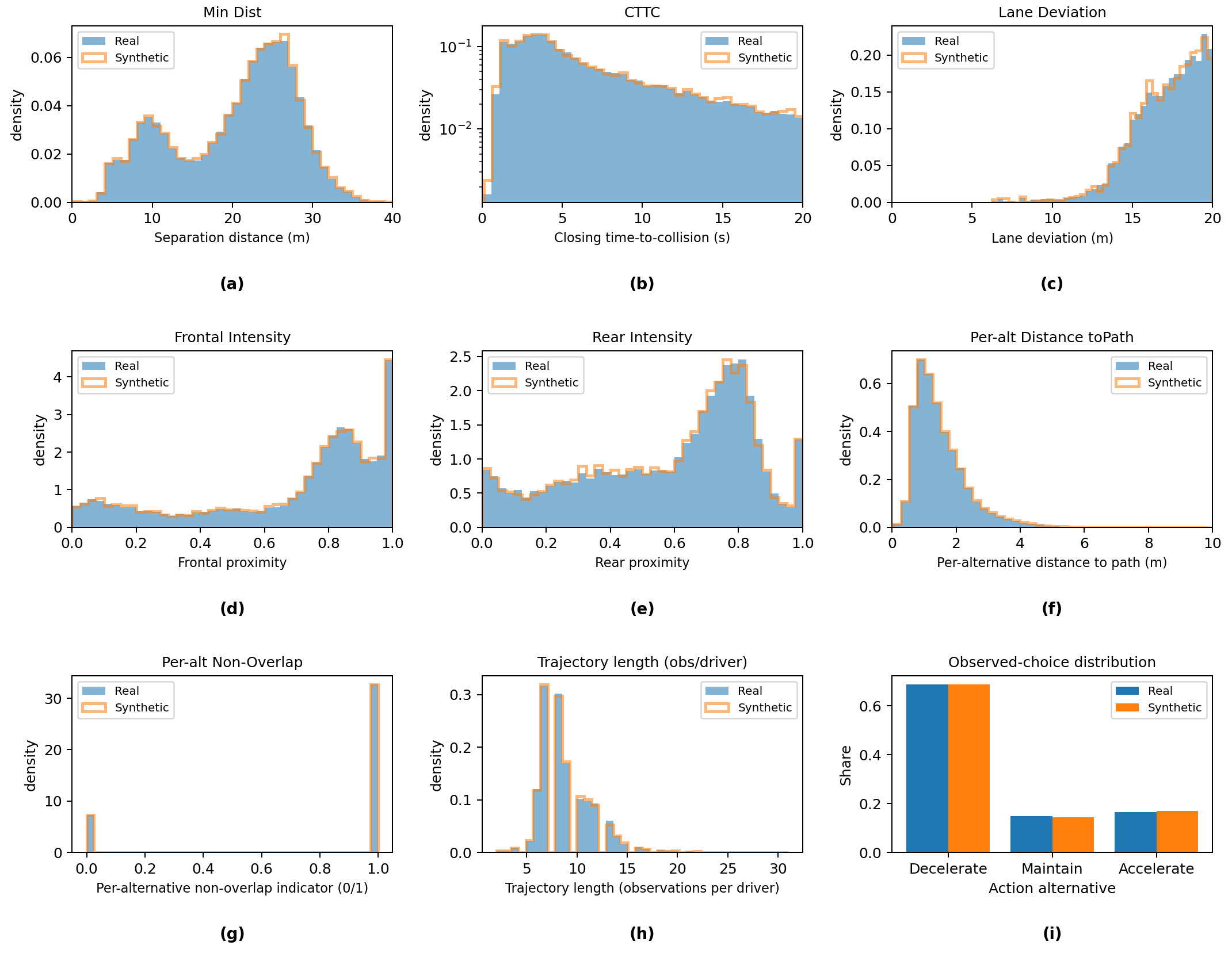}
\caption{Comparison of empirical and synthetic distributions.}
\label{fig:synth_distributions}
\end{figure}

Table~\ref{tab:synth_recovery} shows the Q-SCM parameter recovery results. For each parameter, the table reports the fixed data generating value, estimate obtained from the synthetic data, the difference labelled as bias, the driver-level cluster robust standard error, and the standardized recovery statistic. 
% For parameter $\Theta_r$, the reported statistic is calculated as:
% \begin{equation}
% t^{\mathrm{rec}}_r
% =
% \frac{
% \widehat{\Theta}^{\mathrm{synth}}_r-\Theta^{\star}_r
% }{
% \widehat{\operatorname{SE}}_{\mathrm{cl}}
% \left(
% \widehat{\Theta}^{\mathrm{synth}}_r
% \right)
% }.
% \label{eq:synthetic_recovery_statistic}
% \end{equation}
% The statistic evaluates whether the recovered estimate differs from the fixed data-generating value relative to its estimated cluster robust uncertainty. Because the results are based on one synthetic analysis, the reported differences should not be interpreted as estimates of the repeated-sampling bias of the estimator.
In this analysis, several Q-SCM parameters are recovered closely. The relaxation weight $\lambda$ differs from its data-generating value by only $-0.0003$. The maximum rotation parameters for the distance and deviation cues, $\theta_{\max,\mathrm{dist}}$ and $\theta_{\max,\mathrm{dev}}$, are also close to their designated data generating values.

\begin{table}[h]
\centering
\footnotesize
\caption{Q-SCM estimates using synthetic data. }
\label{tab:synth_recovery}
\begin{tabular}{lrrrrr}
\toprule
Parameter & True (Empirical) & Estimate (Synthetic) & Bias & SE & $t$-stat \\
\midrule
\multicolumn{6}{l}{\textit{Quantum layer}} \\
$\theta_{\max,\mathrm{dev}}$    & $ 2.6143$ & $ 2.6033$ & $-0.0110$ & $0.0194$ & $-0.57$ \\
$\theta_{\max,\mathrm{dist}}$   & $ 3.1191$ & $ 3.1349$ & $+0.0158$ & $0.0699$ & $+0.23$ \\
$\theta_{\max,\mathrm{CTTC}}^{*}$ & $ 1.7565$ & $ 2.1316$ & $+0.3751$ & $0.0478$ & $+7.85$ \\
$k_{\mathrm{dev}}^{*}$    & $ 1.0721$ & $ 0.9617$ & $-0.1104$ & $0.0134$ & $-8.26$ \\
$k_{\mathrm{dist}}$             & $ 0.4652$ & $ 0.5307$ & $+0.0655$ & $0.0528$ & $+1.24$ \\
$k_{\mathrm{CTTC}}^{*}$  & $ 1.7781$ & $ 1.5893$ & $-0.1888$ & $0.0220$ & $-8.57$ \\
$\theta_0^{*}$            & $ 0.0649$ & $ 0.0796$ & $+0.0147$ & $0.0064$ & $+2.29$ \\
$\lambda$                       & $ 0.4415$ & $ 0.4412$ & $-0.0003$ & $0.0091$ & $-0.03$ \\
\midrule
\multicolumn{6}{l}{\textit{Choice layer}} \\
$\beta^{N,*}_{\mathrm{decel}}$       & $ 0.4288$ & $ 0.3315$ & $-0.0973$ & $0.0234$ & $-4.16$ \\
$\beta^{N}_{\mathrm{accel}}$                 & $ 2.0198$ & $ 2.0219$ & $+0.0021$ & $0.0265$ & $+0.08$ \\
$\beta^{N,*}_{\mathrm{dist\_path}}$  & $-2.1288$ & $-2.2058$ & $-0.0770$ & $0.0351$ & $-2.19$ \\
$\beta^{N}_{\mathrm{nonoverlap}}$               & $ 0.5190$ & $ 0.5487$ & $+0.0297$ & $0.0292$ & $+1.02$ \\
$\beta^{N}_{\mathrm{int}}$                   & $ 0.3087$ & $ 0.3156$ & $+0.0069$ & $0.0307$ & $+0.22$ \\
$\beta^{D,*}_{\mathrm{decel}}$      & $ 1.8540$ & $ 1.9421$ & $+0.0881$ & $0.0414$ & $+2.13$ \\
$\beta^{D}_{\mathrm{accel}}$                 & $-2.0404$ & $-1.6706$ & $+0.3698$ & $0.2520$ & $+1.47$ \\
$\beta^{D}_{\mathrm{nonoverlap}}$               & $ 0.9378$ & $ 1.2351$ & $+0.2973$ & $0.2019$ & $+1.47$ \\
$\beta^{D}_{\mathrm{int}}$                   & $ 0.3348$ & $ 0.2732$ & $-0.0616$ & $0.0708$ & $-0.87$ \\
\bottomrule
\end{tabular}
\vspace{2pt}

\begin{minipage}{0.96\linewidth}
\footnotesize
\textit{Note:} ``True (Empirical)'' is the fixed data-generating values used to simulate the synthetic actions. Q-SCM is re-estimated using bounded unpenalized maximum likelihood. Bias is calculated as the synthetic estimate minus the fixed data generating value.  
\end{minipage}
\end{table}

The clearest joint parameter trade-off occurs between the CTTC cue parameters $\theta_{\max,\mathrm{CTTC}}$ and $k_{\mathrm{CTTC}}$. 
The recovered value of $\theta_{\max,\mathrm{CTTC}}$ is higher than its data generating value.
The recovered value of $k_{\mathrm{CTTC}}$ is lower. These opposite movements are consistent with the fact that the two parameters jointly determine the bounded cue response function.
A larger maximum rotation can be partly offset by a lower saturation constant.
This produces a similar effective response over the cue values represented in the sample. The two parameters should be interpreted jointly through the resulting rotation response function and not by their individual differences from the data generating values.

The maximum rotation for the deviation cue is recovered closely, although $k_{\mathrm{dev}}$ has a large standardized recovery difference.
The initial angle $\theta_0$ differs from its data generating value by only $0.0147$\,rad, although its standardized difference exceeds $1.96$ since its estimated standard error is small. The relaxation weight is recovered particularly closely.
Most choice layer estimates are also close to their data generating values. 
These results provide a finite-sample recovery diagnostic for the reported synthetic analysis. They do not establish the repeated sampling bias or coverage properties of the estimator.

The synthetic analysis also has a descriptive check of if the  Q-SCM generated action patterns can be diiferentiated from those represented by the selected classical benchmark. 
Since that the synthetic actions are generated from Q-SCM, the contemporaneous classical membership specification is mis-specified relative to the known data generating process and is expected to fit less well.
Table~\ref{tab:gof_synth} compares the maximized log-likelihood and information criteria of the two un-penalized synthetic data estimates.
Q-SCM has a log-likelihood of $-63{,}513.84$, compared with $-65{,}376.97$ for the classical benchmark. This leads to a descriptive difference of $1{,}863.13$ log-likelihood units in favour of Q-SCM.
 The Q-SCM membership layer processes distance, CTTC, and deviation cues sequentially, but the selected classical membership model uses contemporaneous CTTC only. 
 The fit difference therefore reflects the combined effects of the membership layer information set and the  state evolution mechanism. It does not isolate the contribution of cue history or cue order effects alone.

\begin{table}[h]
\centering
\caption{Goodness of fit for the synthetic dataset.}
\label{tab:gof_synth}
\footnotesize
\begin{tabular}{lrrrr}
\toprule
Model      & Log Likelihood   & $df$ & AIC          & BIC \\
\midrule
Q-SCM      & $-63{,}513.84$   & $17$ & $127{,}061.69$ & $127{,}220.76$ \\
Classical  & $-65{,}376.97$   & $11$ & $130{,}775.94$ & $130{,}878.86$ \\
\midrule
$\Delta$   & $+1{,}863.13$    & $+6$ & $-3{,}714.25$  & $-3{,}658.10$ \\
\bottomrule
\end{tabular}
\vspace{2pt}

\begin{minipage}{0.96\linewidth}
\footnotesize
\textit{Note:} Both Q-SCM and the classical benchmark are estimated
without regularization in the synthetic experiment. AIC and BIC are calculated using the
maximized unpenalized log-likelihoods. 
\end{minipage}
\end{table}

Table~\ref{tab:classical_synth} compares the classical model estimates obtained from the synthetic data with the corresponding classical model estimates obtained from the empirical data. 
The empirical estimates are included only as descriptive reference values. 
They are not true parameter values for the synthetic experiment because the synthetic actions are generated from Q-SCM rather than from the classical model.

\newpage
Several coefficients change when the classical model is fitted to the Q-SCM generated actions. The magnitude of the CTTC based class membership coefficient decreases from $\widehat{\gamma}_{1}=-4.44$ in the empirical model to $\widehat{\gamma}_{1}^{\mathrm{synth}}=-0.79$ in the synthetic model.
The class membership constant decreases from $\widehat{\gamma}_{2}=1.01$ to $\widehat{\gamma}_{2}^{\mathrm{synth}}=0.16$. 
The Neutral state acceleration constant also decreases in magnitude, from $\widehat{\beta}^{N}_{\mathrm{accel}}=5.74$ to $3.06$, as does the alternative-specific path-distance coefficient, from $\widehat{\beta}^{N}_{\mathrm{dist\_path}}=-5.71$ to $-2.86$.

These are consistent with two related effects. First, the classical membership equation uses only contemporaneous CTTC and cannot directly represent the distance, deviation, and history-dependent information embedded in the Q-SCM-generated actions. 
Second, under this misspecification, the classical membership and within-state utility coefficients jointly adapt to approximate the choice probabilities generated by Q-SCM. 
The resulting differences should therefore not be interpreted as conventional estimation bias relative to a true classical parameter vector. Rather, they illustrate how the selected classical specification reallocates explanatory influence across its membership and utility parameters when fitted to data generated by a sequential state evolution mechanism.

The synthetic comparison supports a limited structural conclusion. 
Under the controlled Q-SCM data generating process, the selected contemporaneous classical benchmark provides a less adequate approximation of the Q-SCM generated choice probabilities and resulting action patterns. 
The exercise does not directly evaluate recovery of the Q-SCM latent state membership probabilities by the classical model.
It does not establish that the empirical log-likelihood difference is free from all numerical or specification effects.

\begin{table}[h]
\centering
\footnotesize
\caption{Classical model estimates using synthetic data.}
\label{tab:classical_synth}
\begin{tabular}{lrrrr}
\toprule
Parameter & True (Empirical) & Estimate (synthetic) & SE & \textit{p}-value \\
\midrule
\multicolumn{5}{l}{\textit{Within-state choice utility}} \\
$\beta^{N}_{\mathrm{decel}}$        & $-0.1422$ & $ 0.3339$ & $0.0679$ & $<0.001$ \\
$\beta^{N}_{\mathrm{accel}}$        & $ 5.7358$ & $ 3.0579$ & $0.1011$ & $<0.001$ \\
$\beta^{N,*}_{\mathrm{dist\_path}}$ & $-5.7142$ & $-2.8618$ & $0.0879$ & $<0.001$ \\
$\beta^{N}_{\mathrm{nonoverlap}}$      & $ 0.6494$ & $ 0.8260$ & $0.0429$ & $<0.001$ \\
$\beta^{N}_{\mathrm{int}}$          & $ 0.5562$ & $ 0.1707$ & $0.0446$ & $<0.001$ \\
$\beta^{D,*}_{\mathrm{decel}}$      & $ 1.3543$ & $ 1.5614$ & $0.0430$ & $<0.001$ \\
$\beta^{D}_{\mathrm{accel}}$        & $-1.7118$ & $-1.6638$ & $0.1482$ & $<0.001$ \\
$\beta^{D}_{\mathrm{nonoverlap}}$      & $ 0.4611$ & $ 0.6837$ & $0.1468$ & $<0.001$ \\
$\beta^{D,*}_{\mathrm{int}}$        & $ 0.1654$ & $ 0.3598$ & $0.0600$ & $<0.001$ \\
\midrule
\multicolumn{5}{l}{\textit{Class membership}} \\
$\gamma_{1}$ (CTTC coefficient) & $-4.4431$ & $-0.7893$ & $0.1251$ & $<0.001$ \\
$\gamma_{2}$ (constant)         & $ 1.0134$ & $ 0.1602$ & $0.0401$ & $<0.001$ \\
\bottomrule
\end{tabular}
\vspace{2pt}

\begin{minipage}{0.96\linewidth}
\footnotesize
\textit{Note:} These values are not true parameters for the synthetic experiment because the synthetic actions are generated from Q-SCM. The synthetic estimates are obtained by fitting the unpenalized classical model to the Q-SCM generated actions. The reported standard errors are obtained from the inverse observed Hessian.
\end{minipage}
\end{table}

\section{Conclusions}
\label{sec:conc}
Motivated by quantum cognition theory, the Quantum-Sequential Choice Model (Q-SCM) offers a novel two-stage architecture to model dynamic driver behaviour. The framework represents a driver's evolving latent mental state, specifically transitioning between Neutral and Defensive regimes, using a two-dimensional complex state vector. This quantum state evolution generates time-varying latent class membership probabilities, which are then combined with classical conditional formulation to yield unconditional probablity for maneuver choices. The methodology is distinguished by its integration of state carry-forward, non-commuting cue-dependent rotations, phase information, and post-cue relaxation into a single recursive model. To govern the cue-driven state evolution before relaxation, the model incorporates formal mathematical safeguards. A monotonicity constraint ensures that state-changing updates do not inadvertently reduce the pre-relaxation Defensive state probability. Additionally, a geodesic safeguard redirects updates toward the Defensive pole and limits the applied angle to prevent overshooting. 

When applied to 85,754 observations from 9,610 drivers in the rounD naturalistic roundabout dataset, the Q-SCM outperformed a classical contemporaneous latent class benchmark. It achieved this by sequentially processing separation distance and time-to-collision as state-changing cues on the x-axis, alongside path deviation as a phase-changing cue on the z-axis. The inferred state trajectories successfully captured complex behavioural dynamics targeted by the specification. These dynamics include bounded cue responses, the accumulation of historical cue exposure, behavioural persistence even after a threat begins to weaken, and the eventual gradual recovery toward a baseline state. These properties emerge mathematically from the specified recursion rather than direct cognitive measurement. The synthetic data analysis using resampled empirical trajectories successfully validated the framework by closely recovering the underlying state evolution and choice parameters.

Despite these advancements, the current framework is constrained by its parsimonious two-state, pure-state formulation, which cannot capture unobserved cognitive noise, stochastic decoherence, or other driving regimes such as aggression or distraction. Furthermore, the model relies on rigid structural assumptions, including fixed cue-to-axis assignments, a static sequence for processing cues, and the use of absolute cue levels rather than temporal variations. It also assumes homogeneous parameters across the population, lacking the capacity to estimate baseline heterogeneity or individual behavioral sensitivities, and its validation is currently limited to a 1-second temporal resolution at a single roundabout location.

The future research will explore expanding the latent state space using generalized generators instead of relying solely on Pauli matrices, and adopting density-matrix representations to model mixed cognitive configurations. Introducing random parameters through simulated maximum likelihood or hierarchical Bayesian methods could effectively separate within-driver history dependence from between-driver heterogeneity. Additionally, the fixed cue-processing constraints could be relaxed to account for varying attention, salience, and non-commuting processing orders. A rigorous external validation is required, both through comparisons with richer dynamic classical models (e.g., hidden Markov transitions or autoregressive states), and via controlled virtual reality driving experiments that utilize neurophysiological sensors (e.g., electroencephalography and eye tracking) to empirically link predicted latent states with observed cognitive shifts.

\section*{Acknowledgments}
The first author is funded by the Canada Postdoctoral Research Award (CPRA) through the Natural Sciences and Engineering Research Council of Canada (NSERC). First and second authors are funded by the Canada Research Chair program in Disruptive Transportation Technologies and Services (CRC-2021-00480).

\appendix
\titleformat{\section}
  {\normalfont\bfseries\normalsize}
  {Appendix~\thesection}
  {0.6em}
  {}

% ===============================

% add table summary

% \clearpage
%\newpage
\section{Notation Summary}
\addcontentsline{toc}{section}{Notation Summary}

The following table summarizes the indices, sets, variables, operators,
probabilities, and model parameters used in the paper.

\small
\setlength{\LTpre}{4pt}
\setlength{\LTpost}{4pt}
\setlength{\tabcolsep}{5pt}
\renewcommand{\arraystretch}{1.15}

\begin{longtable}{
    L{0.1\linewidth}
    L{0.61\linewidth}
    L{0.12\linewidth}
}
% \footnotesize
\caption{Summary of indices, sets, variables, operators, and parameters
used in the Q-SCM.}
\label{tab:notation_summary}
\\

\toprule
\textbf{Symbol} &
\textbf{Definition} &
\textbf{Range/unit}
\\
\midrule
\endfirsthead

\multicolumn{3}{c}
{\tablename\ \thetable\ -- \textit{continued from the previous page}}
\\
\toprule
\textbf{Symbol} &
\textbf{Definition} &
\textbf{Range/unit}
\\
\midrule
\endhead

\midrule
\multicolumn{3}{r}{\textit{Continued on the next page}}
\\
\endfoot

\bottomrule
\endlastfoot

% ============================================================
\multicolumn{3}{l}{\textit{Indices and sets}}
\\
\midrule

$i$ &
Driver index; &
$i=1,\ldots,I$
\\

$I$ &
Total number of drivers; &
Positive integer
\\

$t$ &
Trajectory-sample index within the observed driving episode of driver
$i$; &
$t=0,\ldots,T_i$
\\

$T_i$ &
Final trajectory-sample index for driver $i$; &
Nonnegative integer
\\

$m$ &
Within-sample state index. The state
$\ket{\psi_{it,m}}$ is the state after the first $m$ cue updates have
been completed. Cue updates themselves run from
$m=0$ to $M_{it}-1$; &
$m=0,\ldots,M_{it}$
\\

$M_{it}$ &
Number of perceptual cues processed by driver $i$ at trajectory sample
$t$; &
Nonnegative integer
\\

$j$ &
Observed manoeuvre-alternative index; &
$j\in\mathcal{A}$
\\

$r,\;j'$ &
Dummy alternative indices used in summations over the choice set; &
$\in\mathcal{A}$
\\

$\mathcal{A}$ &
Set of observed manoeuvre alternatives; &
Finite set
\\

$s$ &
Latent mental state index; &
$s\in\mathcal{S}$
\\

$\mathcal{S}$ &
Set of latent mental states,
$\mathcal{S}=\{N,D\}$; &
Two states
\\

$N$ &
Neutral latent mental state; &
--
\\

$D$ &
Defensive latent mental state; &
--
\\

$k$ &
Basis-state index in the general $K$-state formulation; &
$k=1,\ldots,K$
\\

$K$ &
Dimension of the general latent mental state space. The proposed Q-SCM
uses $K=2$; &
Positive integer
\\

$c_{it,m}$ &
Type of cue processed by driver $i$ at trajectory sample $t$ and
within-sample update $m$; &
$c_{it,m}\in\mathcal{K}$
\\

$\mathcal{K}$ &
Set of admissible cue types; &
Finite set
\\

$a$ &
Rotation-axis index used in $\mathcal{U}_a(\theta)$; &
$a\in\{x,z\}$
\\

% ============================================================
\addlinespace
\multicolumn{3}{l}{\textit{Latent mental state representation}}
\\
\midrule

$\mathcal{H}_K$ &
$K$-dimensional complex Hilbert space used to represent the general
latent mental state; &
$\mathbb{C}^{K}$
\\

$\ket{S_k}$ &
Basis vector corresponding to latent mental state $S_k$ in the general
$K$-state formulation; &
Unit vector
\\

$\ket{N}$ &
Neutral-state basis vector,
$\ket{N}=[1\;0]^{\mathsf T}$; &
Unit vector
\\

$\ket{D}$ &
Defensive state basis vector,
$\ket{D}=[0\;1]^{\mathsf T}$; &
Unit vector
\\

$\ket{\psi_{it,m}}$ &
Latent mental state of driver $i$ at trajectory sample $t$ after the
first $m$ within-sample cue updates have been completed; &
Normalized
\\

$\ket{\psi_{it,0}}$ &
Beginning-of-sample state before any cue at trajectory sample $t$ is
processed; &
Normalized
\\

$\ket{\psi^{*}_{it,m+1}}$ &
Candidate state produced by cue update $m$, before application of the
model-control operation; &
Normalized
\\

$\ket{\widetilde{\psi}_{it}}$ &
Pre-relaxation state after all $M_{it}$ cues at trajectory sample $t$
have been processed,
$\ket{\widetilde{\psi}_{it}}=\ket{\psi_{it,M_{it}}}$; &
Normalized
\\

$\ket{\psi_{it}}$ &
Final post-relaxation latent mental state at trajectory sample $t$; &
Normalized
\\

$\ket{\psi_{i,t+1,0}}$ &
Beginning-of-sample state at trajectory sample $t+1$, obtained by
carrying forward $\ket{\psi_{it}}$; &
Normalized
\\

$\alpha_{k,it,m}$ &
Complex probability amplitude associated with general basis state
$\ket{S_k}$; &
$\mathbb{C}$
\\

$\alpha_{it,m}$ &
Complex probability amplitude associated with the neutral basis state
$\ket{N}$; &
$\mathbb{C}$
\\

$\beta_{it,m}$ &
Complex probability amplitude associated with the defensive basis state
$\ket{D}$; &
$\mathbb{C}$
\\

$\vartheta_{it,m}$ &
Polar angle of the latent mental state on the Bloch sphere; &
$[0,\pi]$
\\

$\varphi_{it,m}$ &
Azimuthal angle representing the relative phase between the neutral and
defensive amplitudes; &
$[0,2\pi)$
\\

$\mathbf{r}_{it,m}$ &
Bloch vector corresponding to $\ket{\psi_{it,m}}$,
with components
$(r_{x,it,m},r_{y,it,m},r_{z,it,m})$; &
$\|\mathbf r\|=1$
\\

$\mathbf{r}_{D}$ &
Bloch vector of the defensive pole,
$\mathbf{r}_{D}=(0,0,-1)$; &
Unit vector
\\

% ============================================================
\addlinespace
\multicolumn{3}{l}{\textit{Initialization, cue encoding, and state evolution}}
\\
\midrule

$\mathcal{U}_{\mathrm{init}}(\boldsymbol{\zeta}_i)$ &
Initialization operator used once at $t=0$ to construct the initial
state of driver $i$; &
Operator
\\

$\boldsymbol{\zeta}_i$ &
Parameters governing the initial state of driver $i$; &
Application-dependent
\\

$x_{it,m}$ &
Observed input associated with cue $c_{it,m}$; &
Cue-dependent
\\

$s_{it,m}$ &
Encoded representation of the observed cue input. This symbol is
distinct from the latent-state index $s$; &
Cue-dependent
\\

$f_c(\cdot)$ &
Cue-encoding function for cue type $c$; &
Function
\\

$\boldsymbol{\xi}_c$ &
Parameter vector governing the cue-encoding function $f_c(\cdot)$; &
Application-dependent
\\

$g_c(\cdot)$ &
Cue-specific function mapping the encoded cue to a rotation angle; &
Function
\\

$\boldsymbol{\rho}_c$ &
Parameter vector governing the rotation-response function
$g_c(\cdot)$; &
Application-dependent
\\

$\mathfrak{a}(c)$ &
Mapping assigning cue type $c$ to a rotation axis; &
$\{x,z\}$
\\

$\theta_{it,m}$ &
Rotation angle generated by cue $c_{it,m}$ at update $(i,t,m)$; &
Radians
\\

$\theta_{\max,c}$ &
Maximum rotation angle permitted for cue type $c$; &
Typically $[0,\pi]$
\\

$k_c$ &
Positive cue-specific sensitivity or saturation parameter. This symbol
is distinct from the general basis-state index $k$; &
$k_c>0$
\\

$\sigma_x,\sigma_y,\sigma_z$ &
Pauli matrices generating rotations around the $x$-, $y$-, and
$z$-axes of the Bloch sphere; &
$2\times2$
\\

$\boldsymbol{\sigma}$ &
Vector of Pauli matrices,
$\boldsymbol{\sigma}=(\sigma_x,\sigma_y,\sigma_z)$; &
--
\\

$\mathcal{U}_a(\theta)$ &
Unitary rotation operator through angle $\theta$ around axis $a$; &
Unitary
\\

$\mathcal{U}_{it,m}$ &
Cue-induced unitary operator associated with update $(i,t,m)$; &
Unitary
\\

$I_2$ &
Two-dimensional identity matrix; &
$2\times2$
\\

$(\cdot)^{\dagger}$ &
Conjugate-transpose operation; &
--
\\

$p^{D}_{it,m}$ &
Intermediate defensive state probability before cue $m$ is processed,
$p^{D}_{it,m}=|\braket{D}{\psi_{it,m}}|^2$; &
$[0,1]$
\\

$p^{D,*}_{it,m+1}$ &
Defensive state probability of the candidate state before model
control; &
$[0,1]$
\\

$\mathcal{M}_{it,m}$ &
Monotonicity acceptance operation used to assess a candidate
state-changing update; &
Operation
\\

$\mathcal{C}_{it,m}$ &
Model-control operation incorporating the monotonicity constraint and,
when needed, the geodesic safeguard; &
Operation
\\

$\phi_{it,m}$ &
Remaining angular distance between the current Bloch vector and the
defensive pole. This is distinct from the relative phase
$\varphi_{it,m}$; &
$[0,\pi]$
\\

$\widehat{\mathbf{n}}_{\mathrm{geo},it,m}$ &
Unit rotation axis defining the shortest Bloch-sphere path from the
current state toward the defensive pole; &
Unit vector
\\

$\mathcal{U}_{\mathrm{geo},it,m}(\theta)$ &
Geodesic rotation operator activated when the originally assigned Pauli
rotation would violate the monotonicity requirement; &
Unitary
\\

$\mathcal{R}_{\lambda}$ &
Normalized interpolation operation moving the pre-relaxation state
toward its initial baseline; &
Non-unitary
\\

$\lambda$ &
Relaxation weight governing movement toward the initial baseline after
all within-sample cues have been processed; &
$0\leq\lambda<1$
\\

$\mathcal{B}$ &
Born-rule probability-extraction operation; &
Operation
\\

% ============================================================
\addlinespace
\multicolumn{3}{l}{\textit{Latent-state and manoeuvre probabilities}}
\\
\midrule

$P_{ist}$ &
Probability that driver $i$ is in latent mental state $s$ at trajectory
sample $t$; &
$[0,1]$
\\

$P_{iNt}$ &
Final neutral-state membership probability at trajectory sample $t$,
obtained from the post-relaxation state; &
$[0,1]$
\\

$P_{iDt}$ &
Final defensive state membership probability at trajectory sample $t$,
obtained from the post-relaxation state; &
$[0,1]$
\\

$P_{ijt\mid s}$ &
Probability that driver $i$ selects manoeuvre $j$ at trajectory sample
$t$, conditional on latent mental state $s$; &
$[0,1]$
\\

$P_{ijt}$ &
Unconditional probability that driver $i$ selects manoeuvre $j$ at
trajectory sample $t$; &
$[0,1]$
\\

$V_{ijt\mid s}$ &
Systematic utility of manoeuvre $j$ for driver $i$ at trajectory sample
$t$, conditional on latent mental state $s$; &
$\mathbb{R}$
\\

$a_{it}$ &
Observed manoeuvre selected by driver $i$ at trajectory sample $t$.
The driver index may be suppressed and written as $a_t$ in the empirical
application; &
$\in\mathcal A$
\\

$y_{ijt}$ &
Choice indicator equal to one when driver $i$ selects manoeuvre $j$ at
trajectory sample $t$, and zero otherwise; &
$\{0,1\}$
\\

$\mathbf{1}[\cdot]$ &
Indicator function equal to one when the condition in brackets is true,
and zero otherwise; &
$\{0,1\}$
\\

% ============================================================
\addlinespace
\multicolumn{3}{l}{\textit{Likelihood and estimation}}
\\
\midrule

$\boldsymbol{q}$ &
Parameter vector governing the quantum state membership layer; &
$\mathbb{R}^{n_q}$
\\

$\boldsymbol{c}$ &
Parameter vector governing the classical state-conditional choice
layer. This bold symbol is distinct from cue type $c_{it,m}$; &
$\mathbb{R}^{n_c}$
\\

$\boldsymbol{\Theta}$ &
Complete model parameter vector,
$\boldsymbol{\Theta}=(\boldsymbol{q},\boldsymbol{c})$; &
$\mathbb{R}^{n_q+n_c}$
\\

$\mathcal{P}_D$ &
Set of defensive state choice coefficients subject to the
$\ell_2$ regularization penalty; &
Set
\\

$\eta_{\mathrm{reg}}$ &
Regularization weight applied to the defensive state choice coefficients
during penalized Q-SCM estimation; &
$0.1$
\\

$\mathcal{Q}(\boldsymbol{\Theta})$ &
Penalized Q-SCM estimation objective consisting of the negative
unpenalized log-likelihood and the $\ell_2$ regularization term; &
$\mathbb{R}$
\\

$\widehat{\boldsymbol{\Theta}}_{\mathrm{pen}}$ &
Q-SCM parameter vector obtained by minimizing the penalized objective; &
--
\\

$\widehat{\ell}_{\mathrm{eval}}$ &
Unpenalized log-likelihood evaluated at the penalized Q-SCM parameter
estimate; &
$\mathbb{R}$
\\

$\Omega$ &
Feasible parameter space used in likelihood based and penalized estimation; &
Set
\\

$\ell(\boldsymbol{\Theta})$ &
Sample log-likelihood evaluated at parameter vector
$\boldsymbol{\Theta}$; &
$\mathbb{R}$
\\

$\ell_i(\boldsymbol{\Theta})$ &
Log-likelihood contribution of driver $i$; &
$\mathbb{R}$
\\

$\widehat{\boldsymbol{\Theta}}$ &
Maximum-likelihood estimate of the complete parameter vector; &
--
\\

$\mathbf{g}_i$ &
Driver-level score vector evaluated at
$\widehat{\boldsymbol{\Theta}}$; &
Vector
\\

$\mathbf{H}$ &
Observed information matrix, defined as the negative Hessian of the
sample log-likelihood; &
Matrix
\\

$\widehat{\operatorname{Var}}
(\widehat{\boldsymbol{\Theta}})$ &
Estimated driver-clustered covariance matrix of the parameter
estimates; &
Matrix
\\

% ============================================================
\addlinespace
\multicolumn{3}{l}{\textit{Application-specific observed variables}}
\\
\midrule

$\mathrm{CTTC}_{it}$ &
Minimum closing time-to-collision observed for driver $i$ at trajectory
sample $t$. The driver index is suppressed in the application tables; &
Seconds
\\

$\mathrm{mDist}_{it}$ &
Separation distance from the ego vehicle to the nearest interacting
vehicle; &
Metres
\\

$s_{\mathrm{dist},it}$ &
Distance-based cue,
$s_{\mathrm{dist},it}=1/(1+\mathrm{mDist}_{it})$; &
$[0,1]$
\\

$s_{\mathrm{CTTC},it}$ &
CTTC-based cue,
$s_{\mathrm{CTTC},it}=1/(1+\mathrm{CTTC}_{it})$; &
$[0,1]$
\\

$s_{\mathrm{dev},it}$ &
Lane-deviation cue measuring departure from the ego vehicle's intended
path; &
Application scale
\\

$f_{it}$ &
Frontal proximity of the minimum-CTTC interacting vehicle; &
$[0,1]$
\\

$r_{it}$ &
Rear proximity of the minimum-CTTC interacting vehicle; &
$[0,1]$
\\

$o_{ijt}$ &
Non-overlap indicator for manoeuvre alternative $j$; &
$\{0,1\}$
\\

$d_{ijt}$ &
Distance from the centroid of manoeuvre alternative $j$ to the ego
vehicle's intended path; &
Metres
\\

$j=1$ &
Decelerate manoeuvre; &
--
\\

$j=2$ &
Maintain-speed manoeuvre and reference alternative within each latent
state; &
--
\\

$j=3$ &
Accelerate manoeuvre; &
--
\\

% ============================================================
\addlinespace
\multicolumn{3}{l}{\textit{Application-specific model parameters}}
\\
\midrule

$\beta^{N}_{\mathrm{decel}}$,
$\beta^{N}_{\mathrm{accel}}$ &
Alternative-specific constants for deceleration and acceleration in the
neutral state; &
$\mathbb{R}$
\\

$\beta^{D}_{\mathrm{decel}}$,
$\beta^{D}_{\mathrm{accel}}$ &
Alternative-specific constants for deceleration and acceleration in the
defensive state; &
$\mathbb{R}$
\\

$\beta^{N}_{\mathrm{dist\_path}}$ &
Effect of alternative-specific distance from the intended path in the
neutral state; &
$\mathbb{R}$
\\

$\beta^{N}_{\mathrm{nonoverlap}}$,
$\beta^{D}_{\mathrm{nonoverlap}}$ &
Effects of the alternative-specific non-overlap indicator in the neutral
and defensive states; &
$\mathbb{R}$
\\

$\beta^{N}_{\mathrm{int}}$,
$\beta^{D}_{\mathrm{int}}$ &
Interaction-intensity coefficients governing the frontal-threat
deceleration effect and rear-threat acceleration effect in each state; &
$\mathbb{R}$
\\

$\theta_{\max,\mathrm{dist}}$,
$\theta_{\max,\mathrm{CTTC}}$,
$\theta_{\max,\mathrm{dev}}$ &
Maximum rotation angles for the distance, CTTC, and lane-deviation cues,
respectively; &
Radians
\\

$k_{\mathrm{dist}}$,
$k_{\mathrm{CTTC}}$,
$k_{\mathrm{dev}}$ &
Cue-specific sensitivity or saturation parameters; &
Positive
\\

$\theta_0$ &
Population-level initial polar angle used to initialize the latent mental
state at $t=0$; &
$[0,\pi]$
\\

$\gamma_1$ &
CTTC coefficient in the classical latent class membership model; &
$\mathbb{R}$
\\

$\gamma_2$ &
Defensive-class constant in the classical latent class membership
model; &
$\mathbb{R}$
\\

$V^{\mathrm{cl}}_{D,it}$,
$V^{\mathrm{cl}}_{N,it}$ &
Defensive and neutral class membership utilities in the classical
benchmark; &
$\mathbb{R}$
\\

$P^{\mathrm{cl}}(D_{it})$,
$P^{\mathrm{cl}}(N_{it})$ &
Defensive and neutral membership probabilities generated by the
classical benchmark; &
$[0,1]$
\\

$P^{\mathrm{Q}}(D_{it})$,
$P^{\mathrm{Q}}(N_{it})$ &
Application-specific notation for the defensive and neutral membership
probabilities generated by the Q-SCM; &
$[0,1]$
\\

\end{longtable}

\normalsize

%  add to appendix
\section{Bloch Sphere Representation}
\label{sec:bloch}

The Bloch sphere provides a geometric visualization of the two-state
quantum system. The north pole corresponds to $\ket{N}$, representing the neutral state, and the south pole corresponds to $\ket{D}$, representing the defensive state, as illustrated in
Figure~\ref{fig:bloch_sphere}. Consistent with the notation introduced in Eq.~\eqref{eq:bloch_parameterization}, let $\vartheta$ denote the polar angle of the latent mental state and let $\varphi$ denote its relative phase. The polar angle determines the Defensive state probability through
\begin{equation}
P(D) = \sin^2\left(\frac{\vartheta}{2}\right),
\end{equation}
whereas the relative phase angle $\varphi$ contains contextual information that is not visible from the Defensive state probability
alone. The state polar angle $\vartheta$ is distinct from the cue-induced rotation angle $\theta_{it,m}$ defined in Eq.~\eqref{eq:general_rotation_mapping}.

Two latent mental states may have the same Defensive state probability but different relative phases because they were generated by different cue histories. These phase differences can influence the post relaxation state because relaxation acts on the complex amplitudes. 
They are also retained through state carry-forward and can influence
subsequent noncommuting state-changing updates.

Any pure two-state cognitive state can be written as:
\begin{equation}
\ket{\psi} = \alpha\ket{N} + \beta\ket{D}, \qquad |\alpha|^2 + |\beta|^2 = 1,
\label{eq:superposition}
\end{equation}
where $\alpha$ and $\beta$ are complex probability amplitudes associated with the neutral and defensive states, respectively. $|\cdot|$ is the modulus of a complex number. 
Using the standard Bloch sphere parameterization, the same state can be written as:
\begin{equation}
\ket{\psi} = \cos\left(\frac{\vartheta}{2}\right)\ket{N} +
e^{i\varphi} \sin\left(\frac{\vartheta}{2}\right)\ket{D},
\label{eq:bloch_state_parameterization}
\end{equation}
where $\vartheta\in[0,\pi]$ is the polar angle measured from the neutral pole and $\varphi\in[0,2\pi)$ is the azimuthal angle representing the relative phase between the neutral and defensive amplitudes. 
Equivalently,
\begin{equation}
\alpha = \cos\left(\frac{\vartheta}{2}\right),
\qquad 
\beta = e^{i\varphi} \sin\left(\frac{\vartheta}{2}\right).
\label{eq:alpha_beta_bloch}
\end{equation}

The squared moduli of the amplitudes give the probabilities of the two cognitive states:
\begin{equation}
P(N) = |\alpha|^2 = \cos^2\left(\frac{\vartheta}{2}\right),
\qquad
P(D) = |\beta|^2 = \sin^2\left(\frac{\vartheta}{2}\right) = \frac{1-\cos(\vartheta)}{2}.
\label{eq:state_probabilities_alpha_beta}
\end{equation}
Hence, the polar angle $\vartheta$ determines the driver's position on the neutral--defensive spectrum. For instance, a state with
$|\beta|^2=0.3$ assigns a probability of $0.3$ to the defensive state and a probability of $0.7$ to the neutral state.
The relative phase $\varphi$ does not alter these probabilities when the state is considered in isolation. In the Q-SCM, however, it can affect the post relaxation state because relaxation acts on the complex amplitudes, and it can influence subsequent noncommuting cue-induced
rotations through state carry-forward.

The state vector in Eq.~\eqref{eq:superposition} is a pure state, representing a fully specified superposition that lies on the surface of the Bloch sphere. 
The present model operates entirely using pure states. 
In contrast, a mixed state represents classical uncertainty about which pure state the system occupies. 
For example, if there is a 60\% chance that the driver is in pure state $\ket{\psi_1}$ and a 40\% chance that the driver is in pure state $\ket{\psi_2}$, the mixture is described by the density matrix
\begin{equation}
\rho = 0.6\ket{\psi_1}\!\bra{\psi_1} + 0.4\ket{\psi_2}\!\bra{\psi_2}.
\end{equation}
On the Bloch sphere, a pure state lies on the surface with $\|\mathbf{r}\|=1$, whereas a mixed state lies inside the sphere with $\|\mathbf{r}\|<1$~\cite{nielsen2010quantum}. 

The use of a pure state is a parsimonious modelling assumption. Each driver trajectory is represented as a single recursively evolving state initialized once and updated deterministically conditional on the observed cues and model parameters. The present formulation does not separately represent stochastic mixtures of latent cognitive configurations, attentional noise, or loss of phase information within a trajectory. Such effects would require a density matrix formulation or another stochastic state evolution mechanism.

The state can also be mapped to a Bloch vector $\mathbf{r}=(r_x,r_y,r_z)$ by taking the expectation value of each Pauli matrix with respect to the current state:
\begin{equation}
r_x = \bra{\psi}\sigma_x\ket{\psi}, \qquad
r_y = \bra{\psi}\sigma_y\ket{\psi}, \qquad
r_z = \bra{\psi}\sigma_z\ket{\psi},
\label{eq:bloch_vector}
\end{equation}
where $\boldsymbol{\sigma}=(\sigma_x,\sigma_y,\sigma_z)$ are the Pauli matrices. 
Each component measures the alignment of the state with the corresponding axis of the Bloch sphere. 
The $r_z$ component determines the neutral and defensive probabilities through the Born rule:
\begin{equation}
P(N)=\frac{1+r_z}{2}, \qquad P(D)=\frac{1-r_z}{2}.
\label{eq:born_rule_bloch}
\end{equation}
The remaining components, $r_x$ and $r_y$, encode the phase information that is not visible from $P(D)$ alone but influences later state evolution.

\section{Unitary Rotations on the Bloch Sphere}
\label{app:rotations}

Unitary transformations on a two-state system are rotations on the
Bloch sphere. The fundamental generators of all such rotations are the
three Pauli matrices:
\begin{equation}
\sigma_x = \begin{pmatrix} 0 & 1 \\ 1 & 0 \end{pmatrix}, \quad
\sigma_y = \begin{pmatrix} 0 & -i \\ i & 0 \end{pmatrix}, \quad
\sigma_z = \begin{pmatrix} 1 & 0 \\ 0 & -1 \end{pmatrix}.
\label{eq:pauli}
\end{equation}
Any rotation on the Bloch sphere can be built by combining $\sigma_x$,
$\sigma_y$, and $\sigma_z$ in appropriate proportions~\cite{nielsen2010quantum}.
In the most general case, a rotation around an arbitrary axis
$\hat{\mathbf{n}}=(n_x,n_y,n_z)$ with $\|\hat{\mathbf{n}}\|=1$ is
\begin{equation}
\mathcal{U}_{\hat{\mathbf{n}}}(\theta) = \exp\!\left(-i\frac{\theta}{2}\,\hat{\mathbf{n}}\cdot\boldsymbol{\sigma}\right)
= \cos\!\left(\frac{\theta}{2}\right)I
  - i\sin\!\left(\frac{\theta}{2}\right)\!\left(n_x\sigma_x+n_y\sigma_y+n_z\sigma_z\right),
\label{eq:U_general}
\end{equation}
which reduces, when the axis is aligned with a coordinate direction, to
the single axis rotation
\begin{equation}
\mathcal{U}_a(\theta) = \exp\!\left(-i\frac{\theta}{2}\sigma_a\right)
= \cos\!\left(\frac{\theta}{2}\right)I - i\sin\!\left(\frac{\theta}{2}\right)\sigma_a,
\qquad a\in\{x,y,z\}.
\label{eq:Ua_def}
\end{equation}
The $\sigma_x$ and $\sigma_y$ components mix the amplitudes $\alpha$
and $\beta$ between $\ket{N}$ and $\ket{D}$, directly altering
$P(D)=|\beta|^2$. The $\sigma_z$ component
$\mathcal{U}z(\theta)=\mathrm{diag}(e^{-i\theta/2},e^{i\theta/2})$ multiplies
each amplitude by a phase factor and leaves $|\alpha|^2$ and $|\beta|^2$
unchanged. Starting from the neutral pole $\ket{N}$ (Figure~\ref{fig:pauli_rotations}):
\begin{itemize}
\item $\sigma_x$ rotation (state change, imaginary amplitude):
  $\mathcal{U}x(\theta)\ket{N}=\cos(\theta/2)\ket{N}-i\sin(\theta/2)\ket{D}$,
  giving $P(D)=\sin^2(\theta/2)$.
\item $\sigma_z$ rotation (phase change only):
  $\mathcal{U}z(\theta)\ket{N}=e^{-i\theta/2}\ket{N}$, giving $P(D)=0$.
\item $\sigma_y$ rotation (state change, real amplitude):
  $\mathcal{U}y(\theta)\ket{N}=\cos(\theta/2)\ket{N}+\sin(\theta/2)\ket{D}$,
  giving $P(D)=\sin^2(\theta/2)$ as well, with a purely real $\ket{D}$
  amplitude.
\end{itemize}

\begin{figure}[H]
\centering
\begin{subfigure}[b]{0.32\textwidth}
    \centering
    \includegraphics[width=\textwidth,trim=0 0 0 1.7cm,clip]{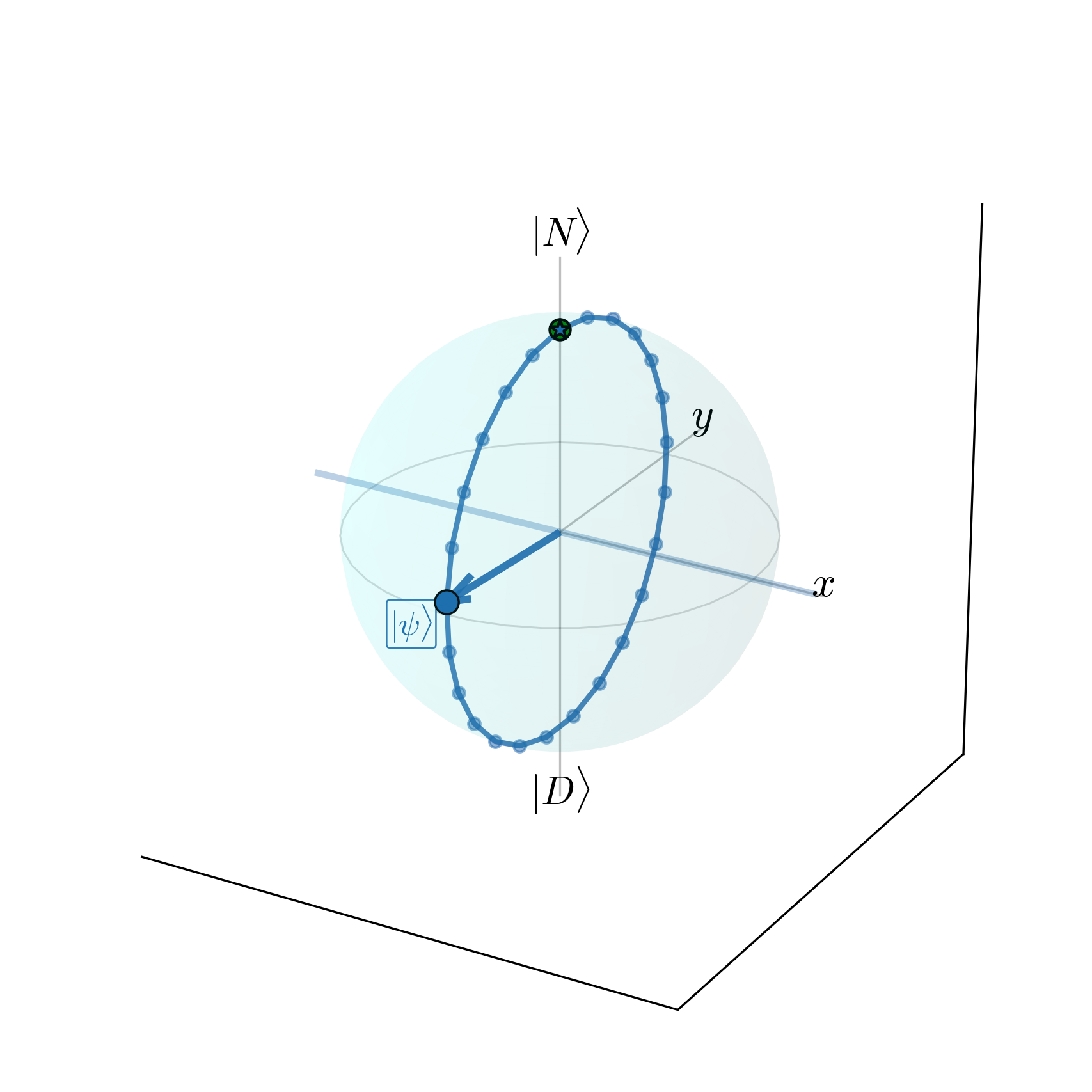}
    \caption{$\sigma_x$ rotation}\label{fig:pauli_rot_x}
\end{subfigure}\hfill
\begin{subfigure}[b]{0.32\textwidth}
    \centering
    \includegraphics[width=\textwidth,trim=0 0 0 1.7cm,clip]{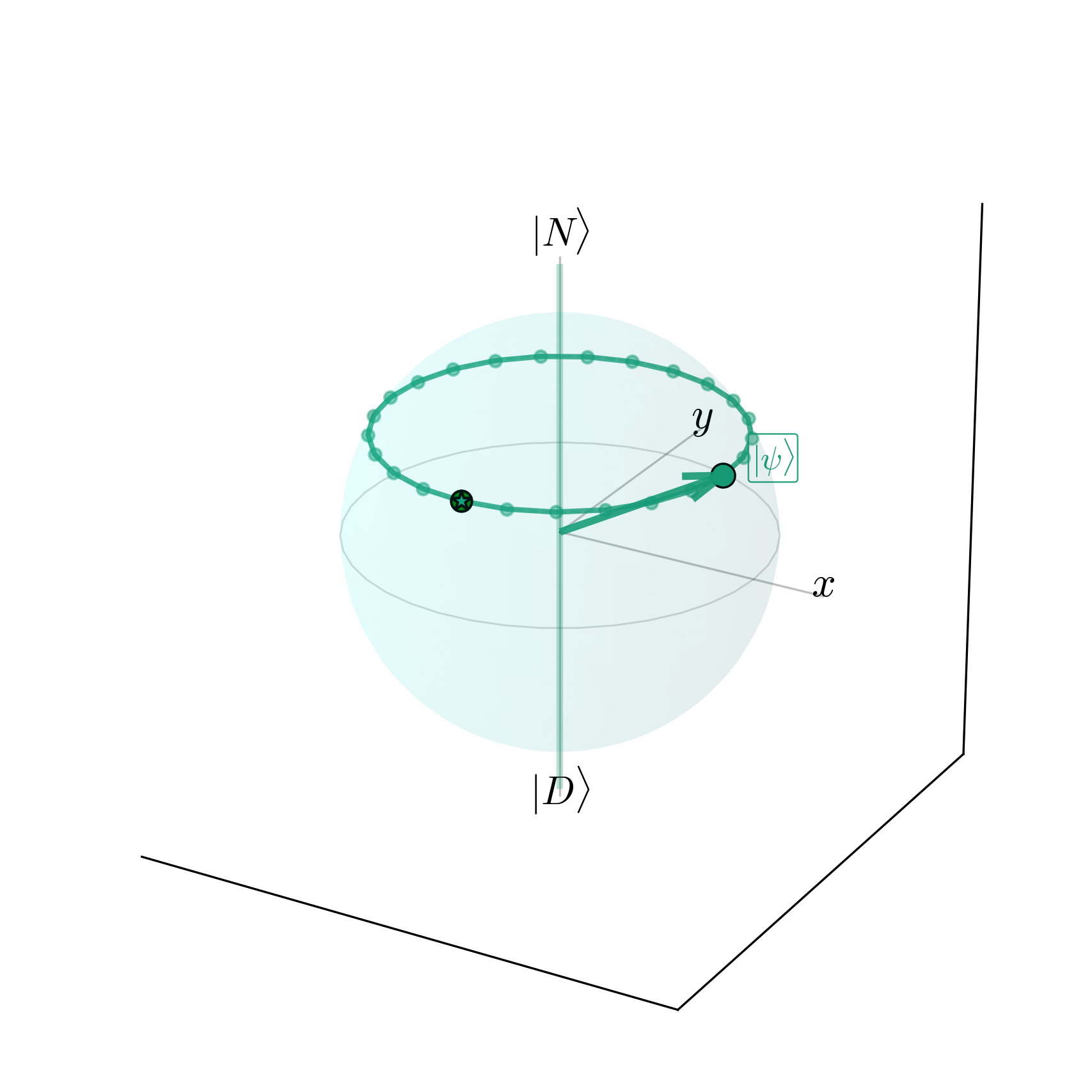}
    \caption{$\sigma_z$ rotation}\label{fig:pauli_rot_z}
\end{subfigure}\hfill
\begin{subfigure}[b]{0.32\textwidth}
    \centering
    \includegraphics[width=\textwidth,trim=0 0 0 1.7cm,clip]{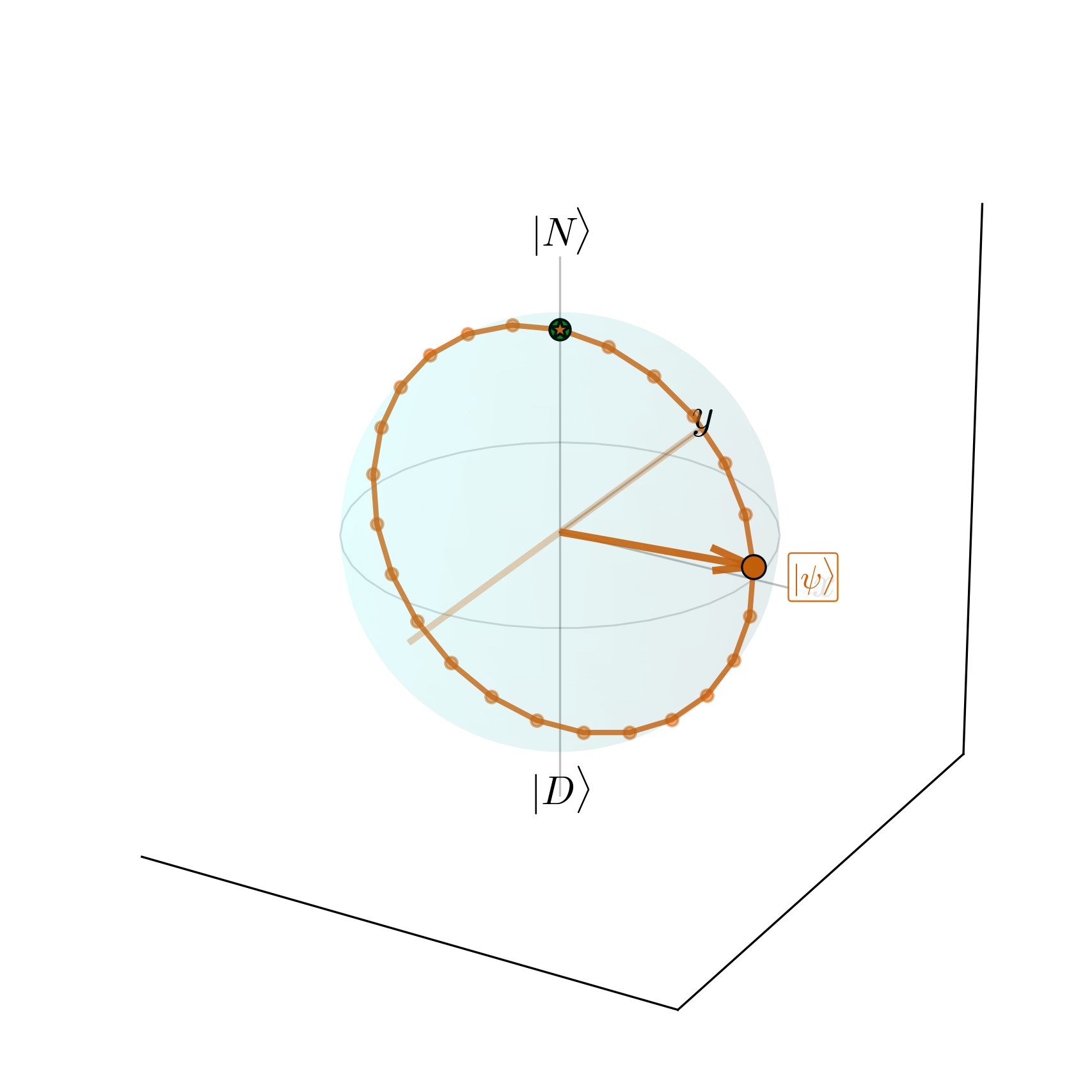}
    \caption{$\sigma_y$ rotation}\label{fig:pauli_rot_y}
\end{subfigure}
\caption{Rotation on the Bloch sphere for each Pauli axis.}
\label{fig:pauli_rotations}
\end{figure}

In Q-SCM, each cue type is assigned to exactly one of two Pauli axes.
The $\sigma_x$ axis is reserved for cues that directly change the
defensive probability (state change) and the $\sigma_z$ axis for cues
that modulate the internal phase without directly altering $P(D)$
(phase change). No cue is assigned to $\sigma_y$, and no cue induces a
rotation around a general off-axis direction. This assumption is
motivated by three considerations. First, the two retained axes capture
the two fundamental and distinct ways a cue can influence cognition.
A $\sigma_z$ cue can be understood as a priming or context-setting
signal which does not, by itself, make the driver more or less
defensive but conditions how subsequent state-changing cues are
processed. Second, confining each cue to a single axis yields a
parsimonious model with clear cognitive interpretation, and the model
parameters (axis assignment and rotation magnitude) are directly
identifiable. Third, the $\sigma_y$ axis simultaneously produces both
state and phase change, coupling the two mechanisms in a way that
increases model complexity without a clear cognitive justification at
this stage. When two cues on different axes ($\sigma_x$ and $\sigma_z$)
are composed sequentially, the resulting product rotation generically
contains a $\sigma_y$ component, a direct consequence of the
commutation relation $[\sigma_x,\sigma_z]=-2i\sigma_y$. Therefore, combined state and phase
effects emerge naturally from the sequential composition of cue
updates rather than being assigned to individual cues. We adopt
$\sigma_x$ as the state changing axis and $\sigma_z$ as the phase
changing axis; the specific assignment of driving cues to these axes is
detailed in Section~\ref{sec:state_evolution}.

\section{Formal Guarantees}
\label{app:formal_guarantees} 
\subsection{Proof of Proposition~\ref{prop:cue_update_monotonicity}}
% \subsection{Cue Update Monotonicity}

% \begin{proposition}
% \label{prop:cue_update_monotonicity}
For every controlled cue update performed before relaxation, the Defensive state probability is non-decreasing.
% \begin{equation}
% p^{D}_{it,m+1} \geq p^{D}_{it,m}, \qquad m=0,\ldots,M_{it}-1.
% \label{eq:cue_update_monotonicity}
% \end{equation}
Consequently, the sequence of Defensive state probabilities generated during the cue-driven evolution within a trajectory sample cannot exhibit an increase followed by a decrease caused by rotational overshoot.

\begin{proof}
Consider the latent state immediately before cue $m$ is processed:
\[ \ket{\psi_{it,m}}, \qquad p^{D}_{it,m} = 
\left| \braket{D}{\psi_{it,m}} \right|^2. \]
For a state-changing cue assigned to the $x$-axis, the originally assigned
rotation produces the candidate state:
\[ \ket{\psi^{*}_{it,m+1}} = \mathcal{U}_x(\theta_{it,m})\ket{\psi_{it,m}}, \]
with candidate Defensive state probability 
\[ p^{D,*}_{it,m+1} = \left| \braket{D}{\psi^{*}_{it,m+1}}
\right|^2. \]
If $p^{D,*}_{it,m+1}\geq p^{D}_{it,m}$,
the candidate state is accepted directly. Then:
\[ p^{D}_{it,m+1} = p^{D,*}_{it,m+1} \geq
p^{D}_{it,m}. \]
If $p^{D,*}_{it,m+1}<p^{D}_{it,m}$, the candidate is not accepted in its original form. The geodesic safeguard redirects the update toward the Defensive pole without passing beyond it. By the no-overshoot property of the safeguard, the redirected state satisfies:
\[ p^{D}_{it,m+1} \geq p^{D}_{it,m}. \]
Thus, every controlled state-changing cue update satisfies:
\[ p^{D}_{it,m+1} \geq p^{D}_{it,m}. \]
For a phase-changing cue assigned to the $z$-axis, the rotation changes only the relative phase and therefore leaves the pre-relaxation Defensive state probability unchanged:
\[ p^{D}_{it,m+1} = p^{D}_{it,m}. \]
Applying these results successively to all within-sample cue updates produces a non-decreasing sequence of pre-relaxation Defensive state probabilities. Such a sequence cannot increase and then decrease because of rotational overshoot.
The result applies only to the cue-driven state evolution before relaxation. The subsequent relaxation operation may reduce the Defensive state probability by moving the state toward the initial baseline.
\end{proof}
\subsection{Proof of Proposition~\ref{prop:geodesic_safeguard}}

\begin{proof}
Consider an $x$-axis candidate update that would reduce the Defensive state probability and activates the geodesic safeguard in Eq.~\eqref{eq:geodesic_control_operation}. Let $\phi_{it,m}\in(0,\pi]$ be the angular distance between the current Bloch vector and the Defensive pole. The safeguard applies a rotation toward the Defensive pole with angle:
\[\delta_{it,m}=\min\left(\theta_{it,m},\phi_{it,m}\right).\]
Because:
\[0\leq\delta_{it,m}\leq\phi_{it,m},\]
the angular distance to the Defensive pole after the redirected update is:
\[\phi'_{it,m}=\phi_{it,m}-\delta_{it,m},\]
satisfies:
\[0\leq\phi'_{it,m}\leq\phi_{it,m}.\]
Thus, the redirected state moves along the geodesic toward the Defensive pole and cannot pass beyond it.
The Defensive state probability can be expressed in terms of the angular distance from the Defensive pole by:
\[p^{D}_{it,m}=\frac{1+\cos\left(\phi_{it,m}\right)}{2}.\]
Since $\cos(\phi)$ is non-increasing on $[0,\pi]$ and $\phi'_{it,m}\leq\phi_{it,m}$, the redirected state satisfies
\[p^{D}_{it,m+1}=\frac{1+\cos\left(\phi'_{it,m}\right)}{2}\geq\frac{1+\cos\left(\phi_{it,m}\right)}{2}=p^{D}_{it,m}.\]
So the geodesic safeguard cannot reduce the Defensive state probability and cannot overshoot the Defensive pole. When $\delta_{it,m}>0$, the inequality is strict.
When the state is already at the Defensive pole, $\phi_{it,m}=0$ and the state remains unchanged.
% This result is only for an individual cue-driven update for which the geodesic safeguard is activated. It does not imply that every nonzero cue produces strict progress, because an originally assigned candidate that preserves the Defensive state probability is accepted directly. It also does not imply convergence to a Defensive state probability of one during the complete trajectory, because the subsequent relaxation operation may move the state toward its baseline.
\end{proof}

\setlength{\bibsep}{0pt plus 0.3ex}
\bibliography{biblio1}
\end{document}